\documentclass[11pt]{article}

\usepackage[left=1in,right=1in,top=1in,bottom=1in]{geometry}
\usepackage{fullpage}

\usepackage{amsthm}
\usepackage{amsmath}
\usepackage{subcaption}
\usepackage{forest}
\usepackage{amssymb}
\usepackage{mathtools}

\usepackage[noend]{algpseudocode}
\usepackage{algorithm}

\usepackage{graphicx}

\usepackage{tikz}
\usetikzlibrary{automata,arrows.meta,positioning}
\usepackage{calc}

\usepackage{subcaption}

\usepackage{color}

\usepackage{wrapfig}

\usepackage{xspace}

\usepackage[bookmarks]{hyperref}

\usepackage{varioref} %automatically inserts page reference, if not on the same page

\usepackage{footnote} %for using footnotes within tables

\usepackage{enumitem} %labeling items in enumerations

\usepackage[textsize=scriptsize]{todonotes}

\newcommand{\cpat}{\ensuremath{\mathcal P}\xspace}
\newcommand{\ctext}{\ensuremath{\mathcal T}\xspace}

\mathchardef\mhyphen="2D
\newcommand{\concat}{\mathop{\circ}}
\newcommand{\bigconcat}{\mathop{\bigcirc}}

\newcommand{\x}{{\textsc{x}}}
\newcommand{\y}{{\textsc{y}}}
\newcommand{\strc}{{\cal S}}

\newcommand{\eps}{\ensuremath{\varepsilon}}

\newcommand{\rev}{\mathrm{rev}}

\newcommand{\sA}{{\cal A}}
\newcommand{\sB}{{\cal B}}

\newcommand{\pad}{\mathrm{pad}}

\newcommand{\hamming}{\text{Hamming}}
\newcommand{\ham}{\text{Ham}}
\newcommand{\subseq}{\text{Subseq}}

\newcommand{\remark}[3]{\textcolor{blue}{\textsc{#1 #2:}}
  \textcolor{red}{\textsf{#3}}}
\newcommand{\karl}[2][says]{\remark{Karl}{#1}{#2}}

\newcommand{\claref}[1]{Claim~\ref{cla:#1}}

\newcommand{\figref}[1]{Figure~\ref{fig:#1}}

\newcommand{\defref}[1]{Definition~\ref{def:#1}}

\newcommand{\thmref}[1]{Theorem~\ref{thm:#1}}

\newcommand{\lemref}[1]{Lemma~\ref{lem:#1}}
\newcommand{\lemrefs}[2]{Lemmas~\ref{lem:#1} and~\ref{lem:#2}}

\newcommand{\obsref}[1]{Observation~\ref{obs:#1}}
\newcommand{\secref}[1]{Section~\ref{sec:#1}}

\newtheorem{thm}{Theorem}[section]
\newtheorem{lem}[thm]{Lemma}

\newtheorem{obs}[thm]{Observation}

\newtheorem{defn}[thm]{Definition}

\newtheorem{claim}[thm]{Claim}
\newtheorem{problem}[thm]{Problem}

\newtheorem{conjecture}[thm]{Conjecture}

\global\long\def\Oh{{O}}
\global\long\def\tOh{\tilde{\Oh}}

\newcommand{\N}{\mathbb{N}}

\newcommand{\poly}{\textup{poly}}

\newcommand{\LCS}{\textup{LCS}\xspace}
\newcommand{\eval}{\textup{eval}}

\newcommand{\depth}{\textup{depth}}
\newcommand{\cT}{\mathcal{T}}
\newcommand{\cP}{\mathcal{P}}

\newcommand{\OV}{\textsc{OV}\xspace}

\newcommand{\SETH}{SETH\xspace}

\newcommand{\kOV}{$k$-\OV}

\newcommand{\DFAAccept}{DFA Acceptance\xspace}
\newcommand{\NFAAccept}{NFA Acceptance}
\newcommand{\CFGrecognition}{CFG Recognition}
\newcommand{\RNAfolding}{RNA Folding}
\newcommand{\weightedRNAfolding}{Weighted RNA Folding}

\newcommand{\transition}[3]{#1 \stackrel{#2}{\to} #3}

\newcommand{\productionsign}{\to}

\DeclarePairedDelimiter{\ceil}{\lceil}{\rceil}

\begin{document}

\title{Fine-Grained Complexity of Analyzing Compressed Data: \\ Quantifying Improvements over Decompress-And-Solve}

\author{
Amir Abboud\thanks{IBM Almaden Research Center, \texttt{abboud@cs.stanford.edu}.  Work done while at Stanford University.} \and
Arturs Backurs\thanks{MIT, \texttt{backurs@mit.edu}} \and
Karl Bringmann\thanks{Max Planck Institute for Informatics, Saarland Informatics Campus, Germany, \texttt{kbringma@mpi-inf.mpg.de}} \and
Marvin K\"unnemann\thanks{Max Planck Institute for Informatics, Saarland Informatics Campus, Germany, \texttt{marvin@mpi-inf.mpg.de}}
}

\maketitle

\medskip

\begin{abstract}
% !TEX root = main.tex

Can we analyze data without decompressing it?  
As our data keeps growing, understanding the time complexity of problems on \emph{compressed} inputs, rather than in convenient uncompressed forms, becomes more and more relevant.
Suppose we are given a compression of size~$n$ of data that originally has size~$N$, and we want to solve a problem with time complexity~$T(\cdot)$.
The na\"{i}ve strategy of ``decompress-and-solve" gives time~$T(N)$, whereas ``the gold standard" is time~$T(n)$: to analyze the compression as efficiently as if the original data was small.

We restrict our attention to data in the form of a string (text, files, genomes, etc.) and study the most ubiquitous tasks.
While the challenge might seem to depend heavily on the specific compression scheme, most methods of practical relevance (Lempel-Ziv-family, dictionary methods, and others) can be unified under the elegant notion of \emph{Grammar-Compressions}.
A vast literature, across many disciplines, established this as an influential notion for Algorithm design.

We introduce a framework for proving (conditional) lower bounds in this field,
allowing us to assess whether decompress-and-solve can be improved, and by how much.
Our main results are:
\begin{itemize}
\item The $O(nN\sqrt{\log{N/n}})$ bound for LCS and the $O(\min\{N \log N, nM\})$ bound for Pattern Matching with Wildcards are optimal up to $N^{o(1)}$ factors, under the Strong Exponential Time Hypothesis. (Here, $M$ denotes the uncompressed length of the compressed pattern.) 
\item Decompress-and-solve is essentially optimal for Context-Free Grammar Parsing and RNA Folding, under the $k$-Clique conjecture.
\item We give an algorithm showing that decompress-and-solve is \emph{not} optimal for Disjointness. 
\end{itemize}

\end{abstract}

%\iffalse
\thispagestyle{empty}
\clearpage
\setcounter{page}{1}

\tableofcontents
\thispagestyle{empty}
\clearpage
\setcounter{page}{1}
%\fi

% !TEX root = main.tex

\section{Introduction} \label{sec:intro}

Computer Science is often called the science of processing digital data. 
A central goal of theoretical CS is to understand the time complexity of the tasks we want to perform on data.
\emph{Data compression} has been one of the most important notions in CS and Information Theory for decades, and it is increasingly relevant in our current age of ``Big Data" where it is hard to think of reasons why \emph{not} to compress our data: smaller data can be stored more efficiently, transmitting it takes less resources such as energy and bandwidth, and perhaps it can even be processed faster.
Since nowadays and for years to come nearly all of our data comes in compressed form, a central question becomes: 

\begin{center}
\emph{What is the time complexity of analyzing compressed data?}
\end{center}

Say we have a piece of data of size $N$ given in a compressed form of size $n$.
For a problem with time complexity $T(\cdot)$, the na\"{i}ve strategy of ``decompress and solve" takes $\Theta(T(N))$ time, while the ``gold standard" is $O(T(n))$ time: we want to  solve the problem on the compression as efficiently as if the original data was small.
To provide meaningful statements we need to decide on three things: What type of \emph{data} is it? What \emph{problem} do we want to solve? Which \emph{compression scheme} is being used?

For the first two questions, the focus of this paper will be on the most basic setting. 
We consider data that comes as strings, i.e. sequences of symbols such as text, computer code, genomes, and so on.
And we study natural and basic questions one could ask about strings such as Pattern Matching, Language Membership, Longest Common Subsequence, Parsing, and Disjointness.

For the third question, we restrict our attention to \emph{lossless} compression and, even then, there are multiple natural settings that we do not find to be the most relevant. 
We could consider Kolmogorov complexity, giving us the best possible compression of our data: assume that a string $T$ is given by a short bitstring $K(T)$ which is a pair of Turing machine $M$ and input $x$ such that running $M$ on $x$ outputs $T$, i.e. $K(T)=\langle M,x \rangle$ such that $M(x)=T$.
The issue with Kolmogorov-compressions is that none of our data comes in this form, for two good reasons: First, it is computationally intractable to compute $K(T)$ given $T$, not even approximately. And second, analyzing arbitrary Turing machines without just running them is an infamously hopeless task.
Thus, while studying the time complexity of analyzing Kolmogorov-compressed strings is natural, it might not be the most relevant for computer science applications.
Another option is to consider the mathematically simplest forms of compression such as \emph{Run-Length Encoding} (RLE): we compress $x$ consecutive letters $\sigma$ into $\sigma^{x}$, so the compression has the form $0^{x_1}1^{x_2}0^{x_3}\cdots1^{x_{\ell}}$, and we only need $n=O(\ell \cdot \log{N})$ bits to describe the potentially exponentially longer string of length $N$.
This compression is at the other extreme of the spectrum: it is trivial to compute and easy to analyze, but it is far less ``compressing" than popular schemes like Lempel-Ziv-compressions. 

Instead, we consider what has proven to be one of the most influential kinds of compression for Algorithm design, namely \emph{Grammar-Compressions}, a notion that has all the right properties.
First, it is mathematically elegant and quite fun to reason about for theoreticians (as evidenced by the many pages of our paper).
Second, it is equivalent \cite{Rytter03} up to low order terms (moderate constants and log factors) to popular schemes like the Lempel-Ziv-family (LZ77, LZ78, LZW, etc.) \cite{LZ76,LZ77,W84}, Byte-Pair Encoding \cite{BytePair}, dictionary methods, and others \cite{NW97,Liu+08}. 
These compressions are used in ubiquitous applications such as the built-in Unix utility {\tt compress}, zip, GIF, PNG, and even in PDF.
Third, it is generic and likely to capture compression schemes that will be engineered in the future (after all, there is a whole industry on the topic and the quest might never be over). 
%Any effort that goes into studying it is likely to remain relevant for decades to come.
Fourth, we can compute the optimal such compression (up to log factors) in linear time \cite{Rytter03,Char+05,Jez16}.
And last but not least, ingenious algorithmic techniques have shown that it is possible to computationally \emph{analyze} grammar-compressed data, beating the ``decompress and solve" bound for many important problems.

A grammar compression of a string $X$ is simply a context-free grammar, whose language is exactly $\{X\}$, that is, the only string the grammar can produce is $X$.
For the purposes of this paper, it is enough to focus on a restricted form of grammars, known as \emph{Straight Line Programs} (SLP). 
An SLP is defined over some alphabet $\Sigma$, say $\{0,1\}$, and it is a set of replacement rules (or productions) of a very simple form: a rule is either a symbol in $\Sigma$ or it is the concatenation of two previous rules  (under some fixed ordering of the rules). The last replacement rule is the sequence defined by the SLP.
For example, we can compress the sequence $01011$ with the rules $S_1\to 0; \ \ S_2\to 1; \ \ S_3\to S_1\,S_2; \ \ S_4\to S_3\,S_3; \ \ S_5\to S_4\,S_2 \ $ and $S_5$ corresponds to the sequence $01011$. For some strings this can give an exponential compression.
A more formal definition and a figure are given in Section~\ref{sec:prelim}.

To learn more about the remarkable success of grammar-compressions, we refer the reader to the surveys \cite{WMB99book,Lar99book,GKPR96survey,SB06survey,GSU09survey,RB10survey,Rytt04survey,Lohrey12,Sak14survey}.
As a side remark, one of the exciting developments in this context was the surprising observation that a ``compress and solve" strategy could actually lead to theoretically new algorithms for some problems, e.g. \cite{WW98,Jez16jacm}.

\medskip
Thus, we focus on what we find the most important interpretation of the central question above:
  \begin{center}
\emph{What is the time complexity of basic problems on grammar-compressed strings?}
\end{center}
 
 \subsection{Previous Work}

As a motivating example, consider the Longest Common Subsequence (LCS) problem.
Given two uncompressed strings of length $N$ we can find the length of the longest common (not necessarily contiguous) subsequence in $O(N^2)$ time using dynamic programming, and there are almost-matching $N^{2-o(1)}$ conditional lower bounds \cite{ABV15,BK15,AHVW16}.
Throughout the paper we mostly ignore log factors, and so we think of LCS as a problem with $\tilde{\Theta}(N^2)$ time complexity (on uncompressed data).
Now, assume our sequences are given in compressed form of size $n$.
A natural setting to keep in mind is where $n \approx N^{1/2}$.
How much time do we need to solve LCS on these compressed strings?
The na\"{i}ve upper bound gives $O(N^2)$ and the gold standard is $O(n^{2}) \approx O(N)$, so which is it?

Besides being a very basic question, LCS and the closely related Edit Distance are a popular theoretical modeling of sequence alignment problems that are of great importance in Bioinformatics\footnote{The {\em heuristic} algorithm BLAST for a generalized version of the  problem has received sixty-thousand citations.}. 
Thus, this is a relatively faithful modeling of the question whether ``compress-and-solve" can speed up genome analysis tasks, a question which has received extensive attention throughout the years \cite{GT94,NW97,Liu+08,GT93,GSU09survey}.

A long line of work \cite{bunke1995improved,manber1997text,apostolico1997matching,arbell2002edit,crochemore2003subquadratic,tiskin2009faster,tiskin2010fast,HLLW13} has shown that we can do \emph{much} better than $O(N^2)$.
The current best algorithm has the curious runtime $O(nN\sqrt{\log{N/n}})$ \cite{gawrychowski2012faster} which is tantalizingly close to a conjectured bound of $O(nN)$ from the seminal paper of Lifshits \cite{lifshits2007processing}.
In our candidate setting of $n \approx N^{1/2}$, this is $\tilde{O}(N^{1.5})$.
This is major speedup over the $\Omega(N^2)$ decompress-and-solve bound, but is still far away from the gold standard of $O(n^2)$ which in this case would be $O(N)$.
Can we do better? For example, an $O(n^2\cdot N^{0.1})$ bound could lead to major real-world improvements.

While there is a huge literature on the topic, both from the Algorithms community and from applied areas, in addition to the potential for real-world impact, studying these questions has not become a mainstream topic in the top algorithms conferences.
In one of the only STOC/FOCS papers on the topic, Charikar et al.~\cite{Char+05} write 
``\emph{In short, the smallest grammar problem has been considered
by many authors in many disciplines for many reasons over a
span of decades. Given this level of interest, it is remarkable
that the problem has not attracted greater attention in the
general algorithms community.}"

We believe that one key reason for this is  the lack of a relevant \emph{complexity theory} and tools for proving \emph{lower bounds}, leaving a confusing state of the art in which it is hard to distinguish algorithms providing fundamental new insights from \emph{ad hoc} solutions.
Most importantly, previous work has not given us the tools to know, when we encounter a data analysis problem in the real-world, what kind of upper bound we should expect.
Instead, researchers have been proving P vs. NP-hard results, classifying problems into ones solvable in $\poly(n,\log{N})$ time and ones that probably require time $N^{\Omega(1)}$.
In fact, even LCS is NP-hard \cite{lifshits2007processing}.
This means that even if we have a compression of very small size $n=O(\log{N})$ then we cannot solve LCS in $\poly(n)$ time, unless $\textup{P}=\textup{NP}$.
Dozens of such negative results have been proven (see \cite{Lohrey12}), and it has long been clear that almost any task of interest is ``NP-hard", including the basic $\textup{poly}(N)$ time solvable problems we discuss in this paper.
However, this is hardly relevant to the questions we ask in this paper since it does not address the possibility of highly desirable bounds such as  $n^2\cdot N^{0.1}$. What we would really like to know is whether the bound should be $\poly(n)\cdot N^{\eps}$, or $\poly(n) \cdot N$, or even higher: could it be that decompress-and-solve is impossible to beat for some problems?

\subsection{Our Work}
 
 In this work, we introduce a framework for showing lower bounds on the time complexity of problems on grammar-compressed strings.
 Our lower bounds are based on popular conjectures from \emph{Hardness in P} and Fine-Grained Complexity.
This is perhaps surprising since the problems we consider are technically NP-hard.
 Our new complexity theoretic study of this field leads to three exciting developments:
 First, we resolve the exact time complexity up to $N^{o(1)}$ factors of some of the most classical problems such as LCS \emph{on compressed data}.
 Second, we discover problems that \emph{cannot be solved faster than the decompress-and-solve bound} by any $N^{\eps}$ factor.
 Third, we \emph{fail} at proving tight lower bounds for some classical problems, which hints to us that known algorithms might be suboptimal.
Indeed, in this paper we also find \emph{new algorithms} for fundamental problems.
We hope that our work will inspire increased interest in this important topic.

\paragraph{Longest Common Subsequence}
Our first result is a resolution of the time complexity of LCS on compressed data, up to $N^{o(1)}$ factors, under the Strong Exponential Time Hypothesis\footnote{SETH is the pessimistic version of $\textup{P} \neq \textup{NP}$, stating  that we cannot solve $k$-SAT in $O((2-\eps)^n)$ time, for some $\eps>0$ independent of and for all constant $k$ \cite{IP01,CIP06}.} (SETH).
We complement the $O(nN \sqrt{\log{N/n}})$ upper bound of Gawrychowski \cite{gawrychowski2012faster} with an $(n N)^{1-o(1)}$ lower bound.
Thus, in the natural setting $n \approx N^{1/2}$ from above, we should indeed be content with the $\tilde{O}(N^{1.5})$ upper bound since we will not be able to get much closer to the gold standard, unless SETH fails.
Assuming SETH, our result confirms the conjecture of  Lifshits, up to $N^{o(1)}$ factors. 
See Theorem~\ref{thm:lcslb} in Section~\ref{sec:lcs} for the formal statement.

One way to view this result is as an \emph{Instance Optimality} result for LCS. The exact complexity of LCS on two strings is precisely proportional to the product of the decompressed size $N$ and the instance-inherent measure $n$ of how compressible they are.

\paragraph{RNA Folding and CFG Parsing}
Next, we turn our attention to two other fundamental problems: Context-Free Grammar Recognition (aka Parsing) and RNA Folding.
Parsing is the core computer science problem in which we want to decide whether a given string (e.g.\ computer code) can be derived from a given grammar (e.g.\ the grammar of a programming language).
Having the ability to efficiently parse a \emph{compressed} file is certainly desirable.
In RNA Folding we are given a string over some alphabet (e.g.\ $\{A,C,G,T\}$) with a fixed pairing between its symbols (e.g.\ $A-T$ match and $C-G$ match), and the goal is to compute the maximum number of non-crossing arcs between matching letters that one can draw above the string (which corresponds to the minimum energy folding in two dimensions).
RNA Folding is one of the most central problems in bioinformatics, and as we have discussed above, the ability to analyze compressed data is important in this field.
How fast can we solve these problems?

Given an uncompressed string of size $N$, classical dynamic programming algorithms, such as the CYK parser \cite{cocke1970programming,younger1967recognition,kasami1965efficient}, solve RNA Folding in $O(N^3)$ time and Parsing in $O(N^3 \cdot g)$ time if the grammar has size $g$.
Wikipedia lists twenty-four parsing algorithms designed throughout the years, all of which take cubic time in the worst case.
A theoretical breakthrough of Leslie Valiant \cite{valiant1975general} in 1975 showed that there are truly sub-cubic $O(gN^{\omega})$ parsing algorithms, where $\omega<2.38$ is the fast matrix multiplication (FMM) exponent.
However, Valiant's algorithm has not been used in practice due the inefficiency of FMM algorithms, and obtaining a \emph{combinatorial}\footnote{For the purposes of this paper, ``combinatorial" should be interpreted as any \emph{practically efficient} algorithm that does not suffer from the issues of FMM such as large constants and inefficient memory usage.} sub-cubic time algorithm would be of major interest.
Alas, it was recently proved \cite{ABV15b} that any improvement over these bounds implies breakthrough $k$-Clique algorithms:
either finding such a combinatorial subcubic algorithm \emph{or} getting any $O(N^{\omega-\eps})$ time algorithm, for any $\eps>0$, would refute the $k$-Clique Conjecture\footnote{Given a graph on $n$ nodes, the $k$-Clique conjecture \cite{ABV15b} is in fact two independent conjectures: The first one states that we cannot solve $k$-clique in $O(n^{(1-\eps) \cdot \omega k/3})$, for any $\eps>0$. The second one states that we cannot solve $k$-Clique combinatorially in $O(n^{(1-\eps)k})$ time, for any $\eps>0$.}.
The situation for RNA is even more interesting since Valiant's sub-cubic algorithm does not generalize to this case. 
Under the $k$-Clique conjecture, the same lower bounds still apply \cite{ABV15b,Chang15}, implying that any improvement will have to use FMM.
Indeed, an $O(N^{2.82})$ algorithm using FMM was recently achieved \cite{bringmann2016truly}.

Cubic time is a real bottleneck when analyzing large genomic data. One would hope that if we are able to compress the data down to size $n$ we could solve problems like RNA Folding and Parsing in time that is much faster than the $N^3$ lower bounds (to simplify the discussion we focus on combinatorial algorithms), such as $n^3 \cdot N^{o(1)}$ or at least $n^{1.5} N^{1.5}$, in certain analogy the LCS case.
No such algorithms were found to date, and we provide an explanation: 
Decompress-and-solve \emph{cannot be beaten} for Parsing and (essentially) for RNA Folding, under the $k$-Clique Conjecture.
For both problems we prove a conditional lower bound of $N^{\omega-o(1)}$ for any kind of algorithm, and $N^{3-o(1)}$ for combinatorial algorithms, even restricted to $n = O(N^\eps)$ for any $\eps > 0$. 
See Theorem~\ref{thm:cfglowerbound} in Section~\ref{sec:cfg} for CFG Parsing and Theorem~\ref{thm:rnalowerbound} in Section~\ref{sec:rna} for RNA Folding.

\iffalse
\begin{thm} \label{thm:cfglowerbound}
  Assuming the $k$-Clique conjecture, there is no $O(N^{\omega - \eps})$ time algorithm for CFG Recognition or RNA Folding for any $\eps > 0$. Assuming the combinatorial $k$-Clique conjecture, there is no combinatorial $O(N^{3 - \eps})$ time algorithm for CFG Recognition or RNA Folding for any $\eps > 0$. These results hold even restricted to instances with $n = O(N^\eps)$, and, for CFG Recognition, for a grammar size of $|\Gamma| = O(\log N)$. 
\end{thm}
\fi

\paragraph{Approximate Pattern Matching}
We continue our quest for quantifying the possible improvements over decompress-and-solve for basic problems. 
Consider the following compressed versions of important primitives in text analysis known as \emph{Approximate Pattern Matching} problems.
In all these problems we assume that we are given a compressed text $T$ of size $n$ (and decompressed size $N$), and a compressed pattern $P$ of size $m$ (and decompressed size $M$), both over some constant size alphabet.
\begin{itemize}
\item {\bf Pattern Matching with Wildcards:} In this problem, the strings contain wildcard symbols that can be replaced by any letter, and our goal is to decide if $P$ appears in $T$.
\item {\bf Substring Hamming Distance:} Compute the smallest Hamming distance of any substring of $T$ to $P$.
\end{itemize}
And a problem that generalizes both is:
\begin{itemize}
\item {\bf Generalized Pattern Matching:} Given some cost function on pairs of alphabet symbols, find the length-$M$ substring $T'$ of $T$ minimizing the total cost of all pairs $(T'[i],P[i])$.
\end{itemize}

The above problems have been extensively studied both in the uncompressed (see \cite{clifford2016k}) and in the compressed~\cite{lifshits2007processing,bille2015random,gagie2011faster} settings.
   All three problems can be solved in time $O(\min\{N \log N, nM\})$ 
   (see Section~\ref{sec:genpatternmatch}).
   Note that this bound beats the decompress-and-solve bound when the pattern is small, but can we avoid decompressing the pattern?
    We show a completely tight SETH-based lower bound of $\min\{N, nM\}^{1-o(1)}$ for all three problems, even for constant size alphabets and
    in all settings where the parameters are polynomially related.
See Theorems~\ref{thm:patmatchlb} and \ref{thm:subhamminglb} in Section~\ref{sec:genpatternmatch}.

\iffalse
\begin{thm}\label{thm:patmatchsubstringHDlb}
Assuming SETH, Substring Hamming Distance and Pattern Matching with Wildcards over alphabet $\{0,1\}$ (plus wildcards $*$) take time $\min\{N, nM\}^{1-o(1)}$. This holds even restricted to instances with $n = \Theta(N^{\alpha_n})$, $M = \Theta(N^{\alpha_M})$ and $m = \Theta(N^{\alpha_m})$ for any $0 < \alpha_n < 1$ and $0 < \alpha_m \le \alpha_M \le 1$.
\end{thm}
\fi

\paragraph{Language Membership}
Consider the compressed version of the most basic language membership problems.
Assume we are given a compressed string $T$ (again, from size $N$ to $n$).
\begin{itemize}
\item {\bf DFA Acceptance:} Given  $T$ and a DFA $F$ with $q$ states, decide whether $F$ accepts $T$.
\item {\bf NFA Acceptance:} Given  $T$ and a NFA $F$ with $q$ states, decide whether $F$ accepts $T$.
\end{itemize}

Classic algorithms solve the DFA Acceptance problem in time $O(\min\{nq, N+q\})$~\cite{plandowski1999complexity, hopcroft2006automata}, 
and we prove a matching SETH-based lower bound of $\min\{nq, N+q\}^{1-o(1)}$. 
See Theorem~\ref{thm:dfaaccept} in Section~\ref{sec:dfaaccept}.

For the NFA problem, the classic algorithms give $O(\min\{n q^\omega, N q^2\})$~\cite{markey2004ptime,plandowski1999complexity, hopcroft2006automata}. For combinatorial algorithms, we prove a matching lower bound of $\min\{n q^3, N q^2\}^{1-o(1)}$, under the (combinatorial) $k$-Clique conjecture. 
See Theorem~\ref{NFA_hardness} in Section~\ref{sec:nfaaccept}.
Our lower bounds hold for constant size alphabets, and in all settings of $n,N,q$, even restricted to instances with $N=\Theta(n^{\alpha_N})$ and $q=\Theta(n^{\alpha_q})$ for any $\alpha_N > 1$ and $\alpha_q > 0$.

 \paragraph{Disjointness, Hamming Distance, and Subsequence}
 Could it be that for other, even more basic problems the decompress-and-solve bound cannot be beaten?
 One candidate might be Disjointness, the canonical hard problem in Communication Complexity.
 \begin{itemize}
 \item {\bf Disjointness:} Given two equal-length bit-strings, is there a coordinate in which both are $1$?
 \end{itemize}

 The following two natural problems are at least as hard as Disjointness 
 (see Section~\ref{sec:partial}).
 
  \begin{itemize}
 \item {\bf Hamming Distance:} Compute the Hamming Distance of two strings.
 \item {\bf Subsequence:} Decide if a pattern of length $M$ is a subsequence of a text of length $N$.
  \end{itemize}

 Note that all these problems can be solved trivially in $O(N)$ time if our strings are uncompressed. 
 Could it be that we cannot solve them without decompressing our data? 
 We are not aware of any known algorithms solving any of these problems in $O(N^{1-\eps})$ time, for any $\eps>0$, even when our strings are compressed into size $n=O(N^{\alpha})$ for some small constant $\alpha>0$.
 The only exceptions are the known $\tilde{O}(M)$ time algorithms \cite{das1997episode,CGLM06,tiskin2009faster,yamamoto2011faster,tiskin2011towards,bille2014compressed} for the Subsequence problem, which beat the decompress-and-solve bound when the pattern is significantly smaller than the text. However, in the case $M=\Theta(N)$ no improvements seem to be known.

In Section~\ref{sec:partial} we present our attempts at proving a matching lower bound.
We prove the following: $N^{1-o(1)}$ for Subsequence in the setting $N = \Theta(M) = \Theta(n^2) = \Theta(m^2)$ and $|\Sigma| = O(N^\eps)$ (Theorem~\ref{thm:subseqlower}). $N^{1/4 -o(1)}$ for Disjointness (and thus also for the other two problems) in the setting $N = M$ and $n,m = O(N^\eps)$ for any $\eps > 0$, and constant alphabet size, assuming the $k$-SUM conjecture (Theorem~\ref{lb1}). Similarly: $N^{1/3 -o(1)}$ for Disjointness under Strong $k$-SUM conjecture (Theorem~\ref{lb2}).

Motivated by our inability to prove tight lower bounds for these basic problems, despite seemingly having the right framework, we have turned our attention to upper bounds.
In Section~\ref{sec:partial} we obtain the \emph{first} improvement over the decompress-and-solve bound for Disjointness, Hamming Distance, and Subsequence.
In particular, we obtain the first improvement over the decompress-and-solve bound for Disjointness, Hamming Distance, and Subsequence.
Our algorithms solve all these problems in $O(n^{1.410}\cdot N^{0.593})$ time.
As a side result, we also design a very simple algorithm for the Subsequence problem with $O((n|\Sigma| + M) \log N)$ runtime (Theorem~\ref{alg_subsequence}), which is comparable to the known but more involved algorithms \cite{bille2014compressed}.

One of the biggest benefits of having complexity theoretic results is that algorithm designers know what to focus on.
We believe that these upper bounds can be improved further and suggest it as an interesting open question: 
\emph{What is the time complexity of computing Disjointness on two grammar-compressed strings?}

$  $% !TEX root = main.tex

 \subsection{Technical Overview} \label{sec:techoverview}

From a technical perspective, our paper is most related to the conditional lower bounds for sequence similarity measures on strings and curves that have been shown in recent years, specifically, the SETH-based lower bounds for edit distance~\cite{BI15}, longest common subsequence~\cite{ABV15,BK15}, Fr\'echet distance~\cite{Bring14}, and others~\cite{AVW14,BI16,BGL16,Polak17}. 

These results all proceed as follows. Let $\phi$ be a given $k$-SAT instance on $\tilde{n}$ variables and clauses $C_1,\ldots,C_{\tilde{m}}$. We can assume that $\tilde{m} = O(\tilde{n})$ by the Sparsification Lemma~\cite{IPZ01}. Split the $\tilde{n}$ variables into two halves $X_1$ and $X_2$ of size $\tilde{n} / 2$. Enumerate all assignments $\alpha_1,\ldots,\alpha_{2^{\tilde{n}/2}}$ of the variables in~$X_1$. For any assignment $\alpha_i$ and any clause $C_\ell$, denote by $\textup{sat}(\alpha_i,C_\ell)$ whether $\alpha_i$ satisfies $C_\ell$, i.e., whether some variable in $X_1$ appears in $C_\ell$ (negated or unnegated) and is set by $\alpha_i$ so that $C_\ell$ is satisfied. Similarly, consider the assignments $\beta_1,\ldots,\beta_{2^{\tilde{n}/2}}$ of $X_2$.
By construction, we can solve the $k$-SAT instance $\phi$ by testing whether there are $\alpha_i, \beta_j$ such that $\textup{sat}(\alpha_i,C_\ell) \vee \textup{sat}(\beta_j,C_\ell)$ holds for all $\ell \in [\tilde{m}]$. 
Making use of this fact, all previous conditional lower bounds for sequence similarity measures essentially construct the following natural sequence:
\begin{align*}
W & = \textup{sat}(\alpha_1,C_1) \ldots \textup{sat}(\alpha_1,C_{\tilde{m}}) 
% \\ &  
\ldots \textup{sat}(\alpha_{2^{\tilde{n}/2}},C_1) \ldots \textup{sat}(\alpha_{2^{\tilde{n}/2}},C_{\tilde{m}}) \\
& = \bigconcat_{i \in [2^{\tilde{n}/2}]} \bigconcat_{\ell \in [\tilde{m}]} \textup{sat}(\alpha_i,C_\ell). 
\end{align*}
One typical variation of this string is to replace the bits $\{0,1\}$, indicating whether $\textup{sat}(\alpha_i,C_\ell)$ holds, by two short strings $\{B(0),B(1)\}$. 
Other typical variations are to add appropriate padding strings around the substrings $\bigconcat_{\ell \in [\tilde{m}]} \textup{sat}(\alpha_i,C_\ell)$ or around the whole sequence $W$. These paddings typically only depend on $\tilde{n}$ and $\tilde{m}$. Constructing a second sequence $W'$ with $\alpha_i$ replaced by $\beta_i$, one can then try to emulate the search for the half-assignments $\alpha_i, \beta_j$ by a similarity measure on $W,W'$.  All previous reductions follow this recipe, and thus construct a sequence like $W$.

\paragraph{Is \boldmath$W$ compressible?} For our purposes we need to construct compressible strings. Considering the entropy, the string~$W$ is very well compressible, since it only depends on the $\tOh(\tilde{n})$ input bits of the sparse $k$-SAT instance $\phi$. This entropy $\tOh(\tilde{n})$ is extremely small compared to the length $O(\tilde{n} 2^{\tilde{n}/2})$ of $W$. However, considering grammar-compression, the sequence $W$ is a bad representation, since \emph{$W$ is not generated by any SLP of size $o(2^{\tilde{n}/2}/\tilde{n})$ in general!} To see this, first observe that all substrings $\bigconcat_{\ell \in [\tilde{m}]} \textup{sat}(\alpha_i,C_\ell)$ of $W$ can potentially be different, meaning that $W$ can have $2^{\tilde{n}/2}$ different substrings of length $\tilde{m}$. This happens e.g.\ if for each variable $x_i \in X$ there is a clause $C_i$ consisting only of $x_i$ (which makes the $k$-SAT instance trivial, but shows that $W$ may have many different substrings in general). 
Second, observe that for any SLP $\cT$ consisting of $n$ non-terminals $S_1 \ldots S_n$ and for any length $L \ge 1$ the generated string $\eval(\cT)$ has at most $n \cdot L$ different substrings of length $L$. Indeed, a rule $S_i \to S_\ell S_r$ can only create a new substring, that is not already contained in $\eval(S_\ell)$ or $\eval(S_r)$, if this substring overlaps the boundary between $\eval(S_\ell)$ and $\eval(S_r)$ in $\eval(S_i)$. Hence, the rule $S_i \to S_\ell S_r$ can contribute at most $L$ new substrings of length $L$, amounting to at most $n L$ different substrings overall. Combining these two facts, with $L = \tilde{m} = O(\tilde{n})$, we see that $W$ in general has no SLP of size $o(2^{\tilde{n}/2}/\tilde{n})$.

Hence, the standard approach to conditional lower bounds for sequence similarity measures fails in the compressed setting, and it might seem like (SETH-based) conditional lower bounds are not applicable here.

\paragraph{A compressible sequence \boldmath$T$} 
On the contrary, we show that by simply inverting the ordering we obtain a very well compressible string:
\begin{align*} T & = \textup{sat}(\alpha_1,C_1) \ldots \textup{sat}(\alpha_{2^{\tilde{n}/2}},C_1) 
%\\ & \qquad 
\ldots \textup{sat}(\alpha_1,C_{\tilde{m}}) \ldots \textup{sat}(\alpha_{2^{\tilde{n}/2}},C_{\tilde{m}}) \\
& = \bigconcat_{\ell \in [\tilde{m}]} \bigconcat_{i \in [2^{\tilde{n}/2}]} \textup{sat}(\alpha_i,C_\ell). 
\end{align*}
The difference between $W$ and $T$ might seem negligible, but it greatly changes the game of emulating $k$-SAT by a sequence similarity measure: In $W$ we are looking for a local structure (a small substring) that ``fits together'' with a local structure in a different string $W'$. In $T$ we have to ensure the choice of a consistent offset $\Delta \in [n]$ and ``read'' the symbols $T[\Delta], T[\Delta + 2^{\tilde{n}/2}], \ldots, T[\Delta + (\tilde{m}-1) 2^{\tilde{n}/2}]$, which seems much more complicated. 

$T$ is compressible to an SLP $\cT$ of size $O(\tilde{n}^2)$, which is much smaller than the $\Omega(2^{\tilde{n}/2}/\tilde{n})$ bound for $W$. Indeed, consider a substring $\bigconcat_{i \in [2^{\tilde{n}/2}]} \textup{sat}(\alpha_i,C_\ell)$. We may assume that no variable appears more than once in $C_\ell$. Consider the following SLP rules, for $1 \le i \le \tilde{n}/2$, 

\noindent\begin{minipage}{.35\linewidth}
\begin{align*}
    A_0 &\to 1,  \\
    A_i &\to A_{i-1} A_{i-1},  \\
    S_0 &\to 0, 
\end{align*}
\end{minipage}%
\begin{minipage}{.65\linewidth}
\begin{align*}
   S_i &\to \begin{cases} S_{i-1} A_{i-1} & \text{if $x_i$ appears in $C_\ell$} \\  A_{i-1} S_{i-1} & \text{if $\neg x_i$ appears in $C_\ell$} \\ S_{i-1} S_{i-1} & \text{otherwise} \end{cases}
\end{align*}
\end{minipage}
\medskip

\noindent
We clearly have $\eval(A_i) = 1^{2^i}$. Moreover, if $\neg x_i$ appears in $C_\ell$, then for $x_i = 0$, no matter what we choose for $x_1,\ldots,x_{i-1}$, we have $\textup{sat}(\alpha_j,C_\ell) = 1$, and thus we may write $A_{i-1}$. For $x_i=1$ we note that the value $\textup{sat}(\alpha_j,C_\ell)$ only depends on the remaining variables $x_1,\ldots,x_{i-1}$, and thus we may write $S_{i-1}$. Along these lines, one can check that $\eval(S_{\tilde{n}/2}) = \bigconcat_{i \in [2^{\tilde{n}/2}]} \textup{sat}(\alpha_i,C_\ell)$. Creating such an SLP for each $\ell \in [\tilde{m}]$ and constructing their concatenation, we obtain an SLP of size $O(\tilde{m} \tilde{n}) = O(\tilde{n}^2)$ generating $T$.

\paragraph{Example Lower Bound: Pattern Matching with Wildcards}
In the remainder of this section, we present an easy example for a conditional lower bound on compressed strings, namely for the problem Pattern Matching with Wildcards. Here we consider an alphabet $\Sigma$ and we say that symbols $\sigma, \sigma' \in \Sigma \cup \{*\}$ \emph{match} if $\sigma = *$ or $\sigma' = *$ or $\sigma = \sigma'$. We say that two equal-length strings $X,Y$ (over alphabet $\Sigma \cup \{*\}$) match if $X[i]$ and $Y[i]$ match for all $i$. Given a text $T$ of length $N$ and a pattern $P$ of length $M \le N$, the task is to decide whether $P$ matches some length-$M$ substring of $T$.

Let $\phi$ be a $k$-SAT instance as above, but this time let $\alpha_1,\ldots,\alpha_{2^{\tilde{n}}}$ be \emph{all} the assignments of the $\tilde{n}$ variables in $\phi$. We define the text $T$ and pattern $P$ by
\[ T = \bigconcat_{\ell \in [\tilde{m}]} \bigconcat_{i \in [2^{\tilde{n}}]} \textup{sat}(\alpha_i,C_\ell) \qquad\qquad P = 1 (*^{2^{\tilde{n}}-1} 1)^{\tilde{m}-1}. \]
Note that $P$ matches some substring of $T$ if and only if there is an offset $\Delta \in [2^{\tilde{n}}]$ such that $T[\Delta] = T[\Delta + 2^{\tilde{n}}] = \ldots = T[\Delta + (\tilde{m}-1) 2^{\tilde{n}}] = 1$, which happens if and only if $\alpha_\Delta$ is a satisfying assignment of $\phi$. Hence, we constructed an equivalent instance of Pattern Matching with Wildcards. 

Analogously to above, one can show that $T$ is generated by an SLP $\cT$ of size $n = O(\tilde{n}^2)$ that can be computed in time $O(\tilde{n}^2)$. Similarly, it is easy to see that $P$ is generated by an SLP $\cP$ of size $O(\tilde{n})$ that can be computed in time $O(\tilde{n})$. Hence, the reduction runs in time $O(\tilde{n}^2)$. We stress that we define strings $T,P$ of exponential length in $\tilde{n}$, but in the reduction we never explicitly write down any such string, but we simply construct compressed representations. Since the resulting strings have length $O(2^{\tilde{n}} \tilde{n})$, any $O(N^{1-\eps})$ time algorithm for Pattern Matching with Wildcards would imply an algorithm for $k$-SAT in time $O(2^{(1-\eps)\tilde{n}} \textup{poly}(\tilde{n}))$, contradicting the Strong Exponential Time Hypothesis (SETH). Note that this conditional lower bound of $N^{1-o(1)}$ holds even for strings compressible to size $\textup{polylog}(N)$. 

In Section~\ref{sec:genpatternmatch} we analyze Pattern Matching with Wildcards in more detail and show that the optimal running time, conditional on SETH, is $\min\{N, n M\}^{1 \pm o(1)}$, and this holds for all settings of the text length $N$, the compressed text size $n$, the pattern length $M$, and the compressed pattern size $m$.

In Pattern Matching with Wildcards, we got a consistent choice of an offset $\Delta$ for free. It is much more complicated to achieve this for other problems such as Longest Common Subsequence, CFG Parsing, or RNA Folding. 
This overview summarized the main technical contributions of this paper, but left out many problem-specific tricks that can be found in the subsequent proofs, and that we think will find more applications for analyzing problems on compressed strings.

% !TEX root = main.tex

\section{Preliminaries} \label{sec:prelim}

Here we give general preliminaries on strings, straight-line programs, and hardness assumptions. Problem definitions and additional problem specific preliminaries will be given in the corresponding sections. For a positive integer $n$ we let $[n] = \{1,\ldots,n\}$, while for a proposition $A$ we let $[A]$ be 1 if $A$ is true and 0 otherwise. 

\paragraph{Strings}
Let $\Sigma$ be a finite alphabet. In most parts of this paper we assume that $|\Sigma| = O(1)$, but in exceptional cases we allow the alphabet to grow with the input size. For a string $T$ over alphabet $\Sigma$, we write $|T|$ for its length, $T[i]$ for its $i$-th symbol, and $T[i..j]$ for the substring from position $i$ to position~$j$. For two strings $T, T'$ we write $T \concat T'$, or simply $T\, T'$, for their concatenation. For $k \ge 1$ we let $T^k := \bigconcat_{i=1}^k T$.

\paragraph{Straight-Line Programs (SLPs)} 
An SLP $\cT$ is a set of non-terminals $S_1,\ldots,S_n$, each equipped with a rule of the form (1) $S_i \to \sigma$ for some $\sigma \in \Sigma$ or (2) $S_i \to S_{\ell(i)}, S_{r(i)}$ with $\ell(i),r(i) < i$. The string $T$ generated by SLP $\cT$ is recursively defined as follows. For a rule $S_i \to \sigma$ we let $\eval(S_i) := \sigma$, and for a rule $S_i \to S_{\ell(i)}, S_{r(i)}$ we let $\eval(S_i) := \eval(S_{\ell(i)}) \concat \eval(S_{r(i)})$. Then $T = \eval(\cT) := \eval(S_n)$ is the string generated by SLP $\cT$. Note that an SLP is a context-free grammar describing a unique string; so $\cT$ is a \emph{grammar-compressed} representation of $T$.
We call $|\cT| = n$ the \emph{size} of $\cT$.
See Figure~\ref{fig_slp} for the depiction of an SLP; in particular note the difference between the \emph{directed acyclic graph} that is the compressed representation $\cT$ and the \emph{parse tree} that we obtain by decompressing $\cT$ to a tree whose leaves spell the decompressed text $T$.

For an SLP $\cT$ with non-terminals $S_1,\ldots,S_n$, we recursively define the depth $\depth(S_i)$ as follows. For a rule $S_i \to \sigma$ we set $\depth(S_i):=0$. For a rule $S_i \to S_{\ell(i)}, S_{r(i)}$ we set $\depth(S_i) = \max\{\depth(S_{\ell(i)}), \depth(S_{r(i)})\} + 1$. The depth of $\cT$ is defined as $\depth(S_n)$. 
The SLP $\cT$ is called an \emph{AVL-grammar}~\cite{Rytter03} if it is balanced: for any rule $S_i \to S_{\ell(i)}, S_{r(i)}$ in the SLP we have $|\depth(S_{\ell(i)}) - \depth(S_{r(i)})| \le 1$. This implies that the depth of $\cT$ is $O(\log N)$, where $N = |\eval(\cT)|$. 

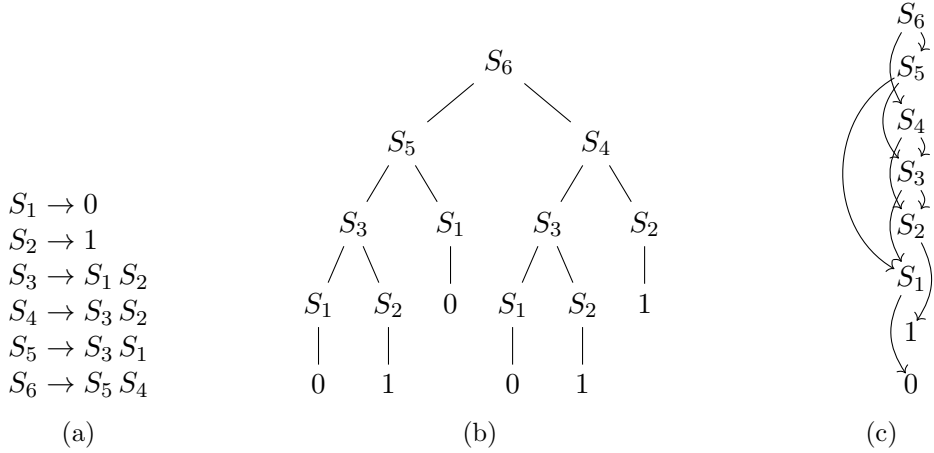
\begin{figure*}
	\centering
	
	\begin{subfigure}[b]{0.3\textwidth}
        \centering
			\begin{tabular}{ l }
			 $S_1\to 0$ \\
			 $S_2\to 1$ \\ 
			 $S_3\to S_1\, S_2$ \\ 
			 $S_4\to S_3\, S_2$ \\ 
			 $S_5\to S_3\, S_1$ \\ 
			 $S_6\to S_5\, S_4$
			\end{tabular}
        \caption{}
        \label{fig_a}
    \end{subfigure}
	~
	\begin{subfigure}[b]{0.3\textwidth}
        \centering
        \begin{forest}
			[$S_6$
				[$S_5$
					[$S_3$						
						[$S_1$
							[$0$]
						]
						[$S_2$
							[$1$]
						]
					]
					[$S_1$
						[$0$]
					]
				]
				[$S_4$
					[$S_3$						
						[$S_1$
							[$0$]
						]
						[$S_2$
							[$1$]
						]
					]
					[$S_2$
						[$1$]
					]
				]
			]
		\end{forest}
        \caption{}
        \label{fig_b}
    \end{subfigure}
	~
	\begin{subfigure}[b]{0.3\textwidth}
        \centering
		%[shorten >=1pt,auto,node distance=1cm]
        \begin{tikzpicture}[
      %>=Stealth,
      shorten >=-3pt,
      shorten <=-3pt,
      auto,
      node distance=.7 cm,
      scale = 1,
      transform shape,
      state/.style={circle,inner sep=2pt}
      ]
			%\node[initial,state] (q1)      {$1$};
			\node[state] (z0) {$0$};
			\node[state] (z1) [above of=z0] {$1$};
			\node[state] (s1) [above of=z1] {$S_1$};
			\node[state] (s2) [above of=s1] {$S_2$};
			\node[state] (s3) [above of=s2] {$S_3$};
			\node[state] (s4) [above of=s3] {$S_4$};
			\node[state] (s5) [above of=s4] {$S_5$};
			\node[state] (s6) [above of=s5] {$S_6$};
			%\node[state, accepting]         (q3) [above right of=q4] {$3$};

			%\path[->]          (q1)  edge                 node {a} (q2);
			\path[->] (s1) edge [bend right] (z0);
			\path[->] (s2) edge [bend left] (z1);
			\path[->] (s3) edge [bend right=30] (s1);
			\path[->] (s3) edge [bend left] (s2);
			\path[->] (s4) edge [bend left] (s3);
			\path[->] (s4) edge [bend right=30] (s2);
			\path[->] (s5) edge [bend right=40] (s3);
			\path[->] (s5) edge [bend right=60] (s1);
			\path[->] (s6) edge [bend left] (s5);
			\path[->] (s6) edge [bend right=30] (s4);
			%\path[->]          (q3)  edge   [bend left]   node {a} (q2);
			%\path[->]          (q2)  edge                 node {b} (q4);
			%\path[->]          (q4)  edge                 node {a} (q3);
		\end{tikzpicture}
        \caption{}
        \label{fig_c}
    \end{subfigure}
	
	\caption{(a) An SLP generating the sequence $010011$. (b) The corresponding parse tree. (c) The acyclic graph corresponding to the SLP.}
	\label{fig_slp}
\end{figure*}

\begin{thm}[\cite{Rytter03}] \label{AVL}
	Given a text $T$ of length $N$ by an SLP $\cT$ of size~$n$, in $O(n\log N)$ time we can construct an AVL-grammar $\cT'$ for $T$ with size $O(n\log N)$ and depth $O(\log N)$.
\end{thm}

\begin{obs} \label{obs:repetition}
  For any string $T$ and $k \ge 1$, there is an SLP of size $O(|T| + \log k)$ generating the string $T^k$.
\end{obs}

In all problems considered in this paper, the input contains a text $T$ given by a grammar-compressed representation $\cT$, such that $T = \eval(\cT)$. We always denote by $N = |T|$ the length of the text and by $n = |\cT|$ the size of its representation. Sometimes we are additionally given a pattern $P$ by a grammar-compressed representation $\cP$, and we denote the pattern length by $M = |P|$ and its representation size by $m = |\cP|$.

\subsection{Hardness Assumptions} \label{sec:hardnessassumptions}

\paragraph{SETH and OV}
The Strong Exponential Time Hypothesis (SETH) was introduced by Impagliazzo, Paturi, and Zane~\cite{IPZ01} and asserts that the central NP-hard satisfiability problem has no algorithms that are much faster than exhaustive search.
\begin{conjecture}[SETH]
\label{conj:seth}
  There is no $\varepsilon>0$ such that for all $k \ge 3$, $k$-SAT on $n$ variables can be solved in time $O(2^{(1-\varepsilon)n})$.
\end{conjecture}
 Effectively all known SETH-based lower bounds for polynomial-time problems use reductions via the \emph{Orthogonal Vectors problem} (OV):
Given sets $\sA$, $\sB\subseteq \{0,1\}^d$ of size $|\sA| = A$, $|\sB| = B$, determine whether there exist vectors $a \in \sA$, $b \in \sB$ with $\sum_{i=1}^d a[i] \cdot b[i] = 0$.
Simple algorithms solve OV in time $\Oh(2^d (A+B))$ and $\Oh(d A B)$. For $A=B$ and $d=c(A)\log A$ the fastest known algorithm runs in time $A^{2-1/\Oh(\log c(A))}$~\cite{AWY15}, which is only slightly subquadratic for $d \gg \log A$. This has led to the following conjecture, which follows from SETH~\cite{W04}.
\begin{conjecture}[OV]
\label{conj:ov}
  For any $\varepsilon>0$ and $\beta > 0$, on instances with $B = \Theta(A^\beta)$ OV has no $O(A^{1+\beta-\eps} \poly(d))$ time algorithm.
\end{conjecture}
It is known that if this conjecture holds for some $\beta > 0$ then it holds for all $\beta > 0$, see e.g.~\cite{BK15}.

More generally, for $k \ge 2$ we say that a tuple $(a_1,\ldots,a_k)$ with $a_i \in \{0,1\}^d$ is \emph{orthogonal} if for all $\ell \in [d]$ there exists an $i \in [k]$ such that $a_i[\ell] = 0$.
In the $k$-OV problem we are given a set $\sA \subseteq \{0,1\}^d$ of size $A$ and want to determine whether there is an orthogonal tuple $(a_1,\ldots,a_k)$ with $a_i \in \sA$. The fastest known algorithm for $k$-OV is to run an easy reduction to OV and then solve OV. The following conjecture follows from SETH.
\begin{conjecture}[$k$-OV]
\label{conj:kov}
  For any $\varepsilon>0$ and $k \ge 2$, $k$-OV is not in time $O(A^{k - \eps} \poly(d))$.
\end{conjecture}

\paragraph{\boldmath$k$-Clique}
The fundamental $k$-Clique problem asks whether a given (undirected, unweighted) graph $G=(V,E)$ contains $k$ nodes that are pairwise adjacent. 
$k$-Clique is among the most well-studied problems in theoretical computer science, and it is the canonical intractable (W[1]-complete) problem in parameterized complexity.
With slight abuse of notation, we will denote the number of vertices and edges of $G$ by $V$ and $E$, respectively. 
The naive algorithm for $k$-Clique takes time $O(V^k)$. If $k$ is divisible by 3, the fastest known algorithm runs in time $O(V^{\omega k/3})$, where $\omega < 2.373$ is the exponent of matrix multiplication~\cite{NP85}. See~\cite{EG04} for the case that $k$ is not divisible by~3. To improve this bound is a longstanding open problem~\cite{woeginger,babai}. Since fast matrix multiplication is considered impractical, researchers also studied \emph{combinatorial} algorithms, that avoid fast matrix multiplication\footnote{Combinatorial algorithms are a notion without agreed upon definition; finding a formal definition is considered an open problem.}. The fastest combinatorial algorithm runs in time $O(V^k/\log^k{V})$~\cite{VClique}. The following conjectures assert that these bounds are close to optimal, and have been used e.g. in~\cite{ABV15b,BGL16}. 

\begin{conjecture}[$k$-Clique] \label{conj:clique}
For any $\varepsilon>0$ and $k \ge 3$, $k$-Clique has no $O(V^{(1-\varepsilon) \omega k / 3})$ algorithm.
\end{conjecture}

\begin{conjecture}[Combinatorial $k$-Clique] \label{conj:combclique}
For any $\varepsilon>0$ and $k \ge 3$, $k$-Clique has no combinatorial $O(V^{(1-\varepsilon) k})$ algorithm.
\end{conjecture}

\paragraph{\boldmath$k$-SUM}
In the $k$-SUM problem, we are given integers $R,t\geq 0$ and a set $Z\subseteq \{0,1, \ldots, R\}$ of $|Z|=r$ integers, and the task is to decide whether there are $k$ (not necessarily distinct) integers $z_1, \ldots, z_k \in Z$ that sum to $t$, i.e., $z_1+\ldots+z_k=t$.
This problem has well-known algorithms in time $O(r^{\lceil k/2 \rceil})$ and $O(r + R \log R)$, and it is conjectured that no much faster algorithms exist. The following conjectures, which generalize the more popular 3-SUM conjecture~\cite{overmars,patrascu2010towards} and Strong 3-SUM conjecture~\cite{amir2014hardness}, remain believable despite recent algorithmic progress~\cite{Austrin13,add_comb,GP14,Wang14}.

\begin{conjecture}[$k$-SUM] \label{conj1}
	For any $k \ge 3$ and $R=r^k$, the $k$-SUM problem is not in time $O(r^{\ceil{k/2}-\eps})$ for any $\eps > 0$.
\end{conjecture}

\begin{conjecture}[Strong $k$-SUM] \label{conj2}
	For any $k \ge 3$ and $R=r^{\ceil{k/2}}$, the $k$-SUM problem is not in time $O(r^{\ceil{k/2}-\eps})$ for any $\eps > 0$.
\end{conjecture}

\section{Tight Bounds Assuming SETH} \label{sec:sethlowerbounds}

% !TEX root = main.tex

\newcommand{\proofpart}[1]{\emph{#1.} }

\newcommand{\fail}{\mathrm{fail}}

In this section we prove matching conditional lower bounds based on the Strong Exponential Time Hypothesis (SETH, see Conjecture~\ref{conj:seth}) for the following problems:
\begin{itemize}
	\item DFA Acceptance, i.e., deciding whether a given deterministic finite automaton accepts a given string,
	\item Substring Hamming Distance, i.e., determining the minimum Hamming distance that can be achieved by aligning a given pattern sequence with a substring of a given text sequence,
	\item Pattern Matching with Wildcards, i.e., deciding whether the given pattern sequence (containing wildcards that match any symbol) matches a substring of the given text,
	\item Longest Common Subsequence, i.e., computing the length of the longest common subsequence of two given strings.
\end{itemize}
See the respective subsections for precise problem definitions. In all our proofs, instead of using SETH directly, we use the more convenient OV conjecture (Conjecture~\ref{conj:ov}) or $k$-OV conjecture (Conjecture~\ref{conj:kov}), which are implied by SETH. 

For DFA Acceptance, the compression used in our reduction from the given \OV instance is extremely simple, in that we only rely on the fact that any repetition $T^\ell$ can be generated by an SLP of size $O(|T| + \log \ell)$ (Observation~\ref{obs:repetition}). 

For Substring Hamming Distance, Pattern Matching with Wildcards and Longest Common Subsequence, however, our construction are more subtle.
We crucially use the following idea: consider a \kOV instance $\sA$ on $A$ vectors in $d$ dimensions. There is a length-$\Oh(dA^k)$ text $T$ representing this instance so that (1) $T$ is succinctly described by an SLP \ctext of size $\Oh(d A)$ and (2) testing whether the \kOV instance has a solution corresponds to determining whether there is some $i=1,\ldots,A^k$ such that all bits $T[i], T[i+A^k], \dots, T[i+(d-1)A^k]$ are equal to zero. 
Intuitively, $i \in \{1,\ldots,A^k\}$ denotes the $i$-th $k$-tuple of vectors in $\sA^k$, and $T[i]=0$ holds if and only if the $k$ vectors in the $i$-th $k$-tuple are orthogonal in the $1$st coordinate. In general, for $1 \le \ell \le d$, $T[i+(\ell-1) A^k]=0$ holds if and only if the $k$ vectors are orthogonal in the $\ell$-th coordinate. More formally, we set $T$ to be
$$
	T=\bigconcat_{\ell=1}^d \bigconcat_{a_1 \in \sA^{(1)}} \ldots \bigconcat_{a_k \in \sA^{(k)}} \big[a_1[\ell]=\ldots=a_k[\ell]=1\big],
$$
where $[.]$ is the Kronecker symbol, i.e., $[\textup{true}] = 1$ and $[\textup{false}] = 0$. For any $\ell$, the sequence $\bigconcat_{a_1 \in \sA^{(1)}} \ldots \bigconcat_{a_k \in \sA^{(k)}} [a_{1}[\ell]=\ldots=a_k[\ell]=1]$ is generated by an SLP of size $O(d A)$: if $a_{1}[\ell]=0$, then for all $a_2, \ldots, a_k$, the vectors will be orthogonal in this coordinate and we can write $0^{A^{k-1}}$,  which is well compressible by Observation~\ref{obs:repetition}. Otherwise, if $a_{1}[\ell]=1$, we recurse on $a_2,\ldots,a_k$ and the following $A^{k-1}$ symbols do not depend on $a_{1}[\ell]$ anymore.

A modification of the above construction of $T$ gives SETH hardness for Substring Hamming Distance and Pattern Matching with Wildcards. Showing hardness for Longest Common Subsequence requires more ideas. In particular, to be able to show tight hardness we extend the framework of~\cite{BK15}.

We stress that if the sequence $T$ would enumerate all $k$-tuples one after another (instead of iterating over the coordinates in the outer loop over $\ell$), then it would not be compressible using SLPs, see Section~\ref{sec:techoverview}. This makes our reductions quite different from all previously known hardness results where the sequences are concatenations of vector gadgets one after another.

\paragraph{Known Lower Bounds from Classic Complexity Theory}
We observe that the Substring Hamming Distance problem is a generalization of the Hamming Distance problem which asks to output the Hamming distance between a compressed text and a compressed pattern of equal length. The latter problem is known to be $\#{\sf P}$-complete and thus the Substring Hamming Distance problem is $\#{\sf P}$-hard (see the discussion at the beginning of Section~\ref{sec:subseqlower}). Similarly, Longest Common Subsequence is a generalization of the problem of deciding whether a given pattern is a subsequence of a given text. The latter problem is known to be ${\sf PP}$-hard (see the aforementioned discussion) and this yields ${\sf PP}$-hardness for Longest Common Subsequence.

The DFA Acceptance problem can be solved in polynomial time (see Section~\ref{sec:dfaaccept}) and no conditional lower bounds were known for this problem.
Finally, our reduction in Theorem~\ref{thm:patmatchlb} below shows that Pattern Matching with Wildcards is ${\sf NP}$-hard.

\subsection{DFA Acceptance} \label{sec:dfaaccept}

Recall that a finite-state automaton $F$ over an alphabet $\Sigma$ consists of a set of states $Z$ of size $q$, a starting state $z_0 \in Z$, a set of accepting states $Z' \subseteq Z$, and a set of transitions $\transition{z}{\sigma}{z'}$ with $z,z' \in Z$ and $\sigma \in \Sigma$. We lift this notation to strings $T=T[1..N]$ by writing $\transition{z}{T}{z'}$ whenever there are states $z_1,\dots,z_{\ell-1}$ and transitions $\transition{z}{T[1]}{z_1}, \transition{z_1}{T[2]}{z_2}, \dots, \transition{z_{\ell-1}}{T[N]}{z'}$. Furthermore, for a set $S\subseteq \Sigma$ we write $\transition{z}{S}{z'}$ whenever $\transition{z}{\sigma}{z'}$ for all $\sigma\in S$. 
The automaton $F$ is deterministic if for any $z \in Z$ and $\sigma \in \Sigma$ there is at most one $z' \in Z$ with transition $\transition{z}{\sigma}{z'}$, and $F$ is non-determinisitic otherwise.
The automaton $F$ \emph{accepts} a given string $T$ if $\transition{z_0}{T}{z'}$ holds for some accepting state $z'$. 

Throughout this section, we assume the alphabet size to be constant. If $F$ is a deterministic finite-state automaton (DFA), we may assume without loss of generality that for every state $z$ and symbol $\sigma \in \Sigma$, there always exists a (uniquely defined) state $z'$ with $\transition{z}{\sigma}{z'}$.\footnote{Note that we can always define an absorbing non-accepting state $z^\fail$ with $\transition{z^{\fail}}{\Sigma}{z^\fail}$ and set, for any undefined transition from $z$ under $\sigma$, $\transition{z}{\sigma}{z^\fail}$, which increases the number of states only by one.}
We fix the input description of $F$ to a list of transitions of $F$ as well as a list of accepting states.
Observe that any DFA $F$ on constant alphabet $\Sigma$ has an input size of $\Oh(q)$.

Consider the compressed variant of the acceptance problem of DFAs.
\begin{problem}[\DFAAccept]
Given a text $T$ of length $N$ by a grammar-compressed representation \ctext of size $n$ as well as a DFA $F$ with $q$ states, decide whether $T$ is accepted by $F$.
\end{problem}

The naive solution decompresses \ctext to obtain $T$ and runs the obvious acceptance algorithm for DFAs, which takes time $\Oh(|T|+q)=\Oh(N+q)$. Exploiting the compressed setting, one can obtain an $\Oh(nq)$-time algorithm~\cite{plandowski1999complexity}: 
Recall that $\cT$ is a set of rules of the form $S_i \to S_{\ell(i)} S_{r(i)}$ or $S_i \to \sigma_i$, with $\ell(i), r(i) < i$ and $\sigma_i \in \Sigma$, for $1 \le i \le n$.
We compute, for increasing $i$, the state transition function $f_i \colon [q] \to [q]$ (we denote states using integers $1, \ldots, q$) that satisfies $\transition{z}{\eval(S_i)}{f_i(z)}$.
For $S_i \to S_{\ell(i)} S_{r(i)}$ we can compute $f_i$ as $f_{r(i)} \circ f_{\ell(i)}$, where $\circ$ is function composition. For $S_i \to \sigma_i$ we simply have $f_i(z) = z'$ for the unique transition $\transition{z}{\sigma_i}{z'}$. 
Hence, $f_i$ can be computed in time $\Oh(q)$ for every $i$. 
The text $T$ is then accepted by $F$ if and only if $f_n(z_0)$ is an accepting state, where $z_0$ is the starting state of $F$. Hence, the best-known algorithm takes time $O(\min\{nq, N+q\})$. 

We prove that DFA Acceptance takes time $\min\{nq, N+q\}^{1-o(1)}$ assuming \SETH, thus providing a conditional lower bound matching the known algorithmic results. It is straightforward to see that any algorithm must read the complete input description of $F$ to always correctly decide the problem, yielding a lower bound of $\Omega(q)$. In the remainder, we provide the remaining conditional lower bound of $\min\{nq, N\}^{1-o(1)}$.

\begin{thm} \label{thm:dfaaccept}
Assuming the OV conjecture, for no $\varepsilon >0$ there is an $\Oh( \min\{nq, N\}^{1-\varepsilon})$-time algorithm for \DFAAccept. This holds even restricted to instances with $N=\Theta(n^{\alpha_N})$ and $q=\Theta(n^{\alpha_q})$ for any $\alpha_N > 1$ and $\alpha_q > 0$.
\end{thm}

\begin{proof}
Let $\sA = \{a_1,\dots, a_A\},\sB= \{b_1,\dots, b_B\}$ be a given \OV instance in $d$ dimensions. We construct a string $T$ of length $N = O(d A B)$ with a representation $\cT$ of size $n = O(d A)$ and a DFA $F$ with $q = O(d B)$. An $O(\min\{nq, N\}^{1-\eps})$-time algorithm for \DFAAccept would then imply an algorithm for \OV in time $O((d^2 A B)^{1-\eps}) = O((A B)^{1-\eps} \poly(d))$, contradicting the \OV conjecture. At the end of this proof we show that this also holds for all restrictions $N = \Theta(n^{\alpha_N})$ and $q = \Theta(n^{\alpha_q})$ with $\alpha_N > 1$ and $\alpha_q > 0$.

\paragraph{Constructing the Text \boldmath$T$} We cast any vector $a\in \{0,1\}^d$ to a string $T(a) := \bigconcat_{k=1}^d a[k]$ by simply concatenating its coordinates. We define the text $T$ over the alphabet $\Sigma = \{0,1,\#, !\}$ as
\begin{equation}
T := \left(! \concat \bigconcat_{i=1}^{A} \#\, T(a_i)\right)^{B}.
\end{equation}
Here, we think of $!$ and $\#$ as ``new group'' and ``new vector within group'' indicators, respectively. Intuitively, the $j$-th repetition of $T(a_i)$ is supposed to lead to an accepting state of $F$ if $a_i$ and $b_j$ are an orthogonal pair. 

\begin{figure}
\centering
\includegraphics[width=0.7\textwidth]{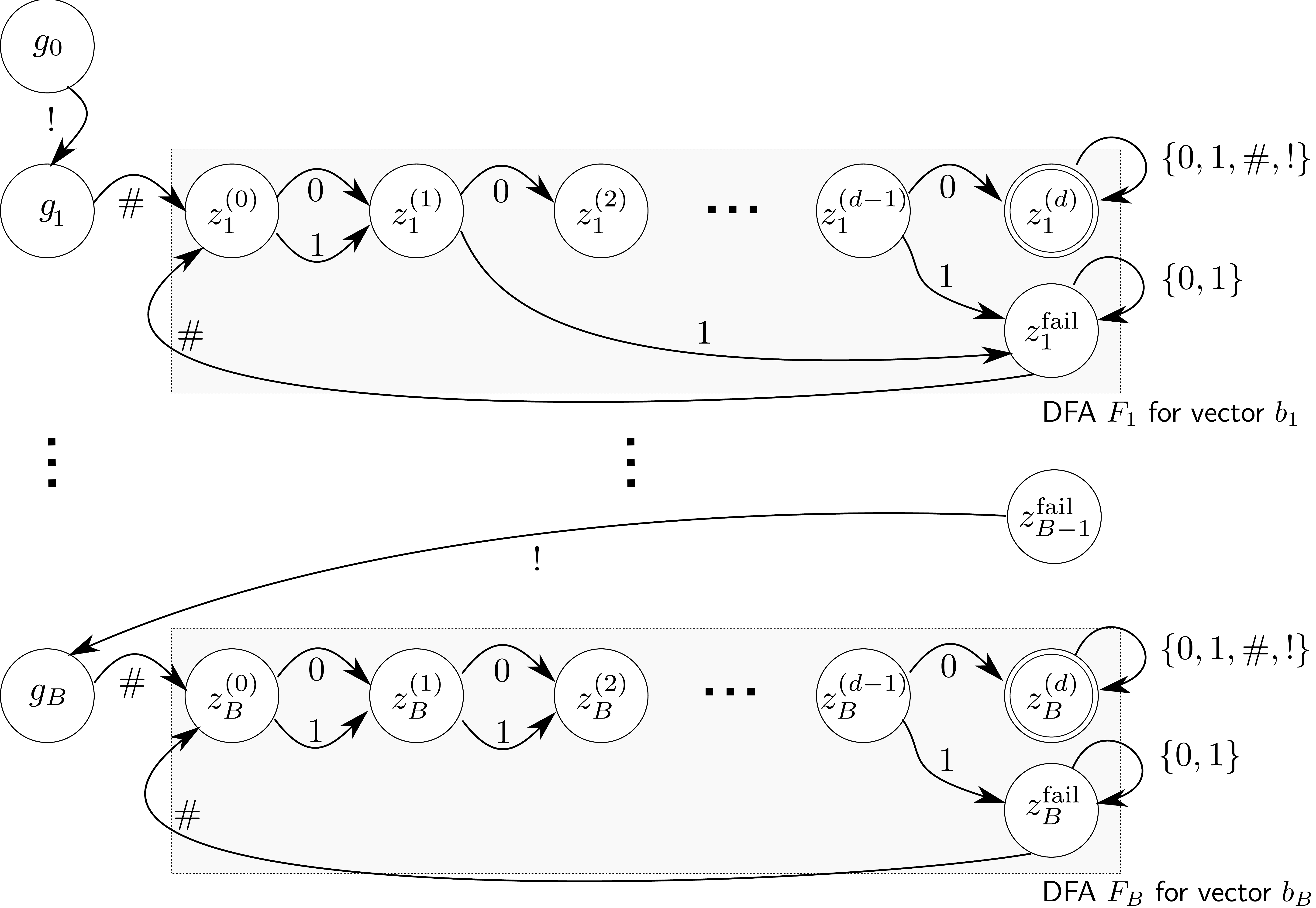}
\caption{Illustration of the DFA $F$. Any transition not specified leads to a absorbing non-accepting state $z^\fail$.}
\label{fig:dfa}
\end{figure}

\paragraph{Constructing the DFA \boldmath$F$}
For an illustration of the DFA construction see \figref{dfa}.
We start by defining ``vector gadgets'': For any vector $b_j\in \sB$ we construct a DFA $F_j$ over alphabet $\{0,1\}$ with states $z^{(0)}_j, z_j^{(1)}, \dots, z_j^{(d)}$ and $z_j^\fail$. The initial state is $z^{(0)}_j$. For any $k\in [d]$ we have the transitions $\transition{z^{(k-1)}_j}{0}{z^{(k)}_j}$ and
\[ \transition{z^{(k-1)}_j}{1}{ \begin{cases} z^{(k)}_j} &  \text{if } b_j[k] = 0, \\ z^{\fail}_j & \text{otherwise.} \end{cases} \]
We let $z^{(d)}_j$ be an accepting state with transition $\transition{z^{(d)}_j}{\{0,1\}}{z^{(d)}_j}$. Furthermore, we have $\transition{z^\fail_j}{\{0,1\}}{z^\fail_j}$. It is easy to see that after reading a string $a \in \{0,1\}^d$, $F_j$ is in either $z^{(d)}_j$ or $z^\fail_j$, and it is in $z^{(d)}_j$ if and only if $a$ and $b_j$ are orthogonal.

We combine these smaller DFAs to our final DFA $F$ over the slightly larger alphabet $\Sigma = \{0,1,\#,!\}$ as follows. We define additional states $g_0, g_1, \dots, g_B$ and let $g_0$ be the initial state of $F$. We define the following additional transitions:
\begin{align*}
&\transition{g_0}{!}{g_1} && \\
&\transition{g_j}{\#}{z_j^{(0)}} && \text{for } 1 \le j \le B,\\
&\transition{z^\fail_{j-1}}{!}{g_{j}} && \text{for } 1 < j \le B,\\
&\transition{z^\fail_j}{\#}{z_j^{(0)}} && \text{for } 1 \le j \le B,\\
&\transition{z^{(d)}_j}{\{\#,!\}}{z_j^{(d)}} && \text{for } 1 \le j \le B.
\end{align*}

In this way, each $z_{j}^{(d)}$ is an absorbing accepting state, and the symbols $\#$ and $!$ satisfy the semantics of jumping to the next vector in $\sA$ and $\sB$, respectively. This finishes the definition of the reduction.

\paragraph{Correctness}
We claim that the constructed DFA $F$ accepts $T$ if and only if $\sA,\sB$ contains an orthogonal pair. 
By structure of $F$ and $T$, as well as the properties argued for $F_j, j\in [B]$, it is straightforward to show that after reading any prefix $T'$ of $T$ ending on $\#$, $F$ is in the initial state of $F_j$, where $j$ is the number of !'s in $T'$  -- this holds until $F$ has encountered an accepting state for the first and final time. Thus, if $T$ is accepted by $F$, then some prefix $T'$ of $T$ that ends on $\#\, a_i$ for some $i\in [A]$ has led an accepting state of $F$. This can only happen if $a_i$ is orthogonal to $b_j$, where $j$ is the number of !'s in $T'$, i.e., $\sA,\sB$ contain an orthogonal pair. Conversely, if $\sA,\sB$ contains an orthogonal pair, let $a_i,b_j$ the smallest such pair in terms of the lexicographic order on $(j,i)$. Then the prefix $T'$ that ends on $\#a_i$ and contains $j$ !'s  leads to the accepting state $z^{(d)}_j$. 

\paragraph{Size Bounds}
We count that $|T|= B ((d+1)A+1) = O(d A B)$. Since $T$ consists of $B$ repetitions of a string of length $(d+1)A + 1$, we can compute an SLP $\ctext$ of size $|\ctext| \le \Oh(\log(B) + dA)= \Oh(d A)$ by \obsref{repetition}. The number of states of $F$ is $O(d B)$. This satisfies the claimed size bounds. Note that the reduction can be implemented in linear time in the output size.

\paragraph{Strengthening the Statement}
In the remainder, we verify that our construction proves the desired lower bound even restricted to instances with $N = \Theta(n^{\alpha_N})$ and $q = \Theta(n^{\alpha_q})$ for any $\alpha_N > 1$ and $\alpha_q > 0$. 
Note that the number of states, the size of the SLP, and the text length can all three be increased by easy padding. E.g., to increase the text length we introduce a garbage symbol ``$\natural$'' that can be read at any state of the automaton, not changing the current state, and add a suitable number of copies of ``$\natural$'' to the text.

We now set $\beta := \min\{\alpha_q, \alpha_N - 1\}$ and only consider \OV instances with $B = \Theta(A^\beta)$. Note that the \OV conjecture asserts a lower bound of $A^{1+\beta-o(1)}$ in this setting.
Note that the above construction yields $n = O(d A) = O(d^{\max\{1,1/\beta\}} A)$, and we can pad to equality. Moreover, we have $q = O(d B) = O(d^{\max\{\beta, 1\}} B) = O(d^{\max\{\beta, 1\}} A^\beta) = O(n^{\beta}) = O(n^{\alpha_q})$, since $\beta \le \alpha_q$, and we can pad to equality to obtain $q = \Theta(n^{\alpha_q})$. Similarly, we have $N = O(d A B) = O(d^{1+\max\{\beta,1\}} A^{1+\beta}) = O(n^{1+\beta}) = O(n^{\alpha_N})$, since $\beta \le \alpha_N - 1$, and we can pad to equality to obtain $N = \Theta(n^{\alpha_N})$. 
Finally, an $O(\min\{nq, N\}^{1-\eps})$-time algorithm for \DFAAccept restricted to $N = \Theta(n^{\alpha_N})$ and $q = \Theta(n^{\alpha_q})$
would imply an algorithm for \OV in time $O(\min\{n^{1+\alpha_q}, n^{\alpha_N} \}^{1-\eps}) = O(n^{(1+\beta)(1-\eps)}) = O((d^{\max\{1,1/\beta\}} A)^{(1+\beta)(1-\eps)}) = O(d^{\max\{1+\beta,1+1/\beta\}} A^{(1+\beta)(1-\eps)}) = O(A^{1+\beta-\eps} \poly(d))$, contradicting the \OV conjecture. This finishes the proof.
\end{proof}

% !TEX root = main.tex

\newcommand{\SLP}{SLP\xspace}

\newcommand{\grammar}{\ensuremath{\mathbb{A}}\xspace }

\newcommand{\Test}{\mathrm{Test}}
\newcommand{\List}{\mathrm{List}}

\newcommand{\s}[1]{S(#1)}
\newcommand{\sBig}[1]{S\Big(#1\Big)}
\newcommand{\sz}{\s{0}}
\newcommand{\so}{\s{1}}

\newcommand{\tuplify}{\mathrm{tuplify}}

\newcommand{\wildcard}{*}

\newcommand{\LB}{\mathrm{LB}}
\newcommand{\inn}{\mathrm{inner}}
\newcommand{\out}{\mathrm{outer}}

\newcommand{\gtwotuple}{(j_1,\dots,j_{k_2})}

\newcommand{\cSZ}{{{\mathcal S}(0)}}
\newcommand{\cSO}{{{\mathcal S}(1)}}

\newcommand{\cV}{\ensuremath{\mathcal V}\xspace}

\newcommand{\dH}{d_H}
\newcommand{\PM}{\mathrm{PM}}
\newcommand{\HD}{\mathrm{HD}}

\newcommand{\cost}{\mathrm{cost}}
\newcommand{\mismatch}{\mathrm{mismatch}}

\subsection{Approximate Pattern Matching and Substring Hamming Distance} \label{sec:genpatternmatch}

We study the following generalization of pattern matching.

\begin{problem}[Generalized Pattern Matching]
Given a text $T$ of length $N$ by an SLP \ctext of size~$n$, a pattern $P$ of length $M$ by an SLP $\cP$ of size $m$, both over some alphabet $\Sigma$, and given a cost function $\cost \colon \Sigma \times \Sigma \to \mathbb{N}$, compute $\min_{0 \le i \le N-M} \sum_{j=1}^M \cost(P[j],T[i+j])$, i.e., the minimum total cost of any alignment.
\end{problem}

In other words, we want to compute the length-$M$ substring $T'$ of $T$ minimizing the total cost of aligned symbols in $P$ and $T'$.
This problem has two important special cases: (1) We obtain Substring Hamming Distance when $\cost(\sigma,\sigma') = [\sigma \ne \sigma']$ for any $\sigma, \sigma' \in \Sigma$.
(2) We obtain Pattern Matching with Wildcards when $T$ is over alphabet $\Sigma$ and $P$ is over alphabet $\Sigma \cup \{*\}$, we have $\cost(*,\sigma) = 0$ for any $\sigma \in \Sigma$ and $\cost(\sigma,\sigma') = [\sigma \ne \sigma']$ for any $\sigma, \sigma' \in \Sigma$, and the task is to decide whether the minimum total cost of any alignment is 0. 

\begin{problem}[Substring Hamming Distance]
Given a text $T$ of length $N$ by an SLP \ctext of size~$n$ and a pattern $P$ of length $M$ by an SLP $\cP$ of size $m$, both over some alphabet $\Sigma$, compute $\min_{0 \le i \le N-M} \sum_{j=1}^M \big[P[j] \ne T[i+j] \big]$, i.e., the minimum Hamming distance of any alignment.
\end{problem}

\begin{problem}[Pattern Matching with Wildcards]
For some alphabet $\Sigma$, we are given a text $T$ of length $N$ by an SLP \ctext of size~$n$ over alphabet $\Sigma$ and a pattern $P$ of length $M$ by an SLP $\cP$ of size $m$ over alphabet $\Sigma \cup \{*\}$. We say that $\sigma' \in \Sigma \cup \{*\}$ and $\sigma \in \Sigma$ \emph{match} if $\sigma' = *$ or $\sigma' = \sigma$. Decide whether for some offset $0 \le i \le N-M$ all pairs $P[j],T[i+j]$ match for $1 \le j \le M$.
\end{problem}

In this section, for all three problems we show an upper bound of $O(\min\{ |\Sigma| N \log N, n M \})$ and a SETH-based lower bound of $\min\{N, n M\}^{1-o(1)}$. This yields a tight bound in case of constant alphabet size, as the lower bound constructs constant-alphabet strings. We leave it as an open problem to get tight bounds for larger alphabet size.

Note that it suffices to prove the upper bound for Generalized Pattern Matching and the lower bound for the special cases Substring Hamming Distance and Pattern Matching with Wildcards.
We start with the following two upper bounds, which follow standard arguments.

\begin{lem}
Generalized Pattern Matching can be solved in time $O(|\Sigma| N \log N)$.
\end{lem}
\begin{proof}
Decompress both the text $T$ and the pattern $P$. For each symbol $\sigma \in \Sigma$, build the
vector $v^\sigma \in \mathbb{R}^N$ with $v_i^\sigma := \cost(\sigma,T[i])$ and the vector $u^\sigma \in \{0,1\}^M$ with $u_j^\sigma := [P[j] = \sigma]$. Compute their convolution $w^\sigma \in \mathbb{R}^{N-M+1}$ with $w_i^\sigma = \sum_{j=1}^M u_j^\sigma v_{i+j}^\sigma$. Using FFT, $w^\sigma$ can be computed in time $O(N \log N)$. 
Finally, compute the vector $r \in \mathbb{R}^{N-M+1}$ with $r_i = \sum_{\sigma \in \Sigma} w_i^\sigma$ and return the minimal entry of $r$. Note that 
\[ r_i = \sum_{\sigma \in \Sigma} w_i^\sigma = \sum_{\sigma \in \Sigma} \sum_{j=1}^M u_j^\sigma v_{i+j}^\sigma = \sum_{\sigma \in \Sigma} \sum_{j=1}^M [P[j] = \sigma] \cdot \cost(\sigma,T[i+j]) = \sum_{j=1}^M \cost(P[j],T[i+j]), \]
which proves correctness. 
\end{proof}

\newcommand{\Mat}{\textup{Match}}
\newcommand{\Matl}{\textup{Match}_\textup{left}}
\newcommand{\Matr}{\textup{Match}_\textup{right}}
\newcommand{\Matf}{\textup{FixMatch}}

\begin{lem}
Generalized Pattern Matching can be solved in time $O(n M)$.
\end{lem}
\begin{proof}[Proof Sketch]
Let $S_1,\ldots,S_n$ be the non-terminals of the SLP $\cT$ that generates the text $T$. In this proof, for simplicity we write $T_i := \eval(S_i)$.
We decompress the pattern $P$. 
For $1 \le i \le n$ we define
\[ \Mat(i) := \min_{0 \le d \le |T_i| - M} \sum_{j=1}^M \cost(P[j], T_i[j+d]), \]
or $\infty$, if $|T_i|<M$.
This solves the Generalized Pattern Matching problem restricted to the substring $T_i$ of $T$.
Clearly, we can solve the given Generalized Pattern Matching instance $(T,P)$ by calling $\Mat(n)$. Moreover, for any offset $d$ and any $i \in [n]$ we define
\[ \Matf(i,d) := \sum_{j: \ \substack{1 \le j+d \le |T_i|,\\1 \le j \le M}} \cost(P[j], T_i[j+d]). \]
In other words, $\Matf(i,d)$ is equal to the total cost between $T_i$ and a shifted pattern~$P$ (by $d$ symbols to the right, or $-d$ symbols to the left), where we consider only the symbols that have an aligned counterpart. 

In the remainder we show how to compute these functions by simple recursive algorithms. We precompute all lengths $|T_i|$ in time $O(n)$. For $\Matf(.,.)$, observe that for a rule $S_i \to S_{\ell} S_{r}$ we have
\[ \Matf(i,d) = \Matf(\ell,d) + \Matf(r,d-|T_{\ell}|), \]
since the offset with respect to the first symbol of $T_r$ differs to the offset with respect to the first symbol of $T_i$ by $|T_\ell|$.
Moreover, for a rule $S_i \to \sigma \in \Sigma$ we can compute $\Matf(i,d)$ in constant time. Note that whenever the offset $d$ is such that no symbols get aligned, we can immediately return $0$. This completes our algorithm for $\Matf(.,.)$.

Now consider $\Mat(i)$. For a rule $S_i \to S_{\ell} S_{r}$, the optimal alignment of the pattern in $T_i$ is either completely contained in $T_{\ell}$ or completely contained in $T_{r}$ or it has a non-empty intersection with both of them, in which case it has an offset $-M < d < 0$ with respect to the starting symbol of~$T_{r}$, or equivalently an offset $|T_\ell|+d$ with respect to the starting symbol of $T_\ell$. Hence, we have
\[ \Mat(i) = \min\Big\{ \Mat(\ell),\, \Mat(r),\, \min_{-M < d < 0} \Matf(r,d) + \Matf(\ell,|T_{\ell}| + d)  \Big\}. \]
Again, for a rule $S_i \to \sigma \in \Sigma$ we can compute $\Mat(i)$ in constant time. This completes the algorithm for $\Mat(.)$. 

To obtain the claimed running time, we use memoization to ensure that each argument is called at most once. Clearly, there are $n$ possible arguments for $\Mat(.)$, and each call takes time $O(M)$, resulting in time $O(nM)$. 
Note that $\Mat(.)$ only calls $\Matf(i,d)$ for offsets $d$ such that the pattern crosses the left or right boundary of $T_i$. This property also holds as an invariant in the recursive subproblems of $\Matf(i,d)$. Hence, there are less than $2M$ possible offsets $d$ (i.e., less than $M$ offsets for the left and right boundary). As there are $n$ possible values for $i$, and each call to $\Matf(.,.)$ takes time $O(1)$, we obtain the claimed total running time of $O(nM)$.
\end{proof}

This completes the upper bound $O(\min\{ |\Sigma| N \log N, n M \})$ for Generalized Pattern Matching. It remains to prove the SETH-based lower bound of $\min\{N, n M\}^{1-o(1)}$ for Substring Hamming Distance and Pattern Matching with Wildcards.

We now make the intuition given at the beginning of Section~\ref{sec:sethlowerbounds} formal, by designing a text $T$ that enumerates all combinations of $k$ vectors in a given \kOV instance, while still being well compressible.
We give a slightly more general construction that will also be useful later for our SETH-based lower bounds for LCS, see \secref{lcs}. As usual, we consider $k$ as a constant.

\begin{lem}\label{lem:tuplify}
Consider a \kOV instance $\sA = \{a_1,\ldots,a_A\} \subseteq \{0,1\}^d$. Let $b \in \{0,1\}^d$ be an additional vector, and let $\sz,\so$ be strings of length $\gamma$ ($S(i)$ is a sequence that represents an entry that is equal to $i$).
We define the tuplified representation as follows:
\begin{align*}
V & = \tuplify\big(\sA, k, b, \sz, \so\big) \\
& := \bigconcat_{\ell=1}^d \bigconcat_{i_1,\dots,i_k \in [A]} \sBig{ b[\ell] \cdot a_{i_1}[\ell] \cdots a_{i_k}[\ell]},
\end{align*}
where the second $\bigconcat$ goes over all tuples $(i_1,\ldots,i_k) \in [A]^k$ in lexicographic order.
This representation satisfies the following properties.
\begin{enumerate}
\item We can compute, in linear time in the output size, an \SLP \cV generating $V$ of size $\Oh(d A + \gamma)$ or, when given SLPs $\cSZ,\cSO$ generating $S(0),S(1)$, of size $\Oh(d A + |\cSZ|+|\cSO|)$. 
\item Write $V = \bigconcat_{i=1}^{dA^k} V_i$ with $V_i \in \{\sz,\so\}$. Then there exist $i_1,\ldots,i_k \in [A]$ such that $(b, a_{i_1}, \dots, a_{i_k})$ is orthogonal if and only if there is an \emph{offset} $1 \le \Delta \le A^k$ such that 
\[V_\Delta =  V_{\Delta+A^k} = \ldots = V_{\Delta+(d-1)A^k} = \sz.\]
%\[v[\Delta \gamma + 1, (\Delta + 1) \gamma - 1] = v[\Delta+\ell^g+1, \Delta+\ell^g + \gamma-1] = \cdots = v[\Delta+(g-1)\ell^g, \Delta+(g-1)\ell^g +\gamma-1] = \sz.\]
\end{enumerate}
\end{lem}
\begin{proof}
For the second property, note that by definition $V_{\Delta}, V_{\Delta+A^k}, \dots, V_{\Delta+(d-1)A^k}$ are all equal to $\sz$ for $\Delta \in [A^k]$ if and only if the $\Delta$-th tuple $(i_1,\dots,i_k)\in [A]^k$ in the lexicographic ordering of $[A]^k$ satisfies 
\[ b[\ell] \cdot a_{i_1}[\ell] \cdots a_{i_k}[\ell] = 0 \qquad \text{for all } \ell\in [d].\] 
This condition is equivalent to $(b, a_{i_1}, \dots, a_{i_k})$ being an orthogonal pair, so the claim follows.

It remains to construct a short SLP $\cV$ generating $V$. We construct non-terminals $P_{\sz},P_{\so}$ with $\eval(P_{\s{i}}) = \s{i}$ by an SLP of size $\gamma_S = O(\gamma)$ as in Observation~\ref{obs:repetition}, or of size $\gamma_S = O(|\cSZ|+|\cSO|)$ by using given SLPs $\cSZ,\cSO$. We can extend this, using \obsref{repetition}, to a slightly larger SLP of size $\Oh(\log A + \gamma_S)$ that includes, for every $1 \le j \le k$, a non-terminal $P_{\sz}^{j}$ with $\eval(P_{\sz}^j) = \sz^{A^j}$.

The crucial observation is the following: for any tuple $(i_1,\dots, i_k)\in [A]^k$, let $p_\ell(i_1,\dots,i_k) = a_{i_1}[\ell] \cdots a_{i_k}[\ell]$. Then for any $\ell \in [d], j\in [k]$ and $(i_1,\dots,i_j)\in [A]^j$, we have that $a_{i_j}[\ell] = 0$ implies $p_\ell(i_1,\dots,i_j,i'_{j+1},\dots,i'_{k}) = 0$ for \emph{all} $(i'_{j+1},\dots,i'_{k})\in [A]^{k-j}$.
We now define the final SLP using the starting non-terminal $S_0$ and the following productions
\begin{align*}
S_0 & \productionsign \Test_1 \dots \Test_d  & & \\
\Test_\ell & \productionsign
\begin{cases} 
  P_{\sz}^k & \text{if } b[\ell] = 0 \\
  \List_\ell^{(1)} & \text{otherwise} 
\end{cases}
& & \ell \in [d],\\
\List^{(j)}_\ell & \productionsign \bigconcat_{i \in [A]}
\begin{cases}
  P_{\sz}^{k-j} & \text{if } a_i[\ell] = 0, \\
  \List^{(j+1)}_\ell & \text{otherwise} 
\end{cases}
& & \ell\in[d],j\in[k], \\
\List^{(k+1)}_\ell & \productionsign P_{\so}.
\end{align*}
It is straight-forward to verify that $\eval(S_0) = V$.
Note that the size of this SLP, i.e., the total number of non-terminals on the right hand side of the above rules, is bounded by $\Oh(\gamma_S + dA)$. Moreover, the SLP can be constructed in linear time in its size.
\end{proof}

After this preparation, we can prove our conditional lower bounds.
%After this preparation, we can prove the lower bound of Theorem~\ref{thm:patmatchsubstringHDlb} for Pattern Matching with Wildcards.

\begin{thm}\label{thm:patmatchlb}
Assuming the $k$-OV conjecture, Pattern Matching with Wildcards over alphabet $\{0,1\}$ (plus wildcards $*$) takes time $\min\{N, nM\}^{1-o(1)}$. This holds even restricted to instances with $n = \Theta(N^{\alpha_n})$, $M = \Theta(N^{\alpha_M})$ and $m = \Theta(N^{\alpha_m})$ for any $0 < \alpha_n < 1$ and $0 < \alpha_m \le \alpha_M \le 1$.
\end{thm}

Before we prove \thmref{patmatchlb}, let us sketch the main idea by providing a simple $N^{1-o(1)}$-time conditional lower bound in the setting $n,m= O(N^{\varepsilon})$ and $N = \Theta(M)$. 
Let $\sA \subseteq \{0,1\}^d$ of size $A$ be an arbitrary $k$-\OV instance with $k > 1/\varepsilon$, and assume for simplicity $d \le A^{o(1)}$. Using \lemref{tuplify} on $\sA$, $k$, $\sz = 0, \so = 1$ and $b = (1,\dots,1) \in \{0,1\}^d$, we compute an SLP \ctext for
\[T = \tuplify(\sA,k,b,\sz,\so).\]
We define the pattern $P$ as
\[P = 0 (\wildcard^{A^k-1} 0)^{d-1}. \]
Note that Pattern Matching with Wildcards on instance $T,P$ checks whether for some offset $\Delta$ we have $T[\Delta] = T[\Delta + A^k] = \ldots = T[\Delta + (d-1) A^k] = 0$. Hence, by \lemref{tuplify}, pattern $P$ matches $T$ if and only if there is an orthogonal tuple $(a_1, \dots, a_k) \in \sA^k$, showing correctness of the reduction.

Note that we have $N = \Theta(M) = \Theta(d A^k)$. By \lemref{tuplify}, $T$ has an SLP of size $O(dA)$, and by \obsref{repetition}, $P$ has an SLP of size $O(d \log A)$. By $d \le A^{o(1)}$ and $k > 1/\eps$, we are indeed in the setting $n,m = O(N^\eps)$ and $N = \Theta(M)$. An $O(N^{1-\eps})$ algorithm for Pattern Matching with Wildcards would now imply an $O(A^{k(1-\eps)} \textup{poly}(d))$ for $k$-OV, contradicting the $k$-OV conjecture.

We now give the slightly more involved general construction.

\begin{proof}[Proof of \thmref{patmatchlb}]
For $k \ge 2$, let $\sA = \{a_1,\ldots,a_A\}$ be a $k$-\OV instance in $d$ dimensions, and let $k_1,k_2 \ge 1$ with $k_1 + k_2 = k$. We will construct an equivalent instance of Pattern Matching with Wildcards with $N = O(d A^k)$, $M = O(d A^{k_1})$, $n = O(d A^{k_2+1})$, and $m = O(d \log A)$. Any $O(\min\{N, n M\}^{1-\eps})$ algorithm for Pattern Matching with Wildcards would then imply an algorithm for $k$-OV in time $O(A^{(k+1)(1-\eps)} \textup{poly}(d)) = O(A^{k(1-\eps/2)} \textup{poly}(d))$ for $k \ge 2/\eps$, contradicting the $k$-OV conjecture. Below we strengthen this statement to hold restricted to instances with $n = \Theta(N^{\alpha_n})$, $M = \Theta(N^{\alpha_M})$ and $m = \Theta(N^{\alpha_m})$ for any $0 < \alpha_n < 1$ and $0 < \alpha_m \le \alpha_M \le 1$.

To give such a reduction, we define the text as 
\begin{align*}
T = \bigconcat_{(j_1,\dots,j_{k_2}) \in [A]^{k_2}} 1^{A^{k_1}} \concat \tuplify( \sA, k_1, \min(a_{j_1}, \dots, a_{j_{k_2}}), 0, 1), \hspace{0.5cm} 
\end{align*}
where $\min(b_1, \dots, b_\ell)$ denotes the component-wise minimum of $b_1,\dots,b_\ell$.

We define the pattern $P$ as
\[ P = 0 (\wildcard^{A^{k_1} - 1}0)^{d-1}. \]

\paragraph{Correctness}
Observe that $P$ cannot overlap any $1^{A^{k_1}}$-block, since never more than $A^{k_1} - 1$ wildcards are followed by a 0 in $P$. Thus, $P$ matches $T$ if and only if there is a tuple $\gtwotuple\in [A]^{k_2}$ such that $P$ matches $T(\gtwotuple) := \tuplify(\sA, k_1, \min(a_{j_1}, \dots, a_{j_{k_2}}), 0, 1)$. By the structure of the pattern, $P$ matches any string $S$ if and only if there is an offset $\Delta$ such that $S[\Delta]=S[\Delta + A^{k_1}]= \dots = S[\Delta + (d-1)A^{k_1}] = 0$. Thus, by \lemref{tuplify}, $P$ matches $T(\gtwotuple)$ if and only if there are vectors $a_1,\ldots,a_{k_1} \in \sA$ for which $(a_1,\dots,a_{k_1}, \min(a_{j_1},\dots,a_{j_{k_2}}))$ is an orthogonal tuple. The latter condition is equivalent to $(a_1,\dots,a_{k_1}, a_{j_1},\dots,a_{j_{k_2}})$ being an orthogonal tuple. Since $k_1 + k_2 = k$ and $T$ contains $T(\gtwotuple)$ for all $\gtwotuple\in [A]^{k_2}$, this proves that $P$ matches $T$ if and only if there is an orthogonal $k$-tuple in the instance $\sA$.

\paragraph{Size Bounds}
Note that $N = |T| = O(d A^{k})$. By \lemref{tuplify} and \obsref{repetition}, we can compute an \SLP $\ctext$ of size $n = \Oh(d A^{k_2 + 1})$ generating $T$, in linear time.
Similarly, note that $M = |P| = O(d A^{k_1})$. By \obsref{repetition}, we can compute an SLP $\cpat$ of length $m = \Oh(d \log A)$ generating $P$, in linear time. This proves the claimed bounds.

\paragraph{Strengthening the Statement}
We now prove the lower bound restricted to instances with $n = \Theta(N^{\alpha_n})$, $M = \Theta(N^{\alpha_M})$ and $m = \Theta(N^{\alpha_m})$ for any $0 < \alpha_n < 1$ and $0 < \alpha_m \le \alpha_M \le 1$. Let $\eps > 0$ and set $\beta := \min\{1, \alpha_M + \alpha_n\}$. We choose $k_1,k_2 \ge 1$ such that $k_1 + k_2 = k$ and $k_1 \approx \min\{\alpha_M, 1-\alpha_n\} k / \beta$ and $k_2 \approx \alpha_n k / \beta$. Note that $k_1,k_2$ are restricted to be integers, however, for sufficiently large $k$ depending only on $\eps, \alpha_M, \alpha_n$, we can ensure $k_1 \le (1+\eps/4) \min\{\alpha_M, 1-\alpha_n\} k / \beta$ and $k_2+1 \le (1+\eps/4) \alpha_n k / \beta$. Note that for the dimension $d$ we can assume $d \le A$, since otherwise an $O(A^{k-\eps} \textup{poly}(d))$ algorithm clearly exists. In particular, for sufficiently large $k$ we have $d \le A^{(\eps/4) \cdot \min\{\alpha_M, \alpha_m, \alpha_n, 1-\alpha_n \} k / \beta}$.  This yields
\begin{align*}
  N &= O(d A^k) = O(A^{(1+\eps/2)k / \beta}), \\
  M &= O(d A^{k_1}) = O(A^{(1+\eps/2) \min\{\alpha_M, 1-\alpha_n\} k / \beta}) = O(A^{(1+\eps/2) \alpha_M k / \beta}), \\
  n &= \Oh(d A^{k_2 + 1}) = O(A^{(1+\eps/2) \alpha_n k / \beta}), \\
  m &= \Oh(d \log A) = O(A^{(1+\eps/2) \alpha_m k / \beta}).
\end{align*}
Standard padding\footnote{Add a prefix of wildcards to the pattern and a prefix of 1's to the text, and partially decompress the SLPs.} of these four parameters allows us to achieve equality, up to constant factors, in the above inequalities, which yields the desired $n = \Theta(N^{\alpha_n})$, $M = \Theta(N^{\alpha_M})$ and $m = \Theta(N^{\alpha_m})$.
Any $O(\min\{N, n M\}^{1-\eps})$ algorithm for Pattern Matching with Wildcards in this setting would now imply an algorithm for $k$-OV in time $O(\min\{ A^{(1+\eps/2)k/ \beta}, A^{(1+\eps/2)(\alpha_M+\alpha_n) k/ \beta} \}^{1-\eps}) = O(A^{(1+\eps/2)(1-\eps) \min\{1, \alpha_M + \alpha_n\} k / \beta}) = O(A^{(1-\eps/2)k})$, where we used the definition of $\beta$ and $(1+\eps/2)(1-\eps) \le 1-\eps/2$. This contradicts the $k$-OV conjecture, finishing the proof.
\end{proof}

We next prove a lower bound similar to \thmref{patmatchlb} for another special case of generalized pattern matching, namely Substring Hamming Distance. Instead of a direct reduction from $k$-OV, we present a linear-time reduction from Pattern Matching with Wildcards over alphabet $\{0,1\}$ to Substring Hamming Distance.

\begin{thm} \label{thm:subhamminglb}
Assuming the $k$-OV conjecture, Substring Hamming Distance on constant-size alphabet takes time $\min\{N, nM\}^{1-o(1)}$. This holds even restricted to instances with $n = \Theta(N^{\alpha_n})$, $M = \Theta(N^{\alpha_M})$ and $m = \Theta(N^{\alpha_m})$ for any $0 < \alpha_n < 1$ and $0 < \alpha_m \le \alpha_M \le 1$.
\end{thm}

\begin{proof}
For short, we write $\dH(X,Y)$ for the Hamming distance of strings $X,Y$.
We prove the result by reducing any Pattern Matching with Wildcards instance $T_\PM,P_\PM$ over alphabet $\Sigma=\{0,1\}$ to an instance  $T_\HD, P_\HD$ of Substring Hamming Distance. We first define coordinate strings
\begin{align*}
s_T(0) & := 100, & s_T(1) & := 010, &  &\\
s_P(0) & := 101, & s_P(1) & := 011, &  s_P(\wildcard) & := 000. 
\end{align*}
Observe that these strings are defined in such a way that $\dH(s_P(\wildcard), s_T(y)) = 1$ for $y \in \{0,1\}$, $\dH(s_P(x), s_T(y)) = 1$ for $x=y \in \{0,1\}$, and $\dH(s_P(x), s_T(y)) = 3$ if $x\ne y$, $x,y \in \{0,1\}$.

We introduce the \emph{guarding} $G(s):=s \concat 2\, 3\, 4$ for length-3 strings $s\in \{0,1\}^3$. This allows us to reduce $T_\PM$, $P_\PM$ to the following instance, using alphabet $\Sigma = \{0,1,2,3,4\}$,
\begin{align*}
 T_\HD & := G(s_T(T_\PM[1])) \dots G(s_T(T_\PM[N])),\\
 P_\HD & := G(s_P(P_\PM[1])) \dots G(s_P(P_\PM[M])).
\end{align*}
Note that for any $0 \le i \le N-M$,
\begin{align*}
 \dH(T_\HD[6i+1..6i+6M], P_\HD) & = \sum_{j=1}^M \dH(s_P(P_\PM[j]), s_T(T_\PM[i+j])) \\
& = M + 2\cdot \mismatch(T_\PM[i+1..i+M],P_\PM), 
\end{align*}
where $\mismatch(z,z') = \#\{i \mid z'[i]\ne \wildcard, z[i] \ne z'[i]\}$ is the number of mismatches of $z$ and $z'$. 

We now observe that for all $i$ with $i \bmod 6 \ne 0$, we have $\dH(T_\HD[i+1..i+6M], P_\HD) \ge 3M$, as no two symbols $2,3,4$ in $P_\HD$ are aligned, so that each $G(s_P(P_\PM[j]))$ contributes at least 3 to the Hamming distance. 
Since $\dH(T_\HD[6i+1..6i+6M], P_\HD) \le 3M$ for all $i$, the substring with smallest Hamming distance has thus a Hamming distance of $M + 2\cdot \min_{0 \le i \le N-M} \mismatch(T_\PM[i+1.. i+M], P_\PM)$. This value is equal to $M$ if and only if $P_\PM$ matches $T_\PM$, proving correctness.

The corresponding reduction of the compressed problems is straightforward: We can augment the SLP $\ctext_\PM$ for $T_\PM$ by $\Oh(1)$-sized productions to obtain an SLP $\ctext_\HD$ for $T_\HD$, by replacing each terminal $\sigma \in \{0, 1, \wildcard\}$ by a non-terminal evaluating to $G(s_T(\sigma))$. Analogously, we can compute an SLP for $P_\HD$ of size $|\cpat_\HD|=|\cpat_\PM|+\Oh(1)$ in linear time. Overall, since also $|T_\HD| = \Oh(|T_\PM|), |P_\HD|=\Oh(|P_\PM|)$, all parameters are preserved up to constant factors.
By this linear-time parameter-preserving reduction, the lower bound of \thmref{patmatchlb} translates to Substring Hamming Distance, yielding the claim.
\end{proof}

% !TEX root = main.tex

\newcommand{\cX}{{\cal X}}
\newcommand{\cY}{{\cal Y}}

\newcommand{\cS}{{\cal S}}

\newcommand{\guard}{\mathrm{G}}
\newcommand{\delmu}{{\mathrm{del-}\mu}}

\newcommand{\matching}{{\cal M}}
\newcommand{\tX}{{\tilde{X}}}
\newcommand{\tY}{{\tilde{Y}}}
\newcommand{\tell}{{\tilde{\ell}}}
\newcommand{\tdelta}{{\tilde{\delta}}}

\newcommand{\sC}{{\cal C}}

\newcommand{\inputs}{{\cal I}}
\newcommand{\type}{\mathrm{type}}
\newcommand{\algn}{{\mathbf{\Lambda}}}
\newcommand{\dist}{\delta}

\newcommand{\GA}{{\mathrm{GA}}}
\newcommand{\TG}{{\mathrm{TG}}}
\newcommand{\cTG}{{\cal TG}}
\newcommand{\NTG}{{\mathrm{NTG}}}
\newcommand{\cNTG}{{\cal NTG}}

\newcommand{\zx}{\mathbf{0}_{\x}}
\newcommand{\zy}{\mathbf{0}_{\y}}
\newcommand{\ox}{\mathbf{1}_{\x}}
\newcommand{\oy}{\mathbf{1}_{\y}}

\newcommand{\tuple}{{\boldsymbol{i}}}
\newcommand{\norm}{{\mathrm{norm}}}
\newcommand{\orth}{{\mathrm{orth}}}
\newcommand{\non}{{\mathrm{non}}}

\subsection{Longest Common Subsequence} \label{sec:lcs}

In this section, we study the Longest Common Subsequence (LCS) problem. Recall that a string $S$ of length $\ell$ is a substring of a string $X$ if there are  $1\le i_1 < \dots < i_{\ell} \le |X|$ with $S[j] = X[i_j]$ for any $j\in [\ell]$. In the LCS problem, given two strings $X,Y$, the task is to determine the longest string $S$ that is a subsequence of both $X$ and $Y$. We denote the length of the LCS by $L(X,Y) = |S|$, and more precisely consider the problem of computing $L(X,Y)$. In the whole section, the alphabet $\Sigma$ has constant size.

\begin{problem}[\LCS]
Given strings $X,Y$ of length at most $N$ by grammar-compressed representations $\cX,\cY$ of size at most $n$, compute the length of the LCS of $X$ and $Y$.
\end{problem}

As discussed in the introduction, the $O(nN \sqrt{\log{N/n}})$ time algorithm by Gawrychowski \cite{gawrychowski2012faster} is the fastest known. Here we prove a matching lower bound of $(Nn)^{1-o(1)}$, assuming the $k$-OV conjecture.

\begin{thm}\label{thm:lcslb}
Assuming the $k$-OV conjecture, there is no $(nN)^{1-\eps}$-time algorithm for LCS for any $\eps>0$. This even holds restricted to instances with $n = \Theta(N^{\alpha_n})$ for any $0 < \alpha_n < 1$, and an alphabet of constant size.
\end{thm}

The general approach is very similar to the lower bound for Pattern Matching with Wildcards given in \secref{genpatternmatch}. In particular, we again use the tuplified representation $T= \bigconcat_{i=1}^{d A^k} T[i]$ of \lemref{tuplify} for a $k$-\OV instance $\sA$. Recall that this allows us to decide the $k$-OV instance by testing whether there is a subsequence of $d$ substrings $T[\Delta], T[\Delta+A^k], \dots, T[\Delta+(d-1)A^k]$ all equal to a certain 0-coordinate string. Finding a pattern to test this was quite simple for Pattern Matching with Wildcards, yielding an $N^{1-o(1)}$ lower bound. For LCS, enforcing a coherent offset is much more complicated, since the ``pattern'' is matched as a subsequence not as a substring. Furthermore, the extension to a $(nN)^{1-o(1)}$ lower bound is more involved and relies on the quadratic-time nature of LCS. Fortunately, we can overcome the technical obstacles for LCS using (an extension of) alignment gadgets developed in~\cite{BK15}. We first redevelop and extend the corresponding alignment gadget tools in \secref{alignmentgadget}, then give the lower bound for compressed instances for general distance measures in \secref{distlb} and then finish our LCS lower bound by designing an alignment gadget for LCS in~\secref{lcslb}.

\subsubsection{Alignment Gadget Framework}
\label{sec:alignmentgadget}

We start by reviewing and adapting the definitions of~\cite{BK15}. In particular, we extend the alignment gadget definition for our purposes.

More generally than LCS, we consider an arbitrary \emph{similarity measure} $\delta : \inputs \times \inputs \to \N$. For LCS, the set of inputs $\inputs$ is the set of all strings over some sufficiently large constant-sized alphabet $\Sigma$, and $\delta(X,Y) := |X|+|Y| - 2 L(X,Y)$, where $L(X,Y)$ is the length of the LCS of $X$ and $Y$.

Any sequence $X \in \inputs$ is assigned an (abstract) type $\type(X)$. For LCS, we use $\type(X) := (|X|, \Sigma)$, where $|X|$ is the length of $X$ and $\Sigma$ the alphabet over which $X$ is defined. We define $\inputs_{t}:=\{X \in \inputs\mid\type(X)=t\}$ as the set of all inputs of type $t$. 

\paragraph{Alignments}
Let $n \ge m$. An \emph{alignment} is a set $\Lambda = \{(i_1,j_1),\ldots,(i_k,j_k)\}$ with $0 \le k \le m$ such that $1 \le i_1 < \ldots < i_k \le n$ and $1 \le j_1 < \ldots < j_k \le m$. We say that $(i,j) \in \Lambda$ are \emph{aligned}. Any $i \in [n]$ or $j \in [m]$ that is not contained in any pair in $\Lambda$ is called \emph{unaligned}.
We denote the set of all alignments (with respect to $n,m$) by~$\algn_{n,m}$.

We call the alignment $\{(\Delta+1,1),\ldots,(\Delta+m,m)\}$, with $0 \le \Delta \le n-m$, a \emph{structured alignment}. We denote the set of all structured alignments by $\strc_{n,m}$.

Defining the \emph{cost} of an alignment $\Lambda \in \algn_{n,m}$, we deviate from~\cite{BK15}:  for any $X_1,\ldots,X_n \in \inputs$ and $Y_1,\ldots,Y_m \in \inputs$, we define the cost of $\Lambda = \{(i_1,j_1),\dots, (i_{|\Lambda|},j_{|\Lambda|})\}$ as
\begin{align*}
%\\cost(\Lambda) = \cost^{x_1,\ldots,x_n}_{y_1,\ldots,y_m}(\Lambda) & := \sum_{k=1}^{|\Lambda|} \dist(x_{i_k},y_{j_k}) + (m-|\Lambda|) \gamma +  \sum_{k=1}^{|\Lambda|-1} (i_{k+1}-i_k -1) \gamma. \\
\cost(\Lambda) = \cost^{X_1,\ldots,X_n}_{Y_1,\ldots,Y_m}(\Lambda) & := \sum_{k=1}^{|\Lambda|} \dist(X_{i_k},Y_{j_k}) + \begin{cases} (m-|\Lambda|) \gamma, & \text{if } |\Lambda| < m \\ (i_{m} - i_1 - m+1) \gamma & \text{if } |\Lambda| = m,\end{cases}
%& = \sum_{k=1}^{|\Lambda|} \dist(x_{i_k},y_{j_k}) +  (m-|\Lambda|)\gamma + (i_{|\Lambda|} - i_1 - |\Lambda|+1) \gamma 
\end{align*}
where we set $\gamma:= \max_{i,j} \dist(X_i,Y_j)$.
In other words, (1) for any $j \in [m]$ which is aligned to some $i$, we ``pay'' the distance $\dist(X_i,Y_j)$, (2) if $\Lambda$ is unstructured because it contains an unaligned $j$, we ``pay'' a penalty of $\gamma$ for each such unaligned $j$ (note that there are $m-|\Lambda|$ unaligned $j \in [m]$) and (3) if $\Lambda$ is unstructured because it aligns all $j$ but leaves out some $i$ between the first and last aligned $i$, then for any unaligned $i$ that is between the first aligned $i_1$ and last aligned $i_{|\Lambda|}$, we also ``pay'' a penalty of $\gamma$ (note that $\sum_{k=1}^{|\Lambda|-1} (i_{k+1}-i_k-1) = i_{|\Lambda|} - i_1 - |\Lambda|+1$). This means that we incur punishment for \emph{any} deviation from a structured alignment.

In~\cite{BK15}, the cost of an alignment was defined to be the smaller quantity
$\sum_{k=1}^{|\Lambda|} \dist(X_{i_k},Y_{j_k}) + (m-|\Lambda|) \gamma$, i.e., unstructured alignments (that still align all $j\in [m]$) were punished less. For structured alignments both definitions coincide. Hence, the following extended alignment gadget is more powerful than the alignment gadget defined in~\cite{BK15}.

\begin{defn}[Extended alignment gadget]
\label{def:algngadget}
The similarity measure $\dist$ admits an
\emph{extended alignment gadget, }if the following
conditions hold: 
given instances $X_{1},\dots,X_{n}\in\inputs_{t_{\x}}$, $Y_{1},\dots,Y_{m}\in\inputs_{t_{\y}}$ with $m\le n$ and types $t_\x=(\ell_\x,\Sigma),t_\y=(\ell_\y,\Sigma)$, we can
construct new instances
$X=\GA_{\x}^{m,t_{\y}}(X_{1},\dots,X_{n})$ and $Y=\GA_{\y}^{n,t_{\x}}(Y_{1},\dots,Y_{m})$
and $C\in\mathbb{Z}$ such that 
\begin{align} \label{eq:algngadget}
\min_{\Lambda \in \algn_{n,m}} \cost(\Lambda) \le \dist(X,Y) - C \le \min_{\Lambda \in \strc_{n,m}} \cost(\Lambda).
\end{align}
Moreover, $\type(X)$, $\type(Y)$ and $C$ only depend on $n,m,t_\x,t_\y$. Finally, $|X|,|Y|=\Theta((n+m)(\ell_\x+\ell_\y))$.
\end{defn}

\begin{defn}[Compressible alignment gadget] \label{def:algngcompr}
We call an extended alignment gadget \emph{compressible}, if $X$ and $Y$ are of the form $X= X_L \left(\bigconcat_{i=1}^n \pad_\x(X_i)\right) X_R$ and $Y= Y_L \left(\bigconcat_{j=1}^m \pad_\y(Y_j)\right) Y_R$ for some strings $X_L,X_R, Y_L, Y_R$ and functions $\pad_\x: \inputs_{t_\x}\to \inputs$ and $\pad_\y: \inputs_{t_\y} \to \inputs$ that satisfy the following properties:
\begin{enumerate}
\item $X_L, X_R, Y_L, Y_R$ have SLPs of size $\Oh(\log n +\log(\ell_\x+\ell_\y))$, computable in linear time in the output.
\item Given SLPs $\cX_i, \cY_j$ for $X_i,Y_j$, we can compute SLPs for $\pad_\x(X_i),\pad_\y(Y_j)$ of size $\Oh(|\cX_i|+\log (\ell_\x+\ell_\y)), \Oh(|\cY_j|+\log(\ell_\x+\ell_\y))$ in linear time in the output.
\end{enumerate} 
\end{defn}

In \secref{lcslb}, we provide a compressible extended alignment gadget for LCS.

At the lowest level of our construction, we need the following notion.

\begin{defn}\label{def:coordValues}
The similarity measure $\dist$ admits \emph{coordinate values}, if there exist $\zx,\zy,\ox,\oy\in\inputs$ satisfying
\[
\dist(\ox,\oy)>\dist(\zx,\oy)=\dist(\zx,\zy)=\dist(\ox,\zy),
\]
and, moreover, $\type(\zx)=\type(\ox)$ and $\type(\zy)=\type(\oy)$.
\end{defn}

\newcommand{\ba}{{\boldsymbol{a}}}
\newcommand{\bb}{{\boldsymbol{b}}}
\newcommand{\bc}{{\boldsymbol{c}}}
\newcommand{\bA}{{\boldsymbol{A}}}
\newcommand{\bB}{{\boldsymbol{B}}}
\newcommand{\bC}{{\boldsymbol{C}}}

\subsubsection{General Lower Bound}
\label{sec:distlb}

The following theorem proves a conditional lower bound of $(Nn)^{1-o(1)}$ for any similarity measure admitting a compressible extended alignment gadget and coordinate values.
\begin{thm}\label{thm:simlb}
Let $\delta$ be a similarity measure admitting a compressible extended alignment gadget and coordinate values. Then unless the $k$-OV conjecture fails, there is no $(nN)^{1-o(1)}$-time algorithm for computing the value $\delta(X,Y)$, given SLPs $\cX, \cY$ of size at most $n$ generating strings $X,Y$ of length at most $N$. This even holds restricted to instances with $n = \Theta(N^{\alpha_n})$ for any $0 < \alpha_n < 1$, and constant alphabet size.
\end{thm}

We prove this theorem in the remainder of this section.

Let $\sA = \{a_1,\ldots,a_A\}$ be a $k$-\OV instance in $d-1$ dimensions. 
We augment all vectors in $\sA$ by another dimension where all vectors are 0 to obtain $\sA_0$, or where all vectors are 1 to obtain $\sA_1$.
For any $k' \ge 1$ we let $\sA^{(k')} := \{ \min(a_{i_1},\ldots,a_{i_{k'}}) \mid i_1,\ldots,i_{k'} \in [A] \}$, i.e., for each $k'$-tuple of vectors in $\sA$ the set $\sA^{(k')}$ contains the pointwise minimum of this $k'$-tuple. 
Note that $\sA^{(k')}$ is in general a multiset, it has size $|\sA^{(k')}| = A^{k'}$, and is naturally ordered by the lexicographic ordering on $k'$-tuples $(i_1,\ldots,i_{k'}) \in [A]^{k'}$.
Similarly, we define $\sA_0^{(k')}$ and $\sA_1^{(k')}$ for the augmented vectors. 
We split $k = k_1 + 2 k_2$ for some $k_1, k_2 \ge 1$ and set 
\[ \bA := \sA_0^{(k_1)}, \quad \bB = \sA_0^{(k_2)}, \quad \bC := \sA_1^{(k_2)}. \] 
Observe that deciding the given $k$-OV instance is equivalent to testing whether there are orthogonal vectors $(\ba,\bb,\bc)$ with $\ba \in \bA$, $\bb \in \bB$ and $\bc \in \bC$. In particular, the additional dimension is irrelevant for orthogonality, since we choose at least one vector in $\bA$ and any such vector has the last coordinate equal to 0.
For any $\ell \in [A^{k_1}]$, we denote by $\ba(\ell)$ the $\ell$-th vector in $\bA$.

\paragraph{Tuple gadgets.}
For any $\bb \in \bB$, $\bc \in \bC$, we define vectors $u_\bb \in \{0,1\}^{d A^{k_1}}$ and $v_\bc \in \{0,1\}^{(d-1)A^{k_1} +1}$:
\begin{align*}
u_\bb & := (\ba(1)[1]\cdot \bb[1], \dots, \ba(A^{k_1})[1]\cdot \bb[1], \dots, \ba(1)[d]\cdot \bb[d], \dots, \ba(A^{k_1})[d]\cdot \bb[d])\\
v_\bc & := (\bc[1], \underbrace{0, \dots, 0}_{A^{k_1} - 1 \text{ times}}, \bc[2], \; \dots, \underbrace{0, \dots, 0}_{A^{k_1} - 1 \text{ times}}, \bc[d]).
\end{align*}
In other words, for $j \in [d]$ and $\ell \in [A^{k_1}]$ we have $(u_\bb)_{j \cdot d + \ell} = \ba(\ell)[j+1] \cdot \bb[j+1]$ as well as $(v_\bc)_{j \cdot d + \ell} = \bc[j+1]$ if $\ell=1$ and $(v_\bc)_{j \cdot d + \ell} = 0$ otherwise.

The key idea is as follows. Consider a structured alignment $\Lambda = \{(\Delta+1, 1), \dots, (\Delta+m, m)\}\in \strc_{n,m}$ for the above vectors, where $n=d A^{k_1}$ and $m= (d-1)A^{k_1} + 1$. This chooses some tuple $\ba(\Delta+1) \in \bA$ and aligns the pairs $(\ba(\Delta+1)[\ell] \cdot \bb[\ell], \bc[\ell])$ for all $\ell\in [d]$, additional to some trivial pairs where the coordinate of $v_\bc$ is 0. This allows us to determine whether $(\ba(\Delta+1), \bb, \bc)$ is orthogonal. 

To formalize this, create $\tilde{u}_\bb$ by replacing each 0- and 1-entry in $u_\bb$ by $\zx$ and $\ox$ (from Definition~\ref{def:coordValues}), and create $\tilde{v}_\bc$ by replacing each 0- and 1-entry in $v_\bc$ by $\zy$ and $\oy$, respectively. Let $t_\x$ and $t_\y$ be the types of $\zx,\ox$ and $\zy,\oy$, respectively. Set $\delta_0 := \dist(\zx,\zy) = \dist(\zx,\oy)=\dist(\ox,\zy)$ and $\delta_1 := \dist(\ox,\oy)$. We define the \emph{tuple gadgets}
\begin{align*}
\TG_\x(\bb) & := \GA^{(d-1) A^{k_1}  + 1, t_\y}(\tilde{u}_\bb),\\
\TG_\y(\bc) & := \GA^{d A^{k_1}, t_\x}(\tilde{v}_\bc).
\end{align*}
Let $t_\x', t_\y'$  denote the types of $\TG_\x(\bb)$, $\TG_\y(\bc)$, and let $C$ be the number obtained from \defref{algngadget} when creating $\TG_\x(\bb)$, $\TG_\y(\bc)$. Note that $t_\x', t_\y'$, and $C$ do not depend on the choice of $\bb \in \bB, \bc \in \bC$.

\begin{claim}\label{cla:TGs}
Let $\bb \in \bB, \bc \in \bC$ and set $n := d A^{k_1}$ and $m := (d-1) A^{k_1} + 1$.
If there exists $\ba \in \bA$ such that $(\ba,\bb,\bc)$ are orthogonal, then $\delta(\TG_\x(\bb),\TG_\y(\bc)) = C + m\cdot  \delta_0$. Otherwise $\delta(\TG_\x(\bb),\TG_\y(\bc)) \ge C + (m-1) \delta_0 +  \delta_1$.
\end{claim}
\begin{proof}
If there is an $\ba \in \bA$ for which $(\ba,\bb,\bc)$ are orthogonal, let $\Delta$ be such that $\ba=\ba(\Delta+1)$, where $\ba(\ell)$ is the $\ell$-th tuple in the lexicographic ordering of $\bA$. The structured alignment $\Lambda = \{(\Delta+1,1),\dots,(\Delta+m,m)\}$ satisfies 
\[ \cost(\Lambda) = \left( \sum_{\ell=1}^d \delta_{\ba(\Delta+1)[\ell] \cdot \bb[\ell] \cdot \bc[\ell]} \right) + (A^{k_1} -1) (d-1) \delta_0 = m \cdot \delta_0. \]
Furthermore, for any $\Lambda \in \algn_{n,m}$, we have $\cost(\Lambda) \ge m\cdot  \delta_0$, since $\cost(\Lambda)$ contains at least $m$ summands of value at least $\min\{\gamma, \min_{i,j}\dist(X_i,Y_j)\} = \min_{i,j} \dist(X_i,Y_j) \ge \delta_0$. Thus $\dist(\TG_\x(\bb),\TG_\y(\bc)) = C + m\cdot  \delta_0$ by \defref{algngadget}.

Otherwise, if no such $\ba$ exists, let $\Lambda \in \algn_{n,m}$ be arbitrary. If $\Lambda= \{(\Delta+1,1),\dots,(\Delta+m,m)\}$ is a structured alignment, then
\[ \cost(\Lambda) = \left( \sum_{\ell=1}^d \delta_{\ba(\Delta+1)[\ell] \cdot \bb[\ell] \cdot \bc[\ell]} \right) + (A^{k_1} -1) (d-1) \delta_0 \ge (m-1) \cdot \delta_0 + \delta_1,\]
since there exists some $\ell\in [d]$ with $\ba(\Delta+1)[\ell] = \bb[\ell] = \bc[\ell] = 1$ which contributes a value of $\dist_1$.

If $\Lambda= \{(i_1,j_1),\dots,(i_{|\Lambda|},j_{|\Lambda|})\}$ is unstructured, then either $|\Lambda| < m$, in which case we have 
\[ \cost(\Lambda) \ge |\Lambda| \cdot \delta_0 + (m-|\Lambda|) \gamma \ge (m-1) \delta_0 + \delta_1,\]
or $|\Lambda|=m$ and $i_{m} -i_1 > m-1$, and thus
\[ \cost(\Lambda) \ge m \delta_0 + (i_m - i_1 - (m-1)) \gamma \ge (m-1) \delta_0 + \delta_1. \]
Thus by \defref{algngadget}, $\dist(\TG_\x(\bb),\TG_\y(\bc)) \ge C + (m-1) \delta_0 +  \delta_1$.
\end{proof}

\paragraph{Normalization.}
As usual in these kinds of reductions, we need a normalization step. We define a normalization sequence as
\[ \TG_\norm = \GA^{(d-1) A^{k_1} + 1, t_\y}( \underbrace{\zx, \dots\ldots, \zx}_{(d-1) A^{k_1} \text{ times}}, \underbrace{\ox, \dots,\ox}_{A^{k_1} \text{ times}} ). \]

\begin{claim}\label{cla:TGnorm}
For any $\bc\in \bC$, we have $\dist(\TG_\norm, \TG_\y(\bc)) = C + (m-1)\delta_0 + \delta_1$.
\end{claim}
\begin{proof}
Let $n = d A^{k_1}$ and $m = (d-1) A^{k_1} + 1$.
Let $\Lambda = \{(\Delta+1, 1), \dots, (\Delta+m,m)\}\in \strc_{n,m}$ be a structured alignment.
Then by construction of $\TG_\norm$ and $\TG_\y(\bc)$, the only pair corresponding to $\ox$ and possibly $\oy$ entries is the pair $(\Delta+m, m)$, since only the last $A^{k_1}$ entries of $\TG_\norm$ are $\ox$, and the only possible $\oy$-entry of $\TG_\y(\bc)$ that could be aligned with one of them is its final entry. Now we use that we constructed the vectors $\bC$ as $\sA_1^{(k_2)}$, i.e., we augmented all vectors by a $d$-th coordinate 1, which implies that the $m$-th entry of $\TG_\y(\bc)$ is indeed $\oy$. Hence, the pair $(\Delta+m,m)$ contributes a distance of $\delta_1$ while all others contribute $\delta_0$. This yields $\cost(\Lambda) = (m-1) \dist_0 + \dist_1$.

Let $\Lambda \in \algn_{n,m} \setminus \strc_{n,m}$ be an unstructured alignment. Then its cost is at least $\cost(\Lambda) \ge (m-1) \delta_0 + \gamma \ge (m-1) \delta_0 + \delta_1$, since it contains at least $m-1$ summands of value $\min\{\gamma, \min_{i,j} \dist(X_i,Y_j)\} \ge \delta_0$ and at least one punishment term $\gamma \ge \delta_1$ for a deviation from a structured assignment. Thus by \defref{algngadget}, we have $\dist(\TG_\norm, \TG_\y(\bc)) = C + (m-1)\delta_0 + \delta_1$.
\end{proof}

We now define for any $\bb \in \bB$, $\bc \in \bC$ the \emph{normalized tuple gadgets}
\begin{align*}
\NTG_\x(\bb) & := \GA^{1, t'_\y}(\TG_\x(\bb),\TG_\norm),\\
\NTG_\y(\bc) & := \GA^{2, t'_\x}(\TG_\y(\bc)).\\
\end{align*}

We let $t_\x'', t_\y''$ denote the resulting types of $\NTG_\x(\bb),\NTG_\y(\bc)$, and $C'$ be the number obtained from \defref{algngadget} when creating $\NTG_\x(\bb),\NTG_\y(\bc)$. Note that $t_\x'', t_\y''$, and $C'$ do not depend on the choice of $\bb \in \bB, \bc \in \bC$. This definitions satisfies the following properties. 

\begin{claim}\label{cla:NTGs}
Let $\bb \in \bB, \bc \in \bC$. If there exists $\ba \in \bA$ such that $(\ba,\bb,\bc)$ are orthogonal, then $\delta(\NTG_\x(\bb),\NTG_\y(\bc)) = \delta_\orth$, otherwise we have $\delta(\NTG_\x(\bb),\NTG_\y(\bc)) = \delta_\non$, where
\begin{align*}
\delta_\orth & := C' + C + ((d-1) A^k + 1) \cdot  \delta_0 ,\\
\delta_\non  & := C' + C + (d-1) A^k \cdot  \delta_0 + \delta_1.
\end{align*}
\end{claim}
\begin{proof}
We check all possible alignments $\Lambda\in \algn_{2,1}$: If $\Lambda = \{(1,1)\}$, then $\cost(\Lambda) = \dist(\TG_\x(\bb),\TG_\y(\bc))$. If $\Lambda = \{(2,1)\}$, we have $\cost(\Lambda) = \dist(\TG_\norm, \TG_\y(\bc))$. For the only unstructured alignment $\Lambda = \emptyset$, we have $\cost(\Lambda) = \gamma \ge \max \{\dist(\TG_\x(\bb),\TG_\y(\bc)), \dist(\TG_\norm, \TG_\y(\bc))\}$. Thus by \defref{algngadget}, we have $\dist(\NTG_\x(\bb),\NTG_\y(\bc)) = C' +  \min\{\dist(\TG_\x(\bb),\TG_\y(\bc)), \dist(\TG_\x(\bb),\TG_\norm)\}$. The claim now follows from Claims~\ref{cla:TGs} and~\ref{cla:TGnorm}.
\end{proof}

\paragraph{Final construction.}
To obtain our final instance, we enumerate all $\bb(1),\ldots,\bb(A^{k_2}) \in \bB$ and $\bc(1),\ldots,\bc(A^{k_2}) \in \bC$ in an arbitrary fashion.
We finally combine their corresponding normalized tuple gadgets by defining
\begin{align*}
X & := \GA^{A^{k_2},t_\y''}(\NTG_\x(\bb(1)), \dots \NTG_\x(\bb(A^{k_2})), \NTG_\y(\bb(1)), \dots \NTG_\y(\bb(A^{k_2}))),\\
Y & := \GA^{2A^{k_2},t_\x''}(\NTG_\y(\bc(1)), \dots \NTG_\y(\bc(A^{k_2}))).
\end{align*}
Let $C''$ be the number obtained from \defref{algngadget} when creating $X$, $Y$. 

\begin{claim}\label{cla:XYcorrect}
We have $\dist(X,Y) \le C'' + (A^{k_2}-1) \delta_\non + \delta_\orth$ if and only if there are $\ba \in \bA, \bb \in \bB, \bc \in \bC$ such that $(\ba,\bb,\bc)$ is orthogonal.
\end{claim}
\begin{proof}
Assume that there exists an orthogonal set of vectors and let $\ba \in \bA$, $\bb(i) \in \bB$, $\bc(j) \in \bC$ be the vectors representing them. Let $n= 2A^{k_2}$ and $m= A^{k_2}$. If $i\ge j$, we consider the structured alignment $\Lambda = \{ (i-j + 1, 1), \dots, (i-j + m, m)\}$. Then $\Lambda$ aligns $\NTG_\x(\bb(i)), \NTG_\y(\bc(j))$, yielding cost $\delta_\orth$ by \claref{NTGs}. Since $\dist(\NTG_\x(\bb),\NTG_\y(\bc))\le \delta_\non$ for any $\bb, \bc$, we conclude that $\cost(\Lambda) \le (m-1) \delta_\non + \delta_\orth$. Similarly, if $i < j$, we define the structured alignment $\Lambda = \{(n+i-j + 1, 1), \dots, (n+i-j + m, m)\}$. Then, again, $\Lambda$ aligns $\NTG_\x(\bb(i)),\NTG_\y(\bc(j)))$. As before, we obtain $\cost(\Lambda) \le (m-1) \delta_\non + \delta_\orth$. Thus, in both cases \defref{algngadget} yields $\dist(X,Y) \le C'' + (m-1) \delta_\non + \delta_\orth$.

To prove the claim, it remains to prove that $\dist(X,Y) \ge C'' + m \delta_\non$ if all choices of vectors are non-orthogonal. Note that for any $\Lambda \in \algn_{n,m}$, $\cost(\Lambda)$ consists of $m$ summands with a value of at least $\min_{i,j} \dist(\NTG_\x(\bb(i)),\NTG_\y(\bc(j))) \ge \delta_\non$. This concludes the claim by \defref{algngadget}.
\end{proof}

It remains to prove bounds on the lengths and compressed sizes of the constructed strings.

\begin{claim}\label{cla:distsizebounds}
The strings $X,Y$ have length $\Oh(d A^{k_1+k_2})$. We can, in linear time in the output size, compute SLPs $\cX, \cY$ for $X,Y$ of size $\Oh(d A^{k_2+1})$. 
\end{claim}
\begin{proof}
We will frequently make use of the compressibility of the alignment gagdet (\defref{algngcompr}).
We start by constructing an SLP $\cTG_\x(\bb)$ for $\TG_\x(\bb)$ for any $\bb\in [A]^{k_2}$.
Note that we can split $\TG_\x(\bb)$ into $\TG_\x(\bb)_L \concat (\bigconcat_{i=1}^{d A^{k_1}} \pad_\x(\tilde{u}_i)) \concat \TG_\x(\bb)_R$.
We can apply \lemref{tuplify} by observing that
\begin{equation}
\label{eq:middlepart}
\bigconcat_{i=1}^{d A^{k_1}} \pad_\x(\tilde{u}_i) = \tuplify(\sA_0, k_1, \bb, \pad_\x(\zx), \pad_\x(\ox)), 
\end{equation}
where $S^\x(0) := \pad_\x(\zx)$ and $S^\x(1):= \pad_\x(\ox)$. 
Since $\zx,\ox$ are of constant size, we can compute SLPs $\cS(0), \cS(1)$ for $S^\x(0), S^\x(1)$ of size $\Oh(1)$ by the compressibility assumption. Thus we can compute an SLP for \eqref{eq:middlepart} of size $\Oh(d A)$. Since the left and right bounding string of $\TG_\x(\bb)$ have SLPs of size $\Oh(\log A)$, we obtain an SLP $\cTG_\x(\bb)$ for $\TG_\x(\bb)$ of size $\Oh(d A)$, while $|\TG_\x(\bb)| = \Theta(d A^{k_1})$.

To compute an SLP $\cTG_\y(\bc)$ for $\TG_\y(\bc)$ for any $\bc\in \bC$, we note that 
\[\TG_\y(\bc) = \TG_\y(\bc)_L \concat \left(\bigconcat_{i=1}^{(d-1) A^{k_1} + 1} \pad_\y(\tilde{v}_i)\right) \concat \TG_\y(\bc)_R,\]
 where 
\begin{equation*}
\bigconcat_{i=1}^{(d-1) A^{k_1} + 1} \pad_\y(\tilde{v}_i) = \left( \bigconcat_{\ell=1}^{d-1} S^\y({\bc[\ell]}) \concat S^\y(0)^{A^{k_1}-1} \right) \concat S^\y(\bc[d]),
\end{equation*}
where $S^\y(0) := \pad_\y(\zy)$ and $S^\y(1):= \pad_\y(\oy)$. This immediately admits an SLP of size $\Oh(d + \log A)$ by \obsref{repetition}. Again, using SLPs of size $\Oh(\log A)$ for $\TG_\y(\bc)_L ,\TG_\y(\bc)_R$, we obtain an SLP $\cTG_\y(\bc)$ for $\TG_\y(\bc)$ of size $\Oh(d+\log A)$, while $|\TG_\y(\bc)| = \Oh(d A^{k_1})$.

In the construction of $\NTG_\x(\bb),\NTG_\y(\bc)$ we use constant $n,m$. Together with the compressibility of the alignment gadget, we obtain SLPs $\cNTG_\x(\bb),\cNTG_\y(\bc)$ for $\NTG_\x(\bb),\NTG_\y(\bc)$ of size $\Oh(|\cTG_\x(\bb)|),\Oh(|\cTG_\y(\bc)|)$. Furthermore, $|\NTG_\x(\bb)|=\Theta(|\TG_\x(\bb)|), |\NTG_\y(\bc)|= \Theta(|\TG_\y(\bc)|)$. 

Finally, to obtain SLPs $\cX,\cY$ for $X,Y$, we use a final application of the compressibility of the alignment gadget. This yields $|\cX| = \Oh(\log A + \sum_{\bb \in \bB} |\cNTG_\x(\bb)|) = \Oh(d A^{k_2+1})$ and $|\cY| = \Oh(\log A + \sum_{\bc \in \bC} |\cNTG_\y(\bc)|) = \Oh(A^{k_2} (d + \log A))$. Note that $|X|,|Y| = \Theta(d A^{k_1+k_2})$. 
It is easy to verify that constructing $\cX,\cY$ takes time $\Oh(d A^{k_1+k_2})$.
\end{proof}

We are now ready to prove the theorem.

\begin{proof}[Proof of \thmref{simlb}]
Let $0 < \alpha_n < 1$ and set $\beta := \frac{\alpha_n}{1+\alpha_n}$.
Let $k \ge 2$ and let $\sA$ be a $k$-OV instance with $A$ vectors in dimension $d$. 
We split $k = k_1 + 2 k_2$ with $k_1,k_2 \ge 1$ and $k_2 \approx \beta k$ and $k_1 \approx (1-2\beta) k$. Note that $k_1,k_2$ are restricted to be integers, however, for any $\eps > 0$ and sufficiently large $k$ depending only on $\eps$ and $\alpha_n$ we can ensure $k_2+1 \le (1-\eps/8) \beta k$ and $k_1 \le (1+\eps/4) (1-2\beta) k$. Since $k = k_1 + 2k_2$, it follows that $k_1 + k_2 \ge (1-\beta) k$. 
Note that for the dimension $d$ we can assume $d \le A$, since otherwise an $O(A^{k-\eps} \textup{poly}(d))$ algorithm clearly exists. In particular, for sufficiently large $k$ we have $d \le A^{(\eps/8) \cdot \min\{\beta, 1-\beta \} k}$. 
By Claim~\ref{cla:distsizebounds}, the constructed strings $X,Y$ have length $N$ bounded from above by $\Theta(d A^{k_1+k_2}) = O(d A^{(1+\eps/4) (1-\beta) k}) = O(A^{(1+\eps/2)(1-\beta) k})$ and bounded from below by $\Theta(d A^{k_1+k_2}) = \Omega(A^{(1-\beta) k})$.
The constructed SLPs have size $n = O(d A^{k_2+1}) = O(d A^{(1-\eps/8) \beta k}) = O(A^{\beta k})$. Since $\beta/(1-\beta) = \alpha_n$, it follows that $n = O(N^{\alpha_n})$, and by partially decompressing the SLPs we can ensure the desired $n = \Theta(N^{\alpha_n})$, while keeping $n = O(A^{(1+\eps/2) \beta k})$. 
By \claref{XYcorrect}, computing $\dist(X,Y)$ allows us to decide feasibility of the given $k$-\OV instance.  Hence, any $O((nN)^{1-\eps})$ time algorithm for $\dist(.,.)$ in the setting $n = \Theta(N^{\alpha_n})$ would yield an algorithm for $k$-OV in time $O((A^{(1+\eps/2) k})^{1-\eps}) = O(A^{(1-\eps/2)k})$, contradicting the $k$-OV conjecture.
\end{proof}

\subsubsection{Extended Alignment Gadget for LCS}
\label{sec:lcslb}

In this section, we fix the distance measure to be the LCS distance $\delta(X,Y) = |X|+|Y|-2\cdot L(X,Y)$, where $L(X,Y)$ denotes the length of an LCS $S$ of $X$ and $Y$. Note that $\dist(X,Y)$ counts the number of symbols to be deleted in $X$ to obtain $S$ plus the number of symbols to be deleted in $Y$ to obtain~$S$. We show that $\dist$ admits coordinate values and a compressible extended alignment gadget. Together with \thmref{simlb}, this will yields our conditional lower bound for LCS.

We make use of the same coordinate values as in~\cite{BK15}.

\begin{lem}[{\cite[Lemma V.2]{BK15}}]\label{lem:coordvalues}
LCS admits coordinate values by setting
\[ \ox := 11100,\; \zx := 10011,\; \oy := 00111,\; \zy := 11001. \]
These strings have type $(5, \{0,1\})$.
\end{lem}

It remains to implement a compressible extended alignment gadget. Let us first disregard compressibility.

\begin{lem}\label{lem:lcsalgngadget}
The following construction implements an extended alignment gadget: Let $X_1,\dots,X_n$ of length $\ell_\x$ and $Y_1,\dots,Y_m$ of length $\ell_\y$ be strings over $\Sigma$. We introduce new symbols $\sigma,\rho,\mu \notin \Sigma$,  define  $\kappa_1:= 4(\ell_\x+\ell_\y)$ and $\kappa_2 := 2\kappa_1 + \ell_\x$, and set
\begin{alignat*}{4}
\guard(S)  &:= \sigma^{\kappa_1} & S & \rho^{\kappa_1},
\end{alignat*}
The alignment gadget is now defined as
\begin{alignat*}{10}
X & = &&\guard(X_1) \, &&Z_1^\x \, &&\guard(X_2) && \dots &&Z^\x_{n-1} &&\guard(X_n), \\
Y & = L^\y &&\guard(Y_1) \, &&Z_1^\y \, &&\guard(Y_2) && \dots &&Z^\y_{m-1} &&\guard(Y_m) R^\y, 
\end{alignat*}
where $Z_i^\x = Z_j^\y = \mu^{\kappa_2}$ for $i\in[n-1],j\in[m-1]$ and $L^\y = R^\y = \mu^{n\kappa_2}$. This satisfies property~\eqref{eq:algngadget} of \defref{algngadget} with $C := 2n \kappa_2$.
\end{lem}

\begin{proof}
To analyze our alignment gadget construction (adapting the proof of the LCS gadget of the full version of~\cite{BK15}), we prepare some useful facts.

\begin{claim}[{\cite[Fact V.7]{BK15}}] \label{cla:lcsobs}
  Let $X$ and $Z_1,\ldots,Z_k$ be strings. Set $Z:= Z_1 \concat \dots \concat Z_k$. We have
  \[ \dist(X,Z) = \min_{X(Z_1),\ldots,X(Z_k)} \sum_{j=1}^k \dist(X(Z_j),Z_j),  \]
  where $X(Z_1),\ldots,X(Z_k)$ range over all \emph{ordered partitions} of $X$ into $k$ substrings, i.e., $X(Z_1) = x[i_0+1..i_1], X(Z_2) = x[i_1+1..i_2], \ldots, X(Z_k) = x[i_{k-1}+1..i_k]$ for any $0 = i_0 \le i_1 \le \ldots \le i_{k} = |X|$.
\end{claim}

  \begin{claim} \label{cla:lcsprop}
    Let $U,V$ be strings over $\Sigma$, $\alpha \in \Sigma$ and $k \in \mathbb{N}_0$. Then we have 
\begin{enumerate}[label=(\roman*)]
\item $\dist(U,V) \ge \left| |U| - |V| \right|$, \label{item:distdiff}
\item $\dist(\alpha^k U, \alpha^k V) = \dist(U,V)$, \label{item:distgreedy}
\item Let $W$ be a string not containing $\alpha$. Then  $\dist(W \alpha U,  \alpha^k V) \ge \min\{k, \dist(\alpha U,\alpha^k V)\}$. \label{item:distblock}
\end{enumerate}
We obtain symmetric statements by reversing all involved strings.
  \end{claim}
  \begin{proof}
    \ref{item:distdiff} Suppose $|U| \ge |V|$, then at least $|U|-|V|$ many symbols must be deleted in $U$. The claim follows by symmetry.

    \ref{item:distgreedy} It suffices to show the claim for $k=1$, then the general statement follows by induction. Consider a \LCS\ $S$ of $(\alpha U, \alpha V)$. At least one $\alpha$ is matched in $S$, as otherwise we can extend $S$ by matching both $\alpha$'s. If exactly one $\alpha$ is matched in $S$, then the other $\alpha$ is free, so we may instead match the two $\alpha$'s. Thus, without loss of generality a \LCS\ of $(\alpha U,\alpha V)$ matches the two $\alpha$'s. This yields $L(\alpha U, \alpha W) = 1 + L(U,V)$. Hence, $\dist(\alpha U, \alpha V) = |\alpha U| + |\alpha V| - 2L(\alpha U, \alpha V) = |U| + |V| - 2 L(U,V) = \dist(U,V)$.
 
    \ref{item:distblock} Fix an LCS $S$ of $W \alpha U $ and $\alpha^k V$. If $S$ starts with a symbol other than $\alpha$, then $S$ cannot use any symbol from the $\alpha^k$-prefix of $\alpha^k V$, i.e., the $\alpha^k$-prefix has to be deleted and thus $\dist(W \alpha U, \alpha^k V) \ge k$. Otherwise, if $S$ starts with an $\alpha$, then $S$ cannot us any symbol from $W$ (which is a string over $\Sigma\setminus\{\alpha\}$), i.e., $S$ is an LCS of $\alpha U$ and $\alpha^k W$. Thus $\dist(W \alpha U, \alpha^k V) = |W \alpha U| + |\alpha^k V| - 2 L(\alpha U, \alpha^k V) = |W| + \dist(\alpha U, \alpha^k V)$ and the claim follows.
  \end{proof}

  \begin{claim} \label{cla:prefix}
    Let $\ell \ge 0$. For any prefix $X'$ of $X$ we have $\dist(X',\mu^\ell) \ge \ell$. Moreover, if $X'$ is of the form $\guard(X_1) Z^\x_1 \ldots \guard(X_i) Z^\x_i$ for some $0 \le i < n$ and $\ell \ge i \cdot \kappa_2$, then $\dist(X',\mu^\ell) = \ell$. Symmetric statements hold for any suffix of $X$.
  \end{claim}
  \begin{proof}
    Note that for any $i \in [n]$ the string $\guard(X_i) Z^\x_{i}$ contains $|Z^\x_i| = \kappa_2$ many $\mu$'s and $|\guard(X_i)| = 2\kappa_1 + \ell_\x = \kappa_2$ many non-$\mu$'s. Furthermore, any prefix of $\guard(X_i) Z^\x_i$ contains at least as many non-$\mu$'s as $\mu$'s.
    Hence, the LCS of $X'$ and $\mu^\ell$ has a length of at most $|X'|/2$. This yields $\dist(X',\mu^\ell) = |X'| + |\mu^\ell| - 2L(X',\mu^\ell) \ge \ell$. If $X'$ is of the form $\guard(X_1) Z^\x_1 \ldots \guard(X_i) Z^\x_i$ and $\mu^\ell$ has at least $|X'|/2=i\kappa_2$ many $\mu$'s, we have equality.
  \end{proof}

We now prove that our construction yields an extended alignment gadget.
We start with the upper bound of property~\eqref{eq:algngadget}, i.e., $\dist(X,Y) \le 2n \kappa_2 + \min_{\Lambda \in \strc_{n,m}} \cost(\Lambda)$.

Let $\Lambda = \{(\Delta+1,1),\dots,(\Delta+m,m)\}$ be a structured alignment and consider an ordered partition of $X$ as in \claref{lcsobs} defined as follows:
\begin{align*}
    X(\guard(Y_j)) &:= \guard(X_{\Delta+j}) \qquad \text{for }j \in [m],  \\
    X(Z^\y_j) &:= Z^\x_{\Delta+j} \qquad \text{for }j \in [m-1],  \\
    X(L^\y) &:= \guard(X_1) Z^\x_1 \ldots \guard(X_{\Delta}) Z^\x_{\Delta},  \\
    X(R^\y) &:= Z^\x_{\Delta+m} \guard(X_{\Delta+m+1}) \ldots Z^\x_{n-1} \guard(X_n). 
\end{align*} 
\claref{lcsobs} thus yields
\[\dist(X,Y) \le \dist(X(L^\y),L^\y) + \dist(X(R^\y),R^\y) + \sum_{j=1}^m \dist(\guard(X_{\Delta+j}),\guard(Y_j)) + \sum_{j=1}^{m-1} \dist(Z^\x_{\Delta+j},Z^\y_j). \]

By \claref{prefix}, we obtain $\dist(X(L^\y),L^y) = n\kappa_2$ and symmetrically, $\dist(X(R^\y),R^\y) = n \kappa_2$. Trivially, $\dist(Z^\x_{\Delta+j},Z^\y_j) = 0$. Finally, by matching the padding around $X_i,Y_j$ in $\guard(X_i),\guard(Y_j)$, we obtain $\dist(\guard(X_{\Delta+j}),\guard(Y_j))= \dist(X_{\Delta+j},Y_j)$ by \claref{lcsprop}\ref{item:distgreedy}. Summing up all contributions, we obtain
\[ \dist(X,Y) \le 2n\kappa_2 + \sum_{(i,j) \in \Lambda} \dist(X_i,Y_j),\]
which holds for an arbitrary $\Lambda \in \strc_{n,m}$, thus concluding the upper bound.

It remains to prove the lower bound of property~\eqref{eq:algngadget}, i.e., $\dist(X,Y) \ge 2n\kappa_2 + \min_{\Lambda \in \algn_{n,m}} \cost(\Lambda)$.
  Set $M^\y = \guard(Y_1) Z^\y_1 \dots Z^\y_{m-1} \guard(Y_m)$. Using \claref{lcsobs}, we let $X(L^\y)$, $X(M^\y)$ and $X(R^\y)$ be an ordered partition of $X$ such that
\[  \dist(X,Y) = \dist(X(L^\y),L^\y) + \dist(X(M^\y),M^\y) + \dist(X(R^\y),R^\y).\] 
   Since $L^\y = \mu^{n\kappa_2}$ and $X(L^\y)$ is a prefix of $X$, by \claref{prefix} we have $\dist(X(L^\y),L^\y) \ge n\kappa_2$, and similarly we get $\dist(X(R^\y),R^\y) \ge n\kappa_2$.
  It remains to construct an alignment $\Lambda \in \algn_{n,m}$ satisfying 
  \begin{align} \label{eq:Alcs}
    \cost(\Lambda) \le \dist(X(M^\y),M^\y),%\sum_{j=1}^m \dLCS(X(\guard(y_j)),\guard(y_j)) + \sum_{j=1}^{m-1} \dist(X(Z^\y_j),Z^\y_j), 
  \end{align}
  then together we have shown the desired inequality 
  $\dist(X,Y) \ge 2n\kappa_2 + \min_{\Lambda \in \algn_{n,m}} \cost(\Lambda)$.

As in  \claref{lcsobs}, we let $X(\guard(Y_j))$ for $j \in [m]$ and $X(Z^\y_j)$ for $j \in [m-1]$ be an ordered partition of $X(M^\y)$ such that
\[  \dist(X(M^\y),M^\y) = \sum_{j=1}^m \dist(X(\guard(Y_j)),\guard(Y_j)) + \sum_{j=1}^{m-1} \dist(X(Z^\y_j),Z^\y_j). \]

Let $\mu(U)$ be the number of $\mu$'s in a string $U$ and let $\dist_\delmu(U,V)$ denote the LCS distance of $U$ and $V$ after deleting all $\mu$'s in $U$ and $V$. Clearly, since $|\mu(U)-\mu(V)|$ $\mu$'s have to be deleted in any LCS, we have 
\begin{equation}
\dist(X(M^\y),M^\y) \ge \left(\sum_{j=1}^m \dist_\delmu(X(\guard(Y_j)),\guard(Y_j))\right)  + |\mu(U)-\mu(V)|. \label{eq:mubound}
\end{equation}

  Let us construct an alignment $\Lambda$ satisfying~\eqref{eq:Alcs}. For any $j \in [m]$, if $X(\guard(Y_j))$ contains more than half of some $X_{i'}$ (which is part of $\guard(X_{i'})$), then let $i$ be the leftmost such index and align $i$ and $j$. Note that the set $\Lambda$ of all these aligned pairs $(i,j)$ is a valid alignment in $\algn_{n,m}$, since no $X_i$ or $Y_j$ can be aligned more than once.
  
We prove the following claims:

\begin{claim}  
For any aligned pair $(i,j) \in \Lambda$, we have $\dist_\delmu(X(\guard(Y_j)),\guard(Y_j)) \ge \dist(X_i,Y_j)$.
\end{claim}
\begin{proof}
Let $U$ be $X(\guard(Y_j))$ with all $\mu$'s deleted (note that $\guard(Y_j)$ contains no $\mu$'s). We will prove $\dist(U,\guard(Y_j)) \ge \dist(X_i,Y_j)$. Recall that $X(\guard(Y_j))$ contains more than half of $X_i$, thus so does $U$. If $\left| |U| - |\guard(Y_j)| \right| \ge \ell_\x+\ell_\y$, then we have $\dist(U,\guard(Y_j)) \ge \ell_\x+\ell_\y \ge \dist(X_i,Y_j)$ by \claref{lcsprop}\ref{item:distdiff}. Since $|\guard(Y_j)| = 2\kappa_1 + \ell_y$, we may hence assume $2\kappa_1 - \ell_\x \le |U| \le 2\kappa_1 + \ell_\x  + 2\ell_\y$.

We distinguish three cases: Either $U$ contains $X_i$ fully (C1), or at least its right half but not fully (C2), or at least its left half but not fully (C3).

In case (C2), $U$ is of the form $X_i' \rho^{\kappa_1} \sigma^{a} X'_{i+1} \rho^{b}$ where $X'_i$ is a suffix of $X_i$, $a\le\kappa_1$, $X'_{i+1}$ is a prefix of $X_{i+1}$ and $b \le 2\ell_\y$. In this case, by \claref{lcsprop}\ref{item:distblock} with $\alpha=\sigma$ and $W=X_i' \rho^{\kappa_1}$, we have $\dist(U,\guard(Y_j)) \ge \min\{\kappa_1, \dist(\sigma^{a} X'_{i+1} \rho^b, \sigma^{\kappa_1} Y_j \rho^{\kappa_1})\}$. Note that since the second string contains $\kappa_1$ $\rho$'s and the first string contains less than $2\ell_\y$ $\rho$'s, we have $\dist(\sigma^{a} X'_{i+1} \rho^b, \sigma^{\kappa_1} Y_j \rho^{\kappa_1}) \ge \kappa_1 - 2\ell_\y$. Thus $\dist(U,\guard(Y_j)) \ge \kappa_1 - 2\ell_\y \ge \ell_\x+\ell_\y \ge \dist(X_i,Y_j)$.

The case (C3) is symmetric to (C2). 

Finally, in case (C1), $U$ takes one of three forms: either (F1) $\sigma^a X_i \rho^{\kappa_1} \sigma^{b} X'_{i+1} \rho^c$, where $a\ge 0$, $b \le \kappa_1$, $X'_{i+1}$ is a (possibly empty) prefix of $X_{i+1}$ and $c \le 2\ell_\y$, or the symmetric version (F2) $X_{i-1}' \rho^b \sigma^{\kappa_1} X_i \rho^a$ with $X_{i-1}'$ a suffix of $X_{i-1}$ and all other paremeters as before, or finally (F3) $\rho^a \sigma^{\kappa_1} X_i \rho^{\kappa_1} \sigma^b$ with $a,b\le 2\ell_\y$.

For form (F1), we compute
\begin{align*}
\dist(U,\guard(Y_j)) & = \dist(\sigma^a X_i \rho^{\kappa_1} \sigma^{b} X'_{i+1}, \sigma^{\kappa_1} Y_j \rho^{\kappa_1-c}) \\
& \ge  \min\{\kappa_1- c, \dist(\sigma^a X_i \rho^{\kappa_1}, \sigma^{\kappa_1} Y_j \rho^{\kappa_1})\},
\end{align*}
where we used \claref{lcsprop}\ref{item:distgreedy} in the first line and \claref{lcsprop}\ref{item:distblock} with $\alpha = \rho$ and $W= \sigma^{b} X'_{i+1}$ in the second line.  Note that by \claref{lcsprop}\ref{item:distgreedy} and by deleting all $\sigma$'s only occuring in one string, $\dist(\sigma^{a} X_i \rho^{\kappa_1}, \sigma^{\kappa_1} Y_j \rho^{\kappa_1}) = (\kappa_1 - a) + \dist(X_i, Y_j) \ge \dist(X_i,Y_j)$. Since $\kappa_1 - c \ge \ell_\x+\ell_\y \ge \dist(X_i,Y_j)$, the claim follows for (F1). Symmetrically, we can do the same for (F2).

For the final form (F3), we compute, using \claref{lcsprop}\ref{item:distblock} from the left with $\alpha = \sigma$ and $W = \rho^a$ and from the right with $\alpha = \rho$ and $W= \sigma^b$, $\dist(U,\guard(Y_j)) \ge \min\{\kappa_1, \dist(\sigma^{\kappa_1} X_i \rho^{\kappa_1}, \sigma^{\kappa_1} Y_j \rho^{\kappa_1})\} = \dist(X_i,Y_j)$, where the last equality follows from \claref{lcsprop}\ref{item:distgreedy} and $\kappa_1 \ge \dist(X_i,Y_j)$.
\end{proof}

  \begin{claim}\label{cla:aligned}
    If $j$ is unaligned in $\Lambda$, then $\dist_\delmu(X(\guard(Y_j)),\guard(Y_j)) \ge \ell_\x+\ell_\y$.
  \end{claim}
\begin{proof}
Let $U$ be $X(\guard(Y_j))$ with all $\mu$'s deleted (note that $\guard(Y_j)$ contains no $\mu$'s). We will prove $\dist(U,\guard(Y_j)) \ge \ell_\x+\ell_\y$.
Since $X(\guard(Y_j))$ contains less than half of any $\guard(X_i)$, $U$ is of the form $X_{i}' \rho^{a} \sigma^{b} X'_{i+1}$ for some $i$, a suffix $X'_i$ of $X_i$, some $a,b \le \kappa_1$ and a prefix $X'_{i+1}$ of $X_{i+1}$.

Furthermore using \claref{lcsprop}\ref{item:distblock} with $\alpha=\sigma$ and $W=X_{i}' \rho^a$, we obtain that $\dist(U,\guard(Y_j)) \ge \min\{\kappa_1, \dist(\sigma^b X'_{i+1}, \sigma^{\kappa_1} Y_j \rho^{\kappa_1})\}.$ Since $\sigma^{\kappa_1} Y_j \rho^{\kappa_1}$ contains $\kappa_1$ $\rho$'s, while $\sigma^b X'_{i+1}$ contains none, we conclude $\dist(U,\guard(Y_j)) \ge \kappa_1 \ge \ell_\x+\ell_\y$.
\end{proof}

Let us prove \eqref{eq:Alcs}. If $|\Lambda| < m$, that is, there is an unaligned $j$, combining the two previous claims with~\eqref{eq:mubound} results in
\begin{align*}
\dist(X(M^\y),M^\y) \ge \left(\sum_{(i,j)\in \Lambda} \dist(X_i,Y_j)\right) + (m-|\Lambda|) (\ell_\x+\ell_\y) + |\mu(X(M^\y))-\mu(M^\y)|\kappa_2\ge \cost(\Lambda),
\end{align*}
since $\ell_\x+\ell_\y \ge \max_{i,j} \dist(X_i,Y_j)$ and $|\mu(X(M^\y))-\mu(M^\y)| \ge 0$.

Otherwise, if $|\Lambda| = m$, we have $\Lambda= \{(i_1,1),\dots,(i_m, m)\}$ with $i_1 < i_2 < \dots < i_m$. Note that $X(M^\y)$ is a substring that contains at least half of all $X_{i_1},\dots,X_{i_m}$ by definition of the alignment~$\Lambda$. Thus, $\mu(X(M^\y)) \ge (i_m-i_1)\kappa_2$, since it contains all  $Z^\x_{i_1}, \dots, Z^\x_{i_m-1}$. Since $\mu(M^\y) = (m-1)\kappa_2$, we obtain by~\eqref{eq:mubound} and \claref{aligned},
\begin{align*}
\dist(X(M^\y),M^\y) \ge \left(\sum_{(i,j)\in \Lambda} \dist(X_i,Y_j)\right) + (i_m-i_1-m+1) \kappa_2 \ge \cost(\Lambda),
\end{align*}
where we used that $\kappa_2 \ge \ell_\x + \ell_\y \ge \max_{i,j} \dist(X_i,Y_j)$.

This concludes the proof of \lemref{lcsalgngadget}, showing that our construction yields an extended alignment gadget.
\end{proof}

It remains to argue that a slight adaption of this gadget is compressible.
\begin{lem}\label{lem:lcscomalgngadget}
Consider the setting of \lemref{lcsalgngadget}. Adapt the definition of the extended alignment gadget slightly by defining
\begin{alignat*}{20}
X' & =\,  && Z^\x_0 \;  && X \;  && Z^\x_{n}   && =  & & Z^\x_0\,  && \guard(X_1)\,  && Z^\x_1  && \dots  && Z^\x_{n-1}\,  && \guard(X_n)\,  && Z^\x_n, \\
Y' & =\,   && Z^\y_0 \;  && Y \;  && Z^\y_{m}   && =  L^\y\, & & Z^\y_0\,  && \guard(Y_1)\,  && Z^\y_1  && \dots  && Z^\y_{m-1}\,  && \guard(Y_m)\,  && Z^\x_m\, R^\y,
\end{alignat*}
where we define the additional blocks $Z^\x_i,Z^\y_j = \mu^{\kappa_2}$ with $i\in \{0,n\},j\in \{0,m\}$. This construction $(X',Y')$ yields a \emph{compressible} extended alignment gadget.
\end{lem}
\begin{proof}
By \claref{lcsprop}\ref{item:distgreedy}, we see that $\dist(X',Y') = \dist(X,Y)$, and thus $X',Y'$ satisfies the extended alignment gadget condition~\eqref{eq:algngadget} of \defref{algngadget} by \lemref{lcsalgngadget}.

We define $\pad_\x(S) = \pad_\y(S) = \mu^{\kappa_2/2} \guard(S) \mu^{\kappa_2/2}$ and $X_L=X_R=\mu^{\kappa_2/2}$ and $Y_L = Y_R = \mu^{n\kappa_2 + \kappa_2/2}$. Then we have
$X' = X_L \left(\bigconcat_{i=1}^n \pad_\x(X_i) \right) X_R$ and $Y' = Y_L \left(\bigconcat_{j=1}^m \pad_\x(Y_j) \right) Y_R$.
By \obsref{repetition}, we can construct SLPs $\cX_L, \cX_R, \cY_L, \cY_R$ for $X_L,X_R,Y_L,Y_R$ of size $\Oh(\log n\kappa_2) = \Oh(\log n+\log(\ell_\x+\ell_\y))$.
Likewise, given SLPs $\cX_i,\cY_j$ for $X_i,Y_j$, we can construct SLPs for $\pad_\x(X_i)$, $\pad_\y(Y_j)$ of size $\Oh(|\cX_i| + \log(\ell_\x+\ell_\y))$, $\Oh(|\cY_j| + \log(\ell_\x+\ell_\y))$, respectively, as we can generate the paddings $\mu^{\kappa_2/2} \sigma^{\kappa_1}$ and $\rho^{\kappa_1} \mu^{\kappa_2/2}$ around $X_i$ and $Y_j$ using \obsref{repetition}. This concludes the proof.
\end{proof}

Our LCS lower bound now follows.

\begin{proof}[Proof of \thmref{lcslb}]
Since $\delta$ admits coordinate values and a compressible extended alignment gadget by \lemrefs{coordvalues}{lcscomalgngadget}, we obtain the claim by the general lower bound of  \thmref{simlb}, as computing the length of the LCS of $X$ and $Y$ is equivalent to computing $\dist(X,Y)$.
\end{proof}

\section{Tight Bounds Assuming (Combinatorial) \boldmath$k$-Clique} \label{sec:cliquelowerbounds}

% !TEX root = main.tex

In this section we prove matching conditional lower bounds based on the $k$-Clique conjecture or combinatorial $k$-Clique conjecture for the following problems: 
\begin{itemize}
\item NFA Acceptance, i.e., deciding whether a given non-deterministic finite automaton accepts a given string,
\item CFG Parsing, i.e., deciding whether a given context-free grammar generates a given string,
\item RNA Folding, i.e., computing the maximum number of non-crossing matching pairs of indices in a given string.
\end{itemize}
See the respective subsections for precise problem definitions. 

For NFA Acceptance, the compression used in our proof is extremely simple, in that we only rely on the fact that any repetition $T^\ell$ can be generated by an SLP of size $O(|T| + \log \ell)$ (Observation~\ref{obs:repetition}). For CFG Parsing and RNA Folding, our construction is much more subtle. For both problems, we use that the following string and some variants thereof are compressible:
\[ S_v := \bigconcat_{u_1,\ldots,u_k \in V} [\text{$v$ is adjacent to every $u_i$}] \]
That is, we enumerate all $k$-tuples $(u_1,\ldots,u_k) \in V^k$ and for each one check whether all $u_i$'s are adjacent to a fixed vertex $v$, writing 1 or 0 depending on this check. This string is generated by an SLP of size $O(V)$: Enumerate all $u_1 \in V$. If $u_1$ is not adjacent to $v$, then for all $u_2,\ldots,u_k$ the check results in 0, so we can simply write $0^{V^{k-1}}$, which is well compressible by Observation~\ref{obs:repetition}. Otherwise, if $u_1$ is adjacent to $v$, then we can recurse to $u_2$, and the following $V^{k-1}$ symbols do not depend on $u_1$ anymore. 
  More formally, denote by $\textup{Repeat}^{(d)}_0$ an SLP generating the string $0^{V^d}$. Then with the following SLP rules, for $1 \le d \le k$, we have $S_v = \eval(\textup{Adj}_v^{(k)})$.
  \begin{align*}
    \textup{Adj}^{(0)}_v &\to 1, \\
    \textup{Adj}^{(d)}_v &\to \bigconcat_{u \in V} \begin{cases} \textup{Adj}^{(d-1)}_v, \text{ if } \{u,v\} \in E \\ \textup{Repeat}^{(d-1)}_0, \text{ otherwise} \end{cases}
  \end{align*}
  Here we use the ``syntactic sugar'' of having more than two SLP symbols on the right hand side, but clearly this can be converted to a proper SLP of size $O(V)$.
  
  We stress that if in the string $S_v$ we would enumerate only the $k$-cliques instead of all $k$-tuples, then $S$ would no longer be easily compressible, since then even the length of a substring depends on the ``history'' of choosing $u_1,\ldots,u_{k-d}$, and thus the above recursive way of writing $S$ would fail. This demonstrates how subtle our argument is.

  \paragraph{Known Lower Bounds from Classic Complexity Theory} 
Plandowski and Rytter~\cite{plandowski1999complexity} showed that deciding whether a given compressed text can be generated by a given CFG is {\sf PSPACE}-complete. Later, Lohrey~\cite{lohrey2006word} showed that this holds even if we restrict the CFG to be fixed (i.e., not part of the input) and deterministic. We observe that the RNA Folding problem is at least as hard as Longest Common Subsequence (see, e.g.~\cite{ABV15b}). This implies that RNA Folding is {\sf PP}-hard (see the discussion at the beginning of Section~\ref{sec:subseqlower}). Finally, the NFA Acceptance problem can be solved in polynomial $O(nq^{\omega})$ time (see below) and previously no conditional lower bounds were known.

\subsection{NFA Acceptance} \label{sec:nfaaccept}

For general notation regarding finite automata, see Section~\ref{sec:dfaaccept}. 
Consider the compressed variant of the acceptance problem of nondeterministic finite automata (NFAs).
\begin{problem}[\NFAAccept]
We are given a text $T$ of length $N$ by a grammar-compressed representation $\cT$ of size $n$ as well as a NFA $F$ with $q$ states, i.e., for any two states $z,z'$ and any symbol $\sigma \in \Sigma$ we are given whether $\transition{z}{\sigma}{z'}$. 
Decide whether $T$ is accepted by $F$.
\end{problem}

Note that the input size is $\tilde O(n + q^2)$, since we again assume the alphabet size $|\Sigma|$ to be constant.

The naive solution is to decompress $\cT$ to obtain $T$ and run the standard acceptance algorithm for NFAs, which takes time $\Oh(|T| q^2)=\Oh(N q^2)$. Exploiting the compressed setting, one can obtain an $\Oh(n q^\omega)$-time algorithm~\cite{plandowski1999complexity}: 
Recall that $\cT$ is a set of rules of the form $S_i \to S_{\ell(i)} S_{r(i)}$ or $S_i \to \sigma_i$, with $\ell(i), r(i) < i$ and $\sigma_i \in \Sigma$, for $1 \le i \le n$.
We compute, for increasing $i$, the state transition matrix $A_i$, where $(A_i)_{z,z'} = 1$ if we can start in state $z$, read the string $\eval(S_i)$, and end in state $z'$, and $(A_i)_{z,z'} = 0$ otherwise. 
For $S_i \to S_{\ell(i)} S_{r(i)}$ we can compute $A_i$ as $A_{\ell(i)} \cdot A_{r(i)}$, where $\cdot$ is Boolean matrix multiplication. For $S_i \to \sigma_i$ we simply have $(A_i)_{z,z'} = 1$ if $\transition{z}{\sigma_i}{z'}$, and 0 otherwise. 
Hence, $A_i$ can be computed in time $\Oh(q^\omega)$ for every $i$. 
The text $T$ is then accepted by $F$ if there is an accepting state $z$ such that $(A_n)_{z_0,z} = 1$, where $z_0$ is the starting state of $F$. 

Note that this best-known upper bound $O(\min\{n q^\omega, N q^2\})$ contains ``mixed terms'' with some factors having exponent $\omega$ but not all. Since no standard conjecture contains such mixed terms, we cannot hope to prove a matching lower bound of $\min\{n q^\omega, N q^2\}^{1-o(1)}$. However, restricting our attention to combinatorial algorithms the best-known running time simplifies to $O(\min\{n q^3, N q^2\})$, and we can hope to prove a matching lower bound under some assumption on combinatorial algorithms, say for matrix multiplication or $k$-Clique. For matrix multiplication, the typical issue that we would need to considerably compress the input graph~\cite{ABV15b} is a barrier for a reduction. Hence, we can only hope to prove a matching lower bound for combinatorial algorithms assuming the $k$-Clique conjecture. We prove such a result in the following.

\begin{thm} \label{NFA_hardness}
  Assuming the combinatorial $k$-Clique conjecture, there is no combinatorial algorithm for \NFAAccept\  in time $O(\min\{n q^3, N q^2\}^{1-\eps})$ for any $\eps > 0$. This holds even restricted to instances with $n = \Theta(q^{\alpha_n})$ and $N = \Theta(q^{\alpha_N})$ for any $\alpha_N \ge \alpha_n > 0$.
\end{thm}

\begin{proof}
%Let $\alpha_N > \alpha_n > 0$ and $\eps > 0$ be given. By slight abuse of notation, we set $n := \lfloor q^{\alpha_n} \rfloor$ and $N := \lfloor q^{\alpha_N} \rfloor$ and will construct, for any $q$, hard instances $(t,D)$ such that $|t| = \Theta(N)$, $|c(t)| = \Theta(n)$ and $D$ has $\Theta(q)$ states. 
Let $k \ge 3$ and let $G=(V,E)$ be a $k$-Clique instance. 
In the following, for any $\kappa, \kappa' \ge 1$ with $3 \kappa + \kappa' = k$ we will construct an equivalent \NFAAccept\ instance with $q = O(V^{\kappa+1} \log V)$, $N=|T| = O(V^{\kappa+\kappa'} \log V)$, and $n=|\cT| = O(V^{\kappa'} \log V)$. Note that a combinatorial $O(\min\{n q^3, N q^2\}^{1-\eps})$ time algorithm for \NFAAccept\ then yields a combinatorial algorithm for $k$-Clique in time $O(V^{(3\kappa + \kappa'+3)(1-\eps)} \log^4 V) = O(V^{(k+4)(1-\eps)})$, which for $k \ge 8/\eps$ is $O(V^{k(1+\eps/2)(1-\eps)}) = O(V^{k(1-\eps/2)})$, contradicting the combinatorial $k$-Clique conjecture. This yields the desired conditional lower bound.
At the end of this proof we will strengthen this statement to even hold for all restrictions $n = \Theta(q^{\alpha_n})$ and $N = \Theta(q^{\alpha_N})$. 

Our construction uses the following gadgets. 

\paragraph{Neighborhood Gadgets} 
Let $V = \{v_1,\ldots,v_n\}$ and denote by $NG_T(v_i)$ the binary encoding of the number $i$ using $\lceil \log V \rceil$ bits.
For any $v \in V$, let $NG_F(v)$ be the NFA that has start state $s$ and target state $t$, and $|N(v)|$ disjoint directed paths from $s$ to $t$ such that the path corresponding to neighbor $u \in N(v)$ spells $NG_T(u)$. Clearly, we can walk from $s$ to $t$ in $NG_F(v)$ parsing the string $NG_T(u)$ if and only if $u$ is a neighbor of $v$.

\paragraph{Clique Gadgets} 
For two neighborhood gadgets $NG_F(u), NG_F(v)$ as above, we define their \emph{concatenation} $NG_F(u) \concat NG_F(v)$ as the NFA where we identify the target state $t$ of $NG_F(u)$ with the starting state $s$ of $NG_F(v)$. The start state of the concatenation is the start state of $NG_F(u)$, and the target state is the target state of $NG_F(v)$.
We combine neighborhood gadgets to clique gadgets as follows. 
Let $\kappa,\kappa' \ge 1$. Let $C = \{u_1,\ldots,u_\kappa\}$ be a $\kappa$-clique and $C' = \{u'_1,\ldots,u'_{\kappa'}\}$ be a $\kappa'$-clique in $G$. 
We define the following concatenation of NFAs and strings, respectively:
\begin{align*}
 CG_F(C, \kappa, \kappa') &:= \bigconcat_{i=1}^{\kappa} \bigconcat_{j=1}^{\kappa'} NG_F(u_i), \\
 CG_T(C', \kappa, \kappa') &:= \bigconcat_{i=1}^\kappa \bigconcat_{j=1}^{\kappa'} NG_T( u'_j ).
\end{align*}
Observe that we can walk from start to target state of $CG_F(C, \kappa, \kappa')$ parsing $CG_T(C', \kappa, \kappa')$ if and only if $C \cup C'$ forms a $(\kappa+\kappa')$-clique, since the neighborhood gadgets check adjacency for each pair of nodes $u_i \in C$ and $u'_j \in C'$. 

\newcommand{\CC}{\mathcal{C}}

\paragraph{Complete Construction}
For $\kappa \ge 1$, let $\CC(\kappa)$ be the set of $\kappa$-cliques in $G$, and set $m(\kappa) := |\CC(\kappa)|$. Let $\kappa, \kappa' \ge 1$ such that $3\kappa + \kappa' = k$. 
The final text is defined as
\[ T := \$ \concat \bigconcat_{C' \in \CC(\kappa')} \Big(\big(\#\concat CG_T(C', \kappa, \kappa') \big)^{m(\kappa)+4} \concat \$ \Big),\]
using alphabet $\{0,1,\#,\$\}$. 

The NFA $F$ consists of four copies of the clique gadgets $CG_F(C,\kappa,\kappa')$ for any $\kappa$-clique $C$, denoted by $CG^r(i)$ for $1 \le i \le m(\kappa)$ and $1 \le r \le 4$. Additionally, we have states $s, s_1, s_2, \ldots, s_{m(\kappa)}$ and $t, t_1, t_2, \ldots, t_{m(\kappa)}$. These states are connected as follows.
In the starting state $s$ we can stay as long as we want, reading any symbol in the alphabet $\{0,1,\#,\$\}$. When reading $\$$ we can alternatively go to state $s_1$. In any state $s_i$ when reading $0$ or $1$ we stay in $s_i$, while when reading $\#$ we either go to to the starting state of $CG^1(i)$ or to $s_{i+1}$ (the latter is only possible if $i < m(\kappa)$). For any $1 \le r < 4$ and $i,j$, from the ending state of $CG^r(i)$ when reading $\#$ we can go to the starting state of $CG^{r+1}(j)$ if the corresponding cliques together form a $2\kappa$-clique. From the ending state of $CG^4(i)$ when reading $\#$ we go to $t_i$.  In any state $t_i$ when reading $0$ or $1$ we stay in $t_i$, while when reading $\#$ we go to $t_{i+1}$, or to $t$ if $i = m(\kappa)$. Finally, $t$ is the only accepting state and we stay in $t$ reading any symbol in the alphabet. This finishes the construction of the \NFAAccept\ instance. See Figure \ref{NFA_figure} for the illustration of the NFA. %\karl{picture}

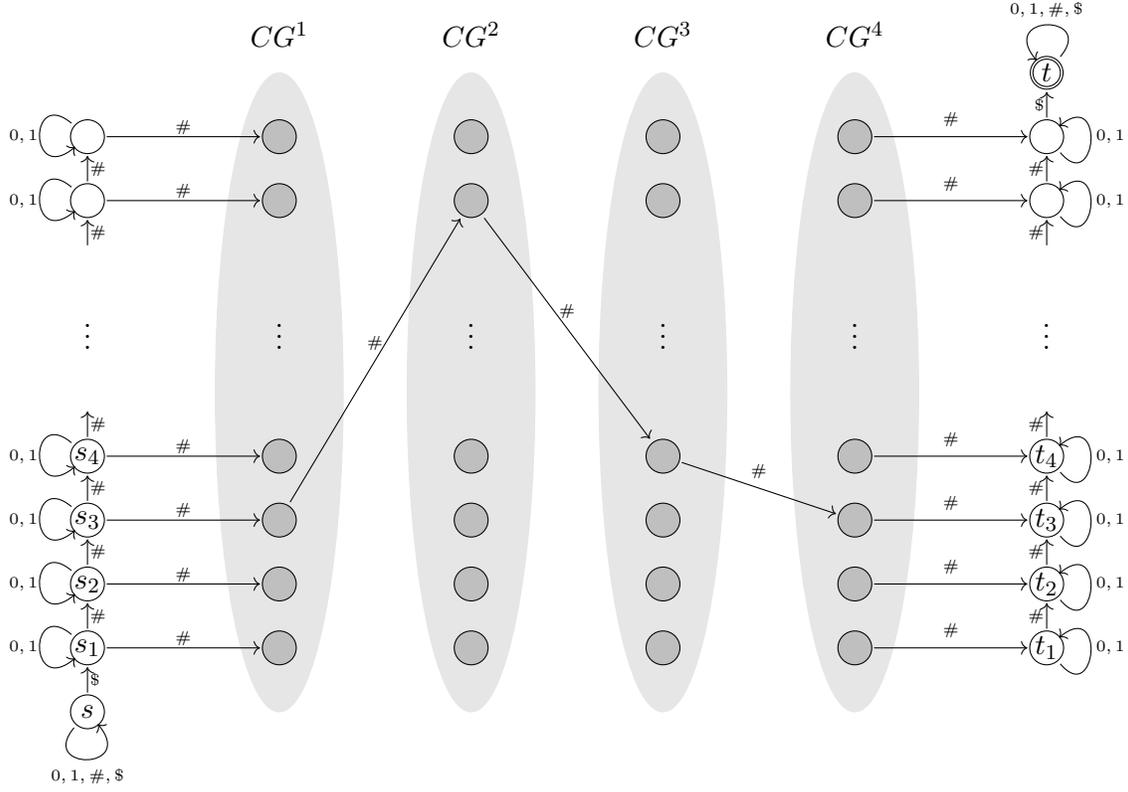
\begin{figure*}
	\centering
	
	\begin{tikzpicture}[scale=1.7]
		\node at (1.5,5.3) {$CG^1$};
		\draw [gray!20,fill=gray!20] (1.5,2.5) ellipse (.5cm and 2.5cm);
	
		\node at (3,5.3) {$CG^2$};
		\draw [gray!20,fill=gray!20] (3,2.5) ellipse (.5cm and 2.5cm);
		
		\node at (4.5,5.3) {$CG^3$};
		\draw [gray!20,fill=gray!20] (4.5,2.5) ellipse (.5cm and 2.5cm);
		
		\node at (6,5.3) {$CG^4$};
		\draw [gray!20,fill=gray!20] (6,2.5) ellipse (.5cm and 2.5cm);

		\node at (0,0) (s) {};
		%\node[anchor=east] at (s) {$s$};
		\node[draw,circle,minimum size=0.45cm,inner sep=0pt] at (s) {$s$};
		%\draw[->,shorten >=1.2 mm,shorten <=1.2 mm,loop] (s) -- (s);
		
		%\draw (0,0)  to[in=-50,out=-130,loop] (0,0);
		\node [draw=none] at (0,0) {} edge [in=-50,out=-130,shorten <=1 mm,loop] ();
		\node at (0,-.5) {\tiny $0,1,\#,\$$};
		
		\node at (0,0.5) (s1) {};
		\node [draw=none] at (s1) {} edge [in=220,out=140,loop,shorten <=1 mm] ();
				\node at (-0.5,0.5) {\tiny $0,1$};
		%\node[anchor=east] at (s1) {$s_1$};
		\node[draw,circle,minimum size=0.45cm,inner sep=0pt] at (s1) {$s_1$};
		
		\node at (0,1) (s2) {};
		\node [draw=none] at (s2) {} edge [in=220,out=140,loop,shorten <=1 mm] ();
				\node at (-0.5,1) {\tiny $0,1$};
		%\node[anchor=east] at (s2) {$s_2$};
		\node[draw,circle,minimum size=0.45cm,inner sep=0pt] at (s2) {$s_2$};
		
		\node at (0,1.5) (s3) {};
		\node [draw=none] at (s3) {} edge [in=220,out=140,loop,shorten <=1 mm] ();
				\node at (-0.5,1.5) {\tiny $0,1$};
		%\node[anchor=east] at (s3) {$s_3$};
		\node[draw,circle,minimum size=0.45cm,inner sep=0pt] at (s3) {$s_3$};
		
		\node at (0,2) (s4) {};
		\node [draw=none] at (s4) {} edge [in=220,out=140,loop,shorten <=1 mm] ();
				\node at (-0.5,2) {\tiny $0,1$};
		%\node[anchor=east] at (s4) {$s_4$};
		\node[draw,circle,minimum size=0.45cm,inner sep=0pt] at (s4) {$s_4$};
		
		\node at (0,2.5) (s5) {};
		
		\node at (0,3) (ss) {$\vdots$};
		
		\node at (0,3.5) (sk2) {};
		
		\node at (0,4) (sk1) {};
		\node [draw=none] at (sk1) {} edge [in=220,out=140,loop,shorten <=1 mm] ();
				\node at (-0.5,4) {\tiny $0,1$};
		
		\node at (0,4.5) (sk) {};
		\node [draw=none] at (sk) {} edge [in=220,out=140,loop,shorten <=1 mm] ();
				\node at (-0.5,4.5) {\tiny $0,1$};
		%\node[anchor=east] at (sk) {$s_{m(\kappa)}$};
		\node[draw,circle,minimum size=0.45cm,inner sep=0pt] at (sk) {};
		\node[draw,circle,minimum size=0.45cm,inner sep=0pt] at (sk1) {};

		%\draw[fill=black] (s) circle (0.2mm);
		
		\draw[->,shorten >=1.2 mm,shorten <=1.2 mm] (s) -- node[midway,right=-1mm]{\tiny $\$$} (s1);
		
		%\draw[fill=black] (s1) circle (0.2mm);
		\draw[->,shorten >=1.2 mm,shorten <=1.2 mm] (s1) -- node[midway,right=-1mm]{\tiny $\#$} (s2);
		
		%\draw[fill=black] (s2) circle (0.2mm);
		\draw[->,shorten >=1.2 mm,shorten <=1.2 mm] (s2) -- node[midway,right=-1mm]{\tiny $\#$} (s3);
		
		%\draw[fill=black] (s3) circle (0.2mm);
		\draw[->,shorten >=1.2 mm,shorten <=1.2 mm] (s3) -- node[midway,right=-1mm]{\tiny $\#$} (s4);
		
		%\draw[fill=black] (s4) circle (0.2mm);
		\draw[->,shorten >=1.2 mm,shorten <=1.2 mm] (s4) -- node[midway,right=-1mm]{\tiny $\#$} (s5);
		
		%\draw[fill=black] (sk1) circle (0.2mm);
		\draw[->,shorten >=1.2 mm,shorten <=1.2 mm] (sk2) -- node[midway,right=-1mm]{\tiny $\#$} (sk1);
		
		%\draw[fill=black] (sk) circle (0.2mm);
		\draw[->,shorten >=1.2 mm,shorten <=1.2 mm] (sk1) -- node[midway,right=-1mm]{\tiny $\#$} (sk);
		
		%%%%%%%%%%%%%%%%%%
		
		\node at (1.5,0.5) (CG11) {};
		%\node[anchor=west] at (CG11) {\scriptsize $CG^1(1)$};
		%\draw[fill=black] (CG11) circle (0.2mm);
		\node[draw,circle,minimum size=0.45cm,inner sep=0pt,fill=gray!50] at (CG11) {};
		
		\node at (1.5,1) (CG12) {};
		%\node[anchor=west] at (CG12) {\scriptsize $CG^1(2)$};
		\node[draw,circle,minimum size=0.45cm,inner sep=0pt,fill=gray!50] at (CG12) {};
		
		\node at (1.5,1.5) (CG13) {};
		%\node[anchor=west] at (CG13) {\scriptsize  $CG^1(3)$};
		\node[draw,circle,minimum size=0.45cm,inner sep=0pt,fill=gray!50] at (CG13) {};
		
		\node at (1.5,2) (CG14) {};
		%\node[anchor=west] at (CG14) {\scriptsize $CG^1(4)$};
		\node[draw,circle,minimum size=0.45cm,inner sep=0pt,fill=gray!50] at (CG14) {};
				
		\node at (1.5,3) (CG1) {$\vdots$};
		
		\node at (1.5,4) (CG1k1) {};
		\node[draw,circle,minimum size=0.45cm,inner sep=0pt,fill=gray!50] at (CG1k1) {};
		
		\node at (1.5,4.5) (CG1k) {};
		%\node[anchor=west] at (CG1k) {\scriptsize $CG^1(m(\kappa))$};
		\node[draw,circle,minimum size=0.45cm,inner sep=0pt,fill=gray!50] at (CG1k) {};
		
		\draw[->,shorten >=1.2 mm,shorten <=1.2 mm] (s1) -- node[midway,above=-1mm]{\tiny $\#$} (CG11);
		\draw[->,shorten >=1.2 mm,shorten <=1.2 mm] (s2) -- node[midway,above=-1mm]{\tiny $\#$} (CG12);
		\draw[->,shorten >=1.2 mm,shorten <=1.2 mm] (s3) -- node[midway,above=-1mm]{\tiny $\#$} (CG13);
		\draw[->,shorten >=1.2 mm,shorten <=1.2 mm] (s4) -- node[midway,above=-1mm]{\tiny $\#$} (CG14);
		\draw[->,shorten >=1.2 mm,shorten <=1.2 mm] (sk1) -- node[midway,above=-1mm]{\tiny $\#$} (CG1k1);
		\draw[->,shorten >=1.2 mm,shorten <=1.2 mm] (sk) -- node[midway,above=-1mm]{\tiny $\#$} (CG1k);
		
		%%%%%%%%%%%%%%%%%%%
		
		\node at (3,0.5) (CG21) {};
		\node[draw,circle,minimum size=0.45cm,inner sep=0pt,fill=gray!50] at (CG21) {};
		
		\node at (3,1) (CG22) {};
		\node[draw,circle,minimum size=0.45cm,inner sep=0pt,fill=gray!50] at (CG22) {};
		
		\node at (3,1.5) (CG23) {};
		\node[draw,circle,minimum size=0.45cm,inner sep=0pt,fill=gray!50] at (CG23) {};
		
		\node at (3,2) (CG24) {};
		\node[draw,circle,minimum size=0.45cm,inner sep=0pt,fill=gray!50] at (CG24) {};
		
		\node at (3,3) (CG2) {$\vdots$};
		
		\node at (3,4) (CG2k1) {};
		\node[draw,circle,minimum size=0.45cm,inner sep=0pt,fill=gray!50] at (CG2k1) {};
		
		\node at (3,4.5) (CG2k) {};
		\node[draw,circle,minimum size=0.45cm,inner sep=0pt,fill=gray!50] at (CG2k) {};
		
		%%%%%%%%%%%%%%%%%%%
		
		\node at (4.5,0.5) (CG31) {};
		\node[draw,circle,minimum size=0.45cm,inner sep=0pt,fill=gray!50] at (CG31) {};
		
		\node at (4.5,1) (CG32) {};
		\node[draw,circle,minimum size=0.45cm,inner sep=0pt,fill=gray!50] at (CG32) {};
		
		\node at (4.5,1.5) (CG33) {};
		\node[draw,circle,minimum size=0.45cm,inner sep=0pt,fill=gray!50] at (CG33) {};
		
		\node at (4.5,2) (CG34) {};
		\node[draw,circle,minimum size=0.45cm,inner sep=0pt,fill=gray!50] at (CG34) {};
		
		\node at (4.5,3) (CG3) {$\vdots$};
		
		\node at (4.5,4) (CG3k1) {};
		\node[draw,circle,minimum size=0.45cm,inner sep=0pt,fill=gray!50] at (CG3k1) {};
		
		\node at (4.5,4.5) (CG3k) {};
		\node[draw,circle,minimum size=0.45cm,inner sep=0pt,fill=gray!50] at (CG3k) {};
		
		%%%%%%%%%%%%%%%%%%%
		
		\node at (6,0.5) (CG41) {};
		\node[draw,circle,minimum size=0.45cm,inner sep=0pt,fill=gray!50] at (CG41) {};
		
		\node at (6,1) (CG42) {};
		\node[draw,circle,minimum size=0.45cm,inner sep=0pt,fill=gray!50] at (CG42) {};
		
		\node at (6,1.5) (CG43) {};
		\node[draw,circle,minimum size=0.45cm,inner sep=0pt,fill=gray!50] at (CG43) {};
		
		\node at (6,2) (CG44) {};
		\node[draw,circle,minimum size=0.45cm,inner sep=0pt,fill=gray!50] at (CG44) {};
		
		\node at (6,3) (CG4) {$\vdots$};
		
		\node at (6,4) (CG4k1) {};
		\node[draw,circle,minimum size=0.45cm,inner sep=0pt,fill=gray!50] at (CG4k1) {};
		
		\node at (6,4.5) (CG4k) {};
		\node[draw,circle,minimum size=0.45cm,inner sep=0pt,fill=gray!50] at (CG4k) {};
		
		\draw[->,shorten >=1.2 mm,shorten <=1.2 mm] (CG13) -- node[midway,above]{\tiny $\#$} (CG2k1);
		\draw[->,shorten >=1.2 mm,shorten <=1.2 mm] (CG2k1) -- node[midway,above]{\tiny $\#$} (CG34);
		\draw[->,shorten >=1.2 mm,shorten <=1.2 mm] (CG34) -- node[midway,above]{\tiny $\#$} (CG43);
		
		%%%%%%%%%%%%%%%%%%%
		
		\node at (7.5,0.5) (t1) {};
		\node [draw=none] at (t1) {} edge [in=40,out=-50,loop,shorten <=1 mm] ();
				\node at (8,0.5) {\tiny $0,1$};
		%\node[anchor=west] at (t1) {$t_1$};
		\node[draw,circle,minimum size=0.45cm,inner sep=0pt] at (t1) {$t_1$};
		
		\node at (7.5,1) (t2) {};
		\node [draw=none] at (t2) {} edge [in=40,out=-50,loop,shorten <=1 mm] ();
				\node at (8,1) {\tiny $0,1$};
		%\node[anchor=west] at (t2) {$t_2$};
		\node[draw,circle,minimum size=0.45cm,inner sep=0pt] at (t2) {$t_2$};
		
		\node at (7.5,1.5) (t3) {};
		\node [draw=none] at (t3) {} edge [in=40,out=-50,loop,shorten <=1 mm] ();
				\node at (8,1.5) {\tiny $0,1$};
		%\node[anchor=west] at (t3) {$t_3$};
		\node[draw,circle,minimum size=0.45cm,inner sep=0pt] at (t3) {$t_3$};
		
		\node at (7.5,2) (t4) {};
		\node [draw=none] at (t4) {} edge [in=40,out=-50,loop,shorten <=1 mm] ();
				\node at (8,2) {\tiny $0,1$};
		%\node[anchor=west] at (t4) {$t_4$};
		\node[draw,circle,minimum size=0.45cm,inner sep=0pt] at (t4) {$t_4$};
		
		\node at (7.5,2.5) (t5) {};
		
		\node at (7.5,3) (tt) {$\vdots$};
		
		\node at (7.5,3.5) (tk2) {};
		
		\node at (7.5,4) (tk1) {};
		\node [draw=none] at (tk1) {} edge [in=40,out=-50,loop,shorten <=1 mm] ();
				\node at (8,4) {\tiny $0,1$};
		
		\node at (7.5,4.5) (tk) {};
		\node [draw=none] at (tk) {} edge [in=40,out=-50,loop,shorten <=1 mm] ();
				\node at (8,4.5) {\tiny $0,1$};
		\node[draw,circle,minimum size=0.45cm,inner sep=0pt] at (tk) {};
		
		\node at (7.5,5) (t) {};
		\node[draw,circle,double,minimum size=0.4cm,inner sep=0pt] at (t) {$t$};

		%\draw[fill=black] (t1) circle (0.2mm);
		\draw[->,shorten >=1.2 mm,shorten <=1.2 mm] (t1) -- node[midway,left=-1mm]{\tiny $\#$} (t2);
		
		%\draw[fill=black] (t2) circle (0.2mm);
		\draw[->,shorten >=1.2 mm,shorten <=1.2 mm] (t2) -- node[midway,left=-1mm]{\tiny $\#$} (t3);
		
		%\draw[fill=black] (t3) circle (0.2mm);
		\draw[->,shorten >=1.2 mm,shorten <=1.2 mm] (t3) -- node[midway,left=-1mm]{\tiny $\#$} (t4);
		
		%\draw[fill=black] (t4) circle (0.2mm);
		\draw[->,shorten >=1.2 mm,shorten <=1.2 mm] (t4) -- node[midway,left=-1mm]{\tiny $\#$} (t5);
		
		%\draw[fill=black] (tk1) circle (0.2mm);
		
		\node[draw,circle,minimum size=0.45cm,inner sep=0pt] at (tk1) {};
		\draw[->,shorten >=1.2 mm,shorten <=1.2 mm] (tk2) -- node[midway,left=-1mm]{\tiny $\#$} (tk1);
		
		%\draw[fill=black] (tk) circle (0.2mm);
		\draw[->,shorten >=1.2 mm,shorten <=1.2 mm] (tk1) -- node[midway,left=-1mm]{\tiny $\#$} (tk);
		
		%\draw[fill=black] (t) circle (0.2mm);
		\draw[->,shorten >=1.2 mm,shorten <=1.2 mm] (tk) -- node[midway,left=-1mm]{\tiny $\$$} (t);

		\draw[->,shorten >=1.2 mm,shorten <=1.2 mm] (CG43) -- node[midway,above]{\tiny $\#$} (t3);
		\draw[->,shorten >=1.2 mm,shorten <=1.2 mm] (CG41) -- node[midway,above]{\tiny $\#$} (t1);
		\draw[->,shorten >=1.2 mm,shorten <=1.2 mm] (CG42) -- node[midway,above]{\tiny $\#$} (t2);
		\draw[->,shorten >=1.2 mm,shorten <=1.2 mm] (CG44) -- node[midway,above]{\tiny $\#$} (t4);
		\draw[->,shorten >=1.2 mm,shorten <=1.2 mm] (CG4k1) -- node[midway,above]{\tiny $\#$} (tk1);
		\draw[->,shorten >=1.2 mm,shorten <=1.2 mm] (CG4k) -- node[midway,above]{\tiny $\#$} (tk);
		
		\node [draw=none] at (t) {} edge [in=130,out=50,loop,shorten <=1 mm] ();
		\node at (7.5,5.5) {\tiny $0,1,\#,\$$};
	\end{tikzpicture}
	
	\caption{An illustration of the NFA constructed in the proof of Theorem \ref{NFA_hardness}. $s$ is the starting state and $t$ is the only accepting state. The first column consists of states $s,s_1,s_2, \ldots, s_{m(\kappa)}$, the last column consists of $t_1,t_2,\ldots,t_{m(\kappa)},t$. The second, third, fourth and fifth columns contain the gadgets $CG^1(i), CG^2(i), CG^3(i), CG^4(i)$, respectively. We do not show all transitions between gadgets $CG^r(i)$ and $CG^{r+1}(j)$ for $r=1,2,3$. As shown in the picture and as we prove, an accepting execution must visit $CG^1(i)$ and $CG^4(j)$ for $i=j$.}
	\label{NFA_figure}
\end{figure*}

\paragraph{Correctness}
Let us first show that if $G$ contains a $(3 \kappa + \kappa')$-clique $C$ then $F$ accepts $T$. Write $C = C_1 + C_2 + C_3 + C'$, where $C'$ is a $\kappa'$-clique and $C_1, C_2, C_3$ are $\kappa$-cliques (with indices $i_1, i_2, i_3$ in $\CC(\kappa)$). 
We can stay in $s$ until the beginning of the substring $T' := \$ \concat \big(\#\concat CG_T(C', \kappa, \kappa') \big)^{m(\kappa)+4} \concat \$ $. With the first symbol $\$$ in $T'$ we go to $s_1$. We then walk to $s_{i_1}$ reading $\big(\#\concat CG_T(C', \kappa, \kappa') \big)^{i_1-1}$. With $\#$ we then step to the starting state of $CG^1(i_1)$, corresponding to clique $C_1$. Since $C_1 \cup C'$ forms a $(\kappa+\kappa')$-clique, we can walk to the ending state of $CG^1(i_1)$ reading $CG_T(C', \kappa, \kappa')$. Since $C_1 \cup C_2$ forms a $2\kappa$-clique, we can next step to the starting state of $CG^2(i_2)$ (corresponding to $C_2$).
Similarly, we can then walk through $CG^2(i_2)$, $CG^3(i_3)$ (corresponding to $C_3$), and $CG^4(i_1)$ (corresponding to $C_1$ again). %Until here we read $\big(\#\concat CG_T(C', \kappa, \kappa') \big)^{i_1-1+4}$. 
Next we step to $t_{i_1}$ reading $\#$, and then we simply walk to $t_{m(\kappa)}$ reading $\big(\#\concat CG_T(C', \kappa, \kappa') \big)^{m(\kappa)-i_1}$. Note that the number of times we read a symbol $\#$ is $i_1-1$ (for walking to $s_{i_1}$) plus 5 (for walking from $s_{i_1}$ to $t_{i_1}$) plus $m(\kappa) - i_1$ (for walking from $t_{i_1}$ to $t_{m(\kappa)}$), summing to $m(\kappa)+4$. Hence, indeed we parse all symbols $\#$ in $T'$. Thus, we can next step to $t$ reading the final symbol $\$$ of $T'$.
We then stay in $t$ reading the remainder of $T$. Since $t$ is accepting, we are done.

For the other direction, note that if $F$ accepts $T$ then it also accepts some substring $T' := \$ \concat \big(\#\concat CG_T(C', \kappa, \kappa') \big)^{m(\kappa)+4} \concat \$ $. Moreover, when reading $T'$ we must walk through some clique gadgets $CG^1(i_1), CG^2(i_2), CG^3(i_3)$, and $CG^4(i_4)$, corresponding to $\kappa$-cliques $C_1, C_2, C_3$, and $C_4$. Note that the number of symbols $\#$ on such a walk is $i_1 -1$ (for walking to $s_{i_1}$) plus 5 (for walking from $s_{i_1}$ to $t_{i_4}$) plus $m(\kappa)-i_4$ (for walking from $t_{i_4}$ to $t_{m(\kappa)}$), summing to $m(\kappa)+4+i_1-i_4$. Since $T'$ contains exactly $m(\kappa)+4$ symbols $\#$, we obtain $i_1 = i_4$ and thus $C_1 = C_4$. By the restrictions on the edges from $CG^r(i)$ to $CG^{r+1}(j)$ we see that $C_1 \cup C_2$, $C_2 \cup C_3$, and $C_3 \cup C_4 = C_3 \cup C_1$ form $2\kappa$-cliques. Moreover, since we walked through the clique gadgets we see that $C' \cup C_1$, $C' \cup C_2$, and $C' \cup C_3$ form $(\kappa + \kappa')$-cliques. In total, we obtain that $C_1 \cup C_2 \cup C_3 \cup C'$ forms a $(3\kappa + \kappa' = k)$-clique, finishing the correctness argument.

\paragraph{Size Bounds}
Note that clique gadgets $CG_T$ in the text have length $O(\log V)$, while the clique gadgets $CG_F$ in the automaton have $O(V \log V)$ states. We can thus read off a text length of $N = O(m(\kappa) m(\kappa') \log V) = O(V^{\kappa + \kappa'} \log V)$. Since the repetition $\big(\#\concat CG_T(C', \kappa, \kappa') \big)^{m(\kappa)+4}$ can be easily compressed to size $\Oh(\log V)$ by Observation~\ref{obs:repetition}, we obtain a compressed size of $n = O(m(\kappa') \log V) = O(V^{\kappa'} \log V)$. Finally, the number of states is $q = O(m(\kappa) V \log V) = O(V^{\kappa+1} \log V)$. Note also that the output of this reduction can be computed in time $O(n + q^2)$, i.e., in linear time in the output description. We thus obtain the desired reduction which, as argued in the beginning of this proof, rules out a combinatorial $O(\min\{n q^3, N q^2\}^{1-\eps})$ algorithm for \NFAAccept, assuming the combinatorial $k$-Clique conjecture.
 
\paragraph{Strengthening the Statement}
In the remainder, we verify that our construction proves the desired lower bound even restricted to instances with $n = \Theta(q^{\alpha_n})$ and $N = \Theta(q^{\alpha_N})$ for any $\alpha_N \ge \alpha_n > 0$.
Note that the number of states, the size of the SLP, and the text length can all three be increased by easy padding. E.g., to increase the text length we introduce a garbage symbol ``!'' that can be read at any state of the automaton, not changing the current state, and add a suitable number of copies of ``!'' to the text.
We now consider two cases. 

Case 1: If $\alpha_N \ge \alpha_n + 1$, then set $\kappa, \kappa' \ge 1$ such that $3\kappa + \kappa'=k$ and $\kappa \approx k/(\alpha_n+3)$ (recall that $\kappa,\kappa'$ are restricted to be integers). We can ensure that $\kappa < k/(\alpha_n+3) + 2$ and $\kappa' < \alpha_n k/(\alpha_n+3) + 3$. Note that for any $\eps > 0$, for sufficiently large $k=k(\eps,\alpha_n)$ we have $\kappa+1 < (1+\eps/2)k/(\alpha_n+3)$ and $\kappa' < (1+\eps/2) \alpha_n k/(\alpha_n+3)$. We can thus pad the number of states from $O(V^{\kappa+1} \log V)$ to $q = \Theta(V^{(1+\eps/2)k/(\alpha_n+3)})$ and the compressed size from $O(V^{\kappa'} \log V)$ to $n = \Theta(V^{(1+\eps/2) \alpha_n k/(\alpha_n+3)}) = \Theta(q^{\alpha_n})$. Similarly, for the decompressed text length, using $\alpha_N \ge \alpha_n + 1$, we have $N = O(V^{\kappa + \kappa'} \log V) = O(V^{(1+\eps/2) (\alpha_n + 1) k/(\alpha_n+3)}) = O(V^{(1+\eps/2) \alpha_N k/(\alpha_n+3)}) = O(q^{\alpha_N})$, which we can pad to equality. Then we indeed end up with an instance with $N = \Theta(q^{\alpha_N})$ and $n = \Theta(q^{\alpha_n})$. Hence, if \NFAAccept\ can be solved in combinatorial time $O(\min\{n q^3, N q^2\}^{1-\eps})$ restricted to such instances, then we obtain a combinatorial algorithm for $k$-Clique in time $O((n q^3)^{1-\eps}) = O(V^{(\alpha_n + 3) \cdot (1-\eps) (1+\eps/2)k/(\alpha_n+3)}) = O(V^{k(1-\eps/2)})$, contradicting the combinatorial $k$-Clique conjecture.

Case 2: If $\alpha_N < \alpha_n + 1$, then we have to slightly adapt the above construction. We introduce a third parameter $\hat \kappa \le \kappa$ and let the first and fourth column of clique gadgets $CG^1(i)$ and $CG^4(i)$ range over $\hat \kappa$-cliques. At the same time, we change the number of repetitions of each part in the text from $m(\kappa) + 4$ to $m(\hat \kappa)+4$. We are now detecting $(2 \kappa + \hat \kappa + \kappa')$-cliques in $G$. It can be checked that this does not violate the correctness of the construction. The new size bounds are $N = O(V^{\hat \kappa + \kappa'} \log V)$, $n = O(V^{\kappa'} \log V)$, and $q = O(V^{\kappa} \log V)$. 
Furthermore, we now allow to set $\kappa' = 0$, in which case the text is not responsible for choosing any part of the clique. Since in this case we do not need any clique gadgets, we define $CG_F(C, \kappa, 0)$ to consist of a single state $s=t$ and $CG_T(C', \kappa, 0)$ to be the empty string. In this case we set the final string to be $T := \$ \, \#^{m(\kappa)+4} \, \$$. The same correctness proof goes through.

We now choose integers $\kappa, \hat \kappa \ge 1$ and $\kappa' \ge 0$ with $2 \kappa + \hat \kappa + \kappa' = k$ and $\hat \kappa \le \kappa$ such that $\kappa \approx k/(\alpha_N+2)$, $\kappa' \approx \max\{0, (\alpha_N-1) k /(\alpha_N+2)\}$, and $\hat \kappa \approx \min\{\alpha_N,1\} \cdot k/(\alpha_N+2)$. Similarly to case 1, we can ensure for any $\eps > 0$ and sufficiently large $k=k(\eps,\alpha_N)$ that $\kappa < (1+\eps/2) k/(\alpha_N+2)$, $\kappa' \le \max\{0, (1+\eps/2) (\alpha_N-1) k /(\alpha_N+2)\}$, and $\hat \kappa < (1+\eps/2) \min\{\alpha_N,1\} \cdot k/(\alpha_N+2)$. We can thus pad the number of states to $q = \Theta(V^{(1+\eps/2) k/(\alpha_N+2)})$ and since $\hat \kappa + \kappa' < (1+\eps/2) \alpha_N k / (\alpha_N + 2)$ we can pad the decompressed text length to $N = \Theta(V^{(1+\eps/2) \alpha_N k/(\alpha_N + 2)}) = \Theta(q^{\alpha_N})$. For the compressed size, note that by the assumptions $\alpha_N < \alpha_n + 1$ and $\alpha_n > 0$ we have $\kappa' \le \max\{0, (1+\eps/2) (\alpha_N-1) k /(\alpha_N+2)\} < (1+\eps/2) \alpha_n k / (\alpha_N+2)$, and thus $n = O(V^{(1+\eps/2) \alpha_n k / (\alpha_N+2)}) = O(q^{\alpha_n})$, which we can pad to equality. Then we indeed end up with an instance with $N = \Theta(q^{\alpha_N})$ and $n = \Theta(q^{\alpha_n})$. Hence, if \NFAAccept\ can be solved in combinatorial time $O(\min\{n q^3, N q^2\}^{1-\eps})$ restricted to such instances, then we obtain a combinatorial algorithm for $k$-Clique in time $O((N q^2)^{1-\eps}) = O(V^{(\alpha_N + 2) \cdot (1-\eps) (1+\eps/2)k/(\alpha_N+2)}) = O(V^{k(1-\eps/2)})$, contradicting the combinatorial $k$-Clique conjecture.
\end{proof}

% !TEX root = main.tex

\subsection{Context-Free Grammar Parsing} \label{sec:cfg}

We again assume that the alphabet size $|\Sigma|$ is constant throughout this section.

In this section we show a strong conditional lower bound for context-free grammar parsing.
Recall that a context-free grammar (CFG) $\Gamma$ consists of a set of terminals $\Sigma$, a set of non-terminals $\Omega$, a starting non-terminal $S \in \Omega$, and a set of productions $\Phi$, each of the form $A \to \alpha$, where $A \in \Omega$ and $\alpha \in (\Sigma \cup \Omega)^*$. The size $|\Gamma|$ is the total length of all $\alpha$ over all productions. 
Applying a production $A \to \alpha$ to a string $\beta = \beta_1 A \beta_2 \in (\Sigma \cup \Omega)^*$ means to generate the string $\beta_1 \alpha \beta_2$. 
The language $L(\Gamma)$ is the set of strings in $\Sigma^*$ that can be generated by starting with $S$ and repeatedly applying productions. More generally, for any non-terminal $A$ the language $L(A)$ is the set of strings in $\Sigma^*$ that can be generated by starting with $A$.

\begin{problem}[\CFGrecognition]
Given a text $T$ of length $N$ by a grammar-compressed representation $\cT$ of size $n$ as well as a CFG $\Gamma$, decide whether $T \in L(\Gamma)$.
\end{problem}

(CFG parsing is an augmentation of this decision problem where in case $T \in L(\Gamma)$ we also need to return a sequence of productions as a certificate.)

As discussed in the introduction, after decompressing the text $T$ we can use classic parsers to solve CFG recognition in time $O(N^3 \poly(|\Gamma|))$~\cite{cocke1970programming,kasami1965efficient,younger1967recognition,earley1970efficient}, while Valiant's parser uses fast matrix multiplication to obtain an improved running time of $O(N^\omega \poly(|\Gamma|))$~\cite{valiant1975general}.\footnote{We ignore the specific polynomial dependence on $|\Gamma|$, since we are more interested in the dependence on $N$.} In the uncompressed setting, matching lower bounds based on the $k$-Clique conjecture were shown by Abboud et al.~\cite{ABV15b}.

In the compressed setting no improved algorithms are known, even for, say, $n = N^{0.01}$. 
Below we prove a matching lower bound for both running times $O(N^3)$ and $O(N^\omega)$, even restricted to very small grammars and quite compressible strings. Our proof differs considerably from the conditional lower bound in the uncompressed setting by Abboud et al.~\cite{ABV15b}, as their strings are not compressible in a strong sense. On a high level, their construction implements adjacency tests locally, around three chosen positions that encode three $k$-cliques. In our construction, we instead implement adjacency tests on a more global level, by choosing three offsets and reading all text positions that adhere to these offsets. This global view makes it possible to construct a compressible text.

\begin{thm} \label{thm:cfglowerbound}
  Assuming the $k$-Clique conjecture, there is no $O(N^{\omega - \eps})$ time algorithm for CFG recognition for any $\eps > 0$. Assuming the combinatorial $k$-Clique conjecture, there is no combinatorial $O(N^{3 - \eps})$ time algorithm for CFG recognition for any $\eps > 0$. Both results hold even restricted to instances with $|\Gamma| = O(\log N)$ and $n = O(N^\eps)$. 
\end{thm}

\begin{proof}
  Let $k \ge 1$ and let $G=(V,E)$ be a $k$-Clique instance. We will construct a CFG $\Gamma$ of size $O(\log V)$ and a text $T$ of length $N = O(V^{k+2})$ generated by an SLP $\cT$ of size $n = O(V^3)$ such that $T \in L(\Gamma)$ holds if and only if $G$ contains a $3k$-clique. Note that an $O(N^{\omega - \eps}) = O(N^{\omega(1-\eps/3)})$ algorithm for CFG recognition would then imply an algorithm for $3k$-Clique in time $O(V^{(k+2)\omega(1-\eps/3)})$, which for $k \ge 12/\eps$ is bounded by $O(V^{k(1+\eps/6)\omega(1-\eps/3)}) = O(V^{\omega k (1-\eps/3)})$, contradicting the $3k$-Clique conjecture. The argument for combinatorial algorithms is analogous. Moreover, we have $|\Gamma| = O(\log V) = O(\log N)$ and $n = O(V^3) = O(N^{3/(k+2)}) = O(N^\eps)$ for $k \ge 3/\eps$.\footnote{Strictly speaking, we need to pad the text length to $\Theta(V^{k+2})$ first. This can easily be accomplished by adding garbage to the text and garbage handling rules to the grammar.}
 
In our construction we enumerate all $k$-tuples of vertices $U = (u_1,\ldots,u_k)$. Choosing three such $k$-tuples $U_1, U_2, U_3$ we then need to check that (1) each $k$-tuple $U_i$ forms a $k$-clique and (2) each pair $U_i, U_j$ forms a biclique for $i \ne j$. We remark that it is indeed necessary to enumerate all $k$-tuples and not just, say, all $k$-cliques, as the $k$-tuples are much more structured, leading to compressible strings.
In the following we construct gadgets that perform these tests. We will use alphabet $\Sigma = \{0,1,\#,\$,x,y,z\}$.

\paragraph{Offsets}
Let $U(i)$ be the $i$-th $k$-tuple $(u_1,\ldots,u_k) \in V^k$ in lexicographic order.
Choosing a $k$-tuple thus correspond to choosing a number $1 \le i \le V^k$, which we will interpret as an offset in the text~$T$, resulting in \emph{relevant} positions of the form $i + V^k \cdot \mathbb{N}$. In order to only read the relevant positions, we need to implement jumping over $V^k - 1$ symbols, so that after reading one relevant symbol we can jump to the next one. To this end, we construct a non-terminal $X$ of $\Gamma$ with $L(X) = \Sigma^{V^k - 1}$. This can be build by constructing non-terminals $X_d$ with $L(X_d) = \Sigma^{2^d}$ by the productions
\begin{alignat*}{2}
  X_0 &\to \sigma &&\qquad \text{for any $\sigma \in \Sigma$}, \\
  X_d &\to X_{d-1} X_{d-1} &&\qquad \text{for $1 \le d \le \log(V^k - 1)$}.
\end{alignat*}
Then the production $X \to X_{i_1} \ldots X_{i_\ell}$, where $i_1,\ldots,i_\ell$ are the 1-bits in the binary encoding of $V^k - 1$, yields the desired non-terminal $X$. Note that this yields a grammar of size $O(\log V)$.

\paragraph{Clique Test}
We now design gadgets that allow to test for any offset $i$ whether $U(i)$ forms a $k$-clique.
Let $\bar E = \binom{V}{2} \setminus E$ be the non-edges of $G$. Let $[.]$ be the Kronecker symbol, i.e., $[\textup{true}] = 1$ and $[\textup{false}] = 0$. We use the following text:
\[ T_C := \$^{V^k} \concat \Big( \bigconcat_{\{u,v\} \in \bar E} \, \bigconcat_{1 \le i \le V^k} \big[ \text{$u$ and $v$ appear in $U(i)$}\big] \Big) \concat \$^{V^k}. \]

For any offset $1 \le i \le V^k$, if $U(i) = (u_1,\ldots,u_k)$ forms a $k$-clique then no non-edge appears among $\{u_1,\ldots,u_k\}$, and thus $T_C[i + j\cdot V^k] = 0$ for all $1 \le j \le |\bar E|$. The opposite implication holds as well. This leads us to testing for a $k$-clique via the following CFG rules:
\[ C \;\;\to\;\; \$ X \tilde C, \qquad \tilde C \;\; \to \;\; 0\, X\, \tilde C \quad | \quad \$. \]

\begin{lem} \label{lem:cfgcorrectone}
  We call $T_C(i) := T_C[i..i + 1 + (|\bar E|+1) \cdot V^k]$ for $1 \le i \le V^k$ the \emph{valid} substrings of $T_C$.
  Any substring of $T_C$ that is parsable by $C$ is valid. Moreover, substring $T_C(i)$ is parsable by $C$ if and only if the $k$-tuple $U(i)$ forms a $k$-clique in $G$.
\end{lem}
\begin{proof}
  The first statement follows by $C$ starting and ending with a $\$$ symbol and advancing by $V^k - 1$ steps via $X$. The second statement follows from the argument above this lemma.
\end{proof}

\begin{lem} \label{lem:cfgslpsize}
  The string $T_C$ has an SLP of size $O(V^3)$.
\end{lem}
\begin{proof}
  For any $1 \le d \le k$, $\sigma \in \Sigma = \{0,1,\$,\#,x,y,z\}$, and $S \subseteq V$ with $|S| \le 2$ we define the following SLP rules:
  \begin{align*}
    \textup{Repeat}^{(0)}_{\sigma} &\to \sigma, \\
    \textup{Repeat}_{\sigma}^{(d)} &\to \bigconcat_{v \in V} \textup{Repeat}^{(d-1)}_{\sigma}, \\
    \textup{Incl}_S^{(0)} &\to \begin{cases} 1, \text{ if } S = \emptyset, \\ 0, \text{ otherwise} \end{cases} \\
    \textup{Incl}_S^{(d)} &\to \bigconcat_{v \in V} \textup{Incl}^{(d-1)}_{S \setminus \{v\}}, \\
    \textup{C-Test} &\to \textup{Repeat}^{(k)}_\$ \concat \Big( \bigconcat_{\{u,v\} \in \bar E} \textup{Incl}^{(k)}_{\{u,v\}} \Big) \concat \textup{Repeat}^{(k)}_\$.
  \end{align*}
  
  We claim that $\eval(\textup{C-Test}) = T_C$. Note that $\textup{Repeat}_{\sigma}^{(d)}$ generates the string $\sigma^{V^d}$, and thus the prefix and suffix $\$^{V^k}$ is correct. Further, it can be checked that $\textup{Incl}_S^{(d)}$ generates a string of length $V^d$ where the $i$-th position, corresponding to a $d$-tuple $(u_1,\ldots,u_d) \in V^d$, is 1 if $S \subseteq \{u_1,\ldots,u_d\}$ and 0 otherwise. Hence, writing the string $\textup{Incl}_{\{u,v\}}^{(k)}$ for all $\{u,v\} \in \bar E$ yields the middle part of the string $T_C$. This proves the claim.
  
  Note that the total size of the above SLP for $T_C$, i.e., the total number of symbols on the right hand sides of the above rules, is indeed $O(V^3)$.
\end{proof}

\paragraph{Biclique Test}
We next design gadgets that allow us to test for two offsets $i,j$ whether $u \sim v$ for all $u \in U(i), v \in U(j)$, i.e., whether $U(i), U(j)$ form a biclique.
To this end, we let $V^\rev$ be the reverse ordering of the vertices in $V$ and define the texts
\begin{align*}
  T_B &:= \#^{V^k} \concat \Big( \bigconcat_{u \in V} \, \bigconcat_{1 \le i \le V^k} \big[ \text{$u$ appears in $U(i)$}\big] \Big) \concat \#^{V^k}, \\
  T_B' &:= \#^{V^k} \concat \Big( \bigconcat_{u \in V^\rev} \, \bigconcat_{1 \le i \le V^k} \big[ \text{$u$ is adjacent to every vertex in $U(i)$}\big] \Big) \concat \#^{V^k}.
\end{align*}

Note that $U(i),U(j)$ form a biclique if every vertex that appears in $U(i)$ is adjacent to every vertex in $U(j)$. Thus, for every $1 \le \ell \le V$ we want that if $T_B[i + \ell \cdot V^k] = 1$ then also $T_B'[j + (V+1-\ell) \cdot V^k] = 1$. This leads us to testing for a biclique via the following CFG rules:
\begin{align*}
  B_{\textup{in}} \;\;&\to\;\; \#\, X\, B\, X\, \# \\
  B \;\;&\to\;\; 1\, X\, B\, X\, 1 \quad | \quad 0\, X\, B\, X\, 1 \quad | \quad 0\, X\, B\, X\, 0 \quad | \quad \# B_\textup{out} \#
\end{align*}
We view this part of the grammar as a subroutine that is started by invoking $B_{\textup{in}}$ and that can be followed by further operations by adding productions starting from $B_\textup{out}$. 
Note that each call of a rule of $B_{\textup{in}}$ or $B$ reads $V^k$ symbols from the left and from the right, except for the last one, which reads 1 symbol from the left and from the right. That is, the offsets are never changed throughout the parsing process. The parsing rules check that a 1 at a certain position in $T_B$ implies a 1 at the corresponding position in $T_B'$. Hence, when starting with offsets $i$ in $T_B$ and $j$ in $T_B'$, this process checks that $U(i),U(j)$ form a biclique. It stops when we reach the $\#$-blocks at the end of $T_B$ and at the beginning of $T_B'$, where we exit to $B_\textup{out}$. Then it depends on the (not yet defined) productions involving $B_\textup{out}$ whether the remainder of the string can be parsed. In summary, we obtain the following.

\begin{lem} \label{lem:cfgcorrecttwo}
  We call $T_B(i) := T_B[i .. i+1+ (V+1)\cdot V^k]$ and $T_B'(j) := T_B'[j .. j + 1 + (V+1)\cdot V^k]$ for $1 \le i,j \le V^k$ the \emph{valid} substrings of $T_B$ and $T_B'$, respectively. 
  Let $R$ be any string. 
  Then $B_{\textup{in}}$ can parse $T_B(i)\, R\, T_B'(j)$ if and only if $U(i),U(j)$ form a biclique and $B_\textup{out}$ can parse $R$. Moreover, if $\tilde T_B$ and $\tilde T_B'$ are substrings of $T_B$ and $T_B'$, respectively, and $B_\textup{in}$ can parse $\tilde T_B\, R\, \tilde T_B'$ such that $B_\textup{out}$ parses $R$, then $\tilde T_B$ and $\tilde T_B'$ are valid.
\end{lem}

\begin{lem} \label{lem:cfgslpboundtwo}
  The strings $T_B$ and $T_B'$ have SLPs of size $O(V^2)$.
\end{lem}
\begin{proof}
  Note that $T_B$ is the string generated by the following SLP, where we use notation as in Lemma~\ref{lem:cfgslpsize}:
  \[ \textup{B-Test} \to \textup{Repeat}^{(k)}_\# \concat \Big( \bigconcat_{v \in V} \textup{Incl}^{(k)}_{\{v\}} \Big) \concat \textup{Repeat}^{(k)}_\#. \]
  This has size $O(V^2)$ as shown in the proof of Lemma~\ref{lem:cfgslpsize}. 
  
  For $T_B'$ we use the following SLP rules for $1 \le d \le k$ and $v \in V$:
  \begin{align*}
    \textup{Adj}^{(0)}_v &\to 1, \\
    \textup{Adj}^{(d)}_v &\to \bigconcat_{u \in V} \begin{cases} \textup{Adj}^{(d-1)}_v, \text{ if } \{u,v\} \in E \\ \textup{Repeat}^{(d-1)}_0, \text{ otherwise} \end{cases} \\
    \textup{B'-Test} &\to \textup{Repeat}^{(k)}_\# \concat \Big( \bigconcat_{v \in V^\rev} \textup{Adj}^{(k)}_{v} \Big) \concat \textup{Repeat}^{(k)}_\#
  \end{align*}
  An easy inductive proof shows that $\textup{Adj}_v^{(d)}$ generates a string of length $V^d$ where the $i$-th position, corresponding to a $d$-tuple $(u_1,\ldots,u_d) \in V^d$, is 1 if $v$ is adjacent to every $u_i$, and 0 otherwise. Hence, writing $\textup{Adj}^{(k)}_{v}$ for all $v \in V$ (in reverse order) yields the middle part of $T_B'$, and thus $\textup{B'-Test}$ generates $T_B'$.
  Again, the total size of the right hand sides is $O(V^2)$, so the SLP has size $O(V^2)$.
\end{proof}

\paragraph{Complete Construction}
The final string is
\[ T := x^{V^k}\, T_C\, T_B\, T_B\, T_B'\, y^{V^k}\, T_C\, T_B\, T_B'\, T_B'\, T_C\, z^{V^k}. \]
Here, the parts $x^{V^k}, y^{V^k}$, and $z^{V^k}$ are used to choose three offsets $i_1,i_2,i_3$, corrsponding to three $k$-tuples $U(i_1),U(i_2),U(i_3)$. The three copies of $T_C$ are used to check that each $U(i_j)$ forms a $k$-clique. The left copy of $T_B T_B'$ is used for checking that $U(i_1),U(i_2)$ forms a biclique, similarly for the right copy and $U(i_2), U(i_3)$. Finally, the leftmost $T_B$ and rightmost $T_B'$ are used to check that $U(i_1),U(i_3)$ form a biclique. 
Note that $T$ uses alphabet $\Sigma = \{0,1,\#,\$,x,y,z\}$. 

We now describe the final grammar $\Gamma$. We copy the non-terminals $B_\textup{in}, B, B_\textup{out}$ to $\tilde B_\textup{in}, \tilde B, \tilde B_\textup{out}$, since we need this subroutine twice with different productions starting from $B_\textup{out}$. We let $S$ be a new starting symbol and define the following productions, additional to the ones defined above:
\begin{align*}
  S \;\;&\to\;\; x\, S \quad | \quad S\, z \quad | \quad X\, C\, X\, B_\textup{in}\, X\, C\, X \\
  B_\textup{out} \;\;&\to\;\; X\, \tilde B_\textup{in}\, X\, y\, X\, C\, X\, \tilde B_\textup{in}\, X \\
  \tilde B_\textup{out} \;\;&\to\;\; \#\, \tilde B_\textup{out} \quad | \quad \epsilon,
\end{align*}
where $\epsilon$ denotes the empty string. This finishes the construction of the CFG recognition instance.

\paragraph{Correctness}
We show that $T \in L(\Gamma)$ holds if and only if there is a $3k$-clique in $G$.
Assume that $G$ contains a $3k$-clique and let $1 \le i_1,i_2,i_3 \le V^k$ be such that $U(i_1) \cup U(i_2) \cup U(i_3)$ forms a $3k$-clique. 
Remove $i_1$ symbols $x$ from the left end of $T$ and $V^k - i_3 + 1$ symbols $z$ from the right, leaving offsets $i_1$ and $i_3$, respectively. Then apply the rule $S \to X\, C\, X\, B_\textup{in}\, X\, C\, X$. The outer calls to $X$ keep the offsets $i_1$ and $i_3$ by advancing to the next relevant positions w.r.t.\ offsets $i_1$ and $i_3$, respectively.
By Lemma~\ref{lem:cfgcorrectone}, the calls of $C$ parse valid substrings of $T_C$ starting and ending with offset $i_1$ and $i_3$, respectively. 
The lemma is applicable since $U(i_1)$ and $U(i_3)$ form $k$-cliques. 
The further calls to $X$ again advance to the next relevant positions w.r.t.\ offsets $i_1$ and $i_3$, now lying in the outer $\#$-blocks in the leftmost $T_B$ and rightmost $T_B'$, respectively. Finally, by Lemma~\ref{lem:cfgcorrecttwo} the call to $B_\textup{in}$ reads valid substrings of the leftmost $T_B$ and rightmost $T_B'$ and ends with $B_\textup{out}$. 
The lemma is applicable since $U(i_1),U(i_3)$ forms a biclique.
The outer calls to $X$ in the rule $B_\textup{out} \to X\, \tilde B_\textup{in}\, X\, y\, X\, C\, X\, \tilde B_\textup{in}\, X$ then advances the left and right end to the first relevant position w.r.t.\ offset $i_1$ in the second copy of $T_B$ and the last relevant position w.r.t.\ offset $i_3$ in the second-to-last copy of $T_B'$. 
We match the $y$ appearing in this rule to the $i_2$-th $y$ in the $y^{V^k}$ part of $T$. To the right of $y$, $C$ parses a valid substring of $T_C$, which works since $U(i_2)$ forms a $k$-clique. The remaining $\tilde B_\textup{in}$ then has to parse valid substrings of the right copy of $T_B T_B'$, starting with offset $i_2$ and ending with offset $i_3$. Similarly, to the left of $y$, $\tilde B_\textup{in}$ has to parse valid substrings of the left copy of $T_B T_B'$, starting with offset $i_1$ and ending with offset $i_2$. This works as $U(i_1), U(i_2)$ and $U(i_2),U(i_3)$ form bicliques. Note that after reaching $\tilde B_\textup{out}$ we are left with some symbols of the last $\#$-block of $T_B$ and some symbols of the first $\#$-block of $T_B'$. Both can be parsed completely using the rules involving $\tilde B_\textup{out}$. Thus, we have $T \in L(\Gamma)$.

For the other direction, we follow the same line of arguments, observing that there was no choice except for the offsets $i_1,i_2,i_3$. The core of the argument is that $U(i_1) \cup U(i_2) \cup U(i_3)$ forms a $3k$-clique if and only if each $U(i_j)$ forms a $k$-clique and each pair $U(i_j),U(i_{j'})$ forms a biclique.

\paragraph{Size Bounds}
Since $T$ consists of $O(|V| + |\bar E|)$ parts of length $V^k$, the text length is $O(V^{k+2})$. By Lemmas~\ref{lem:cfgslpsize} and \ref{lem:cfgslpboundtwo} and since $x^{V^k}$ has an SLP of size $O(\log V)$, $T$ has an SLP of size $O(V^3)$. Finally, the size of the grammar $\Gamma$ is $O(\log V)$, the bottleneck being the non-terminal $X$ that ensures offset consistency. Hence, all claimed size bounds are met. Note also that the constructed instance can be computed in time linear in the output size.
This finishes the proof of Theorem~\ref{thm:cfglowerbound}.
\end{proof}

\newcommand{\RNA}{\textup{RNA}}
\newcommand{\WRNA}{\textup{WRNA}}

\subsection{RNA Folding} \label{sec:rna}

We now give a variant of the construction for CFG recognition, proving a matching conditional lower bound for RNA folding. 

Again we consider a constant-size aphabet $\Sigma$, however, now each symbol $\sigma \in \Sigma$ has a unique counterpart $\bar \sigma \in \Sigma$ such that $\bar {\bar \sigma} = \sigma$. We say that $\sigma \in \Sigma$ and its counterpart $\bar \sigma$ \emph{match}.

Two pairs of indices $(i,j), (i',j')$ with $i < j$ and $i' < j'$ are said to \emph{cross} if at least one of the following conditions holds: (1) $i = i'$ or $i = j'$ or $j = i'$ or $j = j'$, (2) $i < i' < j < j'$, or (3) $i' < i < j' < j$. 
In other words, $(i,j), (i',j')$ with $i < j$ and $i' < j'$ are non-crossing if they are disjoint, i.e., $i < j < i' < j'$ or $i' < j' < i < j$, or they are nesting, i.e., $i < i' < j' < j$ or $i' < i < j < j'$. 

\begin{problem}[\RNAfolding]
Given a text $T$ of length $N$ by a grammar-compressed representation $\cT$ of size $n$, compute the maximum number of pairs $R \subseteq \{(i,j) \mid 1 \le i < j \le N\}$ such that for every $(i,j) \in R$ the symbols $T[i]$ and $T[j]$ match and there are no crossing pairs in $R$. We denote this maximum number by $\RNA(T)$.
\end{problem}

We refer to the set $R$ as a \emph{matching} of $T$.

In the uncompressed setting, \RNAfolding\ has an easy dynamic programming solution in time $O(N^3)$~\cite{eddy2004rna}. Using fast matrix multiplication, this was recently improved to $O(N^{2.82})$~\cite{bringmann2016truly}. For combinatorial algorithms, a matching lower bound of $N^{3-o(1)}$ assuming the combinatorial $k$-Clique conjecture was recently shown by Abboud et al.~\cite{ABV15b}. They also prove a conditional lower bound of $N^{\omega-o(1)}$ assuming the $k$-Clique conjecture, however, this leaves a gap to the current upper bound. 

As for CFG parsing, no improved algorithms are known in the compressed setting, even for, say, $n = N^{0.01}$. Here we prove lower bounds of $N^{3-o(1)}$ for combinatorial algorithms and $N^{\omega-o(1)}$ in general, assuming the (combinatorial) $k$-Clique conjecture.

\begin{thm} \label{thm:rnalowerbound}
  Assuming the $k$-Clique conjecture, there is no $O(N^{\omega - \eps})$ time algorithm for \RNAfolding\ for any $\eps > 0$. Assuming the combinatorial $k$-Clique conjecture, there is no combinatorial $O(N^{3 - \eps})$ time algorithm for \RNAfolding\ for any $\eps > 0$. Both results hold even restricted to instances with $n = O(N^\eps)$. 
\end{thm}

Abboud et al.~\cite{ABV15b} showed that \RNAfolding\ is equivalent to the following weighted variant.

\begin{problem}[\weightedRNAfolding]
We are given a text $T$ of length $N$ by a grammar-compressed representation $\cT$ of size $n$ as well as a weight function $w\colon \Sigma \to [M]$ with $w(\sigma) = w(\bar \sigma)$ for all $\sigma \in \Sigma$. For any set $R \subseteq \{(i,j) \mid 1 \le i < j \le N\}$ define its weight as $\sum_{(i,j) \in R} w(T[i])$. Compute the maximum weight of any set $R$ such that for every $(i,j) \in R$ the symbols $T[i]$ and $T[j]$ match and there are no crossing pairs in $R$. We denote this maximum weight by $\WRNA(T)$.
\end{problem}

\begin{lem}[Lemma 2 in \cite{ABV15b}] \label{lem:equweightedrna}
  For an instance $T$ of \weightedRNAfolding, consider the string $\tilde T := T[1]^{w(T[1])} \ldots T[N]^{w(T[n])}$, i.e., each symbol $T[i]$ is repeated $w(T[i])$ times. Then we have $\WRNA(T) = \RNA(\tilde T)$.
\end{lem}

\begin{proof}[Proof of Theorem~\ref{thm:rnalowerbound}]
Let $k \ge 1$ and let $G=(V,E)$ be a $k$-Clique instance. We will construct a \weightedRNAfolding\ instance $T$ of length $O(V^{k+2})$ (and $\Omega(V^k)$) generated by an SLP $\cT$ of size $O(V^{3})$ and a number $\lambda$ such that $\WRNA(T) \ge \lambda$ holds if and only if $G$ contains a $3k$-clique.
The alphabet size will be $|\Sigma| = 48$ and the weights are bounded by $O(V^{2})$. 
By Lemma~\ref{lem:equweightedrna}, the corresponding unweighted text $\tilde T$ has $\RNA(\tilde T) = \WRNA(T)$ and thus $\RNA(\tilde T) \ge \lambda$ holds if and only if $G$ contains a $3k$-clique. Moreover, since the weights in $T$ are bounded by $O(V^{2})$ we have $N = |\tilde T| = O(V^{2} |T|) = O(V^{k+4})$. Finally, by compressing $O(V^{2})$ repetitions to $O(\log V)$ SLP rules, $\tilde T$ has an SLP $\tilde \cT$ of size $n = O(|\cT| \log V) = O(V^{3} \log V)$. 

Hence, an $O(N^{\omega-\eps}) = O(N^{\omega(1-\eps/3)})$ algorithm for \RNAfolding\ would imply an algorithm for $3k$-Clique in time $O(V^{(k+4)\omega(1-\eps/3)})$, which for $k \ge 24/\eps$ is bounded by $O(V^{k(1+\eps/6)\omega(1-\eps/3)}) = O(V^{\omega k (1-\eps/6)})$, contradicting the $3k$-Clique conjecture. The argument for combinatorial algorithms is analogous. Moreover, we have $n = O(V^{3} \log V) = O(N^{3/k} \log V) = O(N^\eps)$ for $k \ge 3/\eps$.

To construct the desired instance $T,\cT$ of \weightedRNAfolding, we again enumerate all $k$-tuples $U(i)$ for $1 \le i \le V^k$, as in the proof for CFG parsing. We again choose three such $k$-tuples $U(i_1),U(i_2),U(i_3)$ and check that each $U(i_j)$ forms a $k$-clique and all pairs $U(i_j),U(i_{j'})$ form a biclique for $j \ne j'$. 

\paragraph{Clique Test}
Consider alphabet $\{0,\bar 0,1,\bar 1\}$ (with weights 1) and set for $e \in \bar E$ and $1 \le i \le V^k$
\[ r_{e,i} := \begin{cases} \bar 1, \;\text{if some node in $e$ does not appear in $U(i)$} \\ \bar 0, \;\text{otherwise} \end{cases} \]
%Similarly define strings $x'_{e,i}, x''_{e,i}$, with the symbols $\bar 1, \bar 0$ replaced by $\bar 1', \bar 0'$ and $\bar 1'',\bar 0''$, respectively.
Since $U(i)$ forms a $k$-clique iff for every non-edge at least one of the endpoints does not appear in $U(i)$, we obtain:
\begin{lem} \label{lem:rnaclique}
  Set $r_i := \bigconcat_{e \in \bar E} r_{e,i}$. We have $\WRNA\big( 1^{\bar E} r_i \big) \le \bar E$, with equality if and only if $U(i)$ forms a $k$-clique. %\arturs{I think that it should be $\WRNA\big( 0^{\bar E} r_i \big)$ instead of $\WRNA\big( 1^{\bar E} r_i \big)$. Is this correct?} \karl{Right, I changed the definition of $r_{e,i}$}
\end{lem}

\paragraph{Biclique Test}
Consider alphabet $\{2,\bar 2,3, \bar 3,4,\bar 4\}$ (with weights 1) and set for $v \in V$ and $i \in [V^k]$
\[ p_{v,i} := \begin{cases} 2\, 4, \;\;\;\, \text{if $v$ appears in $U(i)$} \\ 2\, 3\, 4, \; \text{otherwise} \end{cases} 
\qquad q_{v,i} := \begin{cases} \bar 2\, \bar 4, \; \text{if $v$ is adjacent to every node in $U(i)$} \\ \bar 3\, \bar 4, \; \text{otherwise} \end{cases} \]
%Similarly define $y'_{v,i},z'_{v,i}$ and $y''_{v,i},z''_{v,i}$ on disjoint alphabets (by priming all symbols in the alphabet).
\begin{lem} \label{lem:rnabiclique}
  Set $p_i := \bigconcat_{v \in V} p_{v,i}$ and $q_i := \bigconcat_{v \in V} q_{v,i}$. For any $i,j$, we have $\WRNA(p_i\, q_j) \le 2V$, with equality if and only if $U(i),U(j)$ form a biclique.
\end{lem}
\begin{proof}
  Note that the total weight of $q_j$ is $2V$, which shows the upper bound $\WRNA(p_i\, q_j) \le 2V$.
  To obtain equality, all symbols in $q_j$ must be matched. In particular, the $\bar 4$ in $q_{v,j}$ must be matched to the 4 in $p_{v,i}$. If follows that the $\bar 2$ or $\bar 3$ in $q_{v,j}$ can only be matched to a 2 or 3 in $p_{v,i}$. Hence, we have $\WRNA(p_i\, q_j) = 2V$ if and only if there is no $v \in V$ such that $v$ appears in $U(i)$ but $v$ is not adjacent to every node in $U(j)$, which happens if and only if $U(i),U(j)$ form a biclique.
\end{proof}

\paragraph{Complete Construction}
For any symbol $\sigma$ used so far, we introduce two copies $\sigma'$ and $\sigma''$. For the strings $r_{e,i},p_{v,i},q_{v,i}$ defined above, we write $r'_{e,i},p'_{v,i},q'_{v,i}$ and $r''_{e,i},p''_{v,i},q''_{v,i}$ to denote that we replace all symbols by their primed copies. 
For $i_1,i_2,i_3 \in [V^k]$ consider the string
\[ T(i_1,i_2,i_3) := 1^{\bar E}\, r_{i_1}\, p_{i_1}\, p'_{i_1}\; 1'^{\bar E}\, r'_{i_2}\, q'_{i_2}\, p''_{i_2}\; 1''^{\bar E}\, r''_{i_3}\, q''_{i_3}\, q_{i_3}. \]
Note that the alphabet is partitioned such that the only possible matchings are among $1^{\bar E}\, r_{i_1}$, $1'^{\bar E}\, r'_{i_2}$, $1''^{\bar E}\, r''_{i_3}$ as well as $p_{i_1}\, q_{i_3}$, $p'_{i_1}\, q'_{i_2}$, $p''_{i_2}\, q''_{i_3}$. Also note that these pairs are non-crossing. Hence, by Lemmas~\ref{lem:rnaclique} and \ref{lem:rnabiclique}, we have $\WRNA(T(i_1,i_2,i_3)) \le 6 V + 3 \bar E$, with equality if and only if $U(i_j)$ forms a $k$-clique and $U(i_j),U(i_{j'})$ form a biclique for any $j \ne j'$, which happens if and only if $U(i_1) \cup U(i_2) \cup U(i_3)$ forms a $3k$-clique. 

This is close to a complete reduction. It remains to force the choice of consistent offsets $i_1,i_2,i_3$, which we accomplish with the following lemma. Its proof is technical and defered to the end of this section.

\begin{lem} \label{lem:rnatechnical}
  Let $A,B,W \ge 1$. Let $x_{a,b}$ for $a \in [A],\, b \in [B]$ be strings over alphabet $\Sigma$, each with total weight $\sum_i w(x_{a,b}[i]) \le W$. 
  Assume that no two symbols in $\bigconcat_{a,b} x_{a,b}$ match.
  Let $5,\bar 5,6,\bar 6,7,\bar 7$ be new symbols not appearing in $\Sigma$, with weights $w(5) = w(\bar 5) = w(7) = w(\bar 7) = 4 A W$ and $w(6) = w(\bar 6) = 8 A W$. Set $\rho := (8A+12) A B W$ and
  \[ G(\{x_{a,b}\}) := 5^B (6\, \bar 5)^B \concat \Big( \bigconcat_{a \in [A]} \Big( \bigconcat_{b \in [B]} \bar 6\, x_{a,b} \Big) \concat 6^B \Big) \concat \bar 6^B (7\, 6)^B \bar 7^B. \]
  Then for any strings $y_1,y_2$ over alphabet $\Sigma$ we have
  \[ \WRNA(y_1\, G(\{x_{a,b}\}) \, y_2) = \rho + \max_{b \in [B]} \WRNA\Big(y_1 \concat \big(\bigconcat_{a \in [A]} x_{a,b} \big) \concat y_2\Big). \]
\end{lem}

We apply the above lemma as follows. Let $B = V^k$ and $A = V + 2\bar E$, and for $b \in [B]$ set $x_{a,b} := r_{a,b}$ for $a \in [\bar E]$, $x_{\bar E + a,b} := p_{a,b}$ for $a \in [V]$, and $x_{\bar E + V + a,b} := p'_{a,b}$ for $a \in [V]$. Note that $\bigconcat_{a \in [A]} x_{a,i_1} = r_{i_1} \, p_{i_1}\, p'_{i_1}$, which is a substring of $T(i_1,i_2,i_3)$. Construct $G(\{x_{a,b}\})$. Similarly define $y_{a,b}$ so that $\bigconcat_{a \in [A]} y_{a,i_2} = r'_{i_2}\, q'_{i_2}\, p''_{i_2}$, and construct $G'(\{y_{a,b}\})$, where the new symbols are now $5',\bar 5',6',\bar 6',7',\bar 7'$. Similarly define $z_{a,b}$ so that $\bigconcat_{a \in [A]} z_{a,i_3} = r''_{i_3}\, q''_{i_3}\, q_{i_3}$, and construct $G''(\{z_{a,b}\})$, where the new symbols are now $5'',\bar 5'',6'',\bar 6'',7'',\bar 7''$. 

The final text is
\[ T := 1^{\bar E}\, G(\{x_{a,b}\}) \; 1'^{\bar E}\, G'(\{y_{a,b}\})\; 1''^{\bar E}\, G''(\{z_{a,b}\}). \]
Applying Lemma~\ref{lem:rnatechnical} three times, we see that
\[ \WRNA(T) = 3 \rho + \max_{i_1,i_2,i_3 \in [V^k]} \WRNA(T(i_1,i_2,i_3)). \]
Since $\WRNA(T(i_1,i_2,i_3)) \le 6 V + 3 \bar E$ with equality if and only if $U(i_1) \cup U(i_2) \cup U(i_3)$ forms a $3k$-clique, we obtain that $\WRNA(T) \ge 3\rho + 6 V + 3 \bar E$ if and only if $G$ contains a $3k$-clique. This finishes the construction and proves the correctness.

\paragraph{Size Bounds}
Note that for each symbol $\sigma \in \{0,1,\ldots,7\}$ we have a counterpart $\bar \sigma$, and both have three primed variants. Thus, the alphabet size is $|\Sigma| = 8 \cdot 2 \cdot 3 = 48$. Since $A = O(V + \bar E) = O(V^2)$ and $B = V^k$, the text length is $N = O(V^{k+2})$. Note that each $x_{a,b}, y_{a,b}$, and $z_{a,b}$ has total weight $W \le 3$. Hence, the weight of the symbols introduced by the guarding $G(.)$ is $8 A W = O(A) = O(V^2)$. The following lemma analyzes the compressibility of the constructed text. We thus obtain all size bounds as claimed in the beginning of this proof.

\begin{lem}
  The text $T$ has an SLP $\cT$ of size $O(V^3)$.
\end{lem}
\begin{proof}
  As in Lemmas~\ref{lem:cfgslpsize} and \ref{lem:cfgslpboundtwo}, for any $a \in A$ there are SLPs for the strings $\bigconcat_{b \in [B]} \bar 6\, x_{a,b}$, $\bigconcat_{b \in [B]} \bar 6\, y_{a,b}$, and $\bigconcat_{b \in [B]} \bar 6\, z_{a,b}$ of size $O(V)$. Indeed, any such string is equal to $\bigconcat_{i \in [V^k]} \bar 6\, r_{e,i}$, $\bigconcat_{i \in [V^k]} \bar 6\, p_{v,i}$, or $\bigconcat_{i \in [V^k]} \bar 6\, q_{v,i}$, or their primed variants, for some $v \in V, e \in \bar E$. By definition of $r_{e,i},p_{v,i},q_{v,i}$, these strings are generated by $\textup{Incl}^{(k)}_{e}$, $\textup{Incl}^{(k)}_{v}$, and $\textup{Adj}^{(k)}_v$, respectively, except that the terminals $0,1$ are replaced by some constant-length strings over $\{0,\bar 0, \ldots, 4, \bar 4\}$. The final text $T$ consists of $O(A) = O(V^2)$ strings of the form $\bigconcat_{b \in [B]} \bar 6\, x_{a,b}$, $\bigconcat_{b \in [B]} \bar 6\, y_{a,b}$, or $\bigconcat_{b \in [B]} \bar 6\, z_{a,b}$, plus some very repetetive padding strings that can be compressed to length $O(\log V)$ by Observation~\ref{obs:repetition}. The bound follows.
\end{proof}

It remains to prove Lemma~\ref{lem:rnatechnical} to finish the proof of Theorem~\ref{thm:rnalowerbound}.

\begin{proof}[Proof of Lemma~\ref{lem:rnatechnical}]
  Let $x := G(\{x_{a,b}\})$ and fix $b \in [B]$. In every block $\bigconcat_{b \in [B]} \bar 6\, x_{a,b}$ or $\bar 6^B$ of $x$, we match the first $b$ $\bar 6$'s to the directly preceeding 6's, and match the last $B-b$ $\bar 6$'s to the directly succeeding 6's. At the beginning, this leaves $B-b$ $\bar 5$'s to be matched to the first 5's, and at the end this leaves $b$ 7's to be matched to the last $\bar 7$'s. Since we match all $(A+1)B$ $\bar 6$'s and $B-b$ $\bar 5$'s and $b$ 7's, the total weight of this matching is $(A+1)B \cdot 8 A W + (b + (B-b)) \cdot 4 A W = \rho$.  
  Note that this matching leaves all $x_{a,b}$ for $a \in A$ unmatched and uncovered, i.e., for no two matched symbols $x[i],x[j]$ we have that $x[i]$ is to the left of $x_{a,b}$ and $x[j]$ is to the right of $x_{a,b}$ in $x$. Hence, any solution to $\WRNA(y_1 \concat (\bigconcat_{a \in [A]} x_{a,b} ) \concat y_2)$ can be added to the pairs matched so far. This yields 
  \[ \WRNA(y_1\, x\, y_2) \ge \rho + \max_{b \in [B]} \WRNA\Big(y_1 \concat \big(\bigconcat_{a \in [A]} x_{a,b} \big) \concat y_2\Big). \]
  
  For the other direction, consider an optimal matching $R$ of $y_1\, x\, y_2$, realizing $\WRNA(y_1\, x\, y_2)$. 
  Write $w_x$ for the total weight of pairs in $R$ with both indices in $x$, and let $w_{x,y}$ be the total weight of pairs in $R$ with one end in $x$ and the other in $y_1$ or $y_2$. 
  Note that $w_x + w_{x,y} \ge \rho$, since otherwise, as shown above, we could replace the pairs of $R$ incident with $x$ to obtain $w_x = \rho$ and $w_{x,y} = 0$, yielding a higher total weight, which contradicts optimality of $R$.
  
  Note that symbols in $x_{a,b}$ can only be matched to symbols in $y_1$ or $y_2$, and the only possible matchings between $x$ and $y_1$ or $y_2$ happen in the strings $x_{a,b}$.
  Let $Z \subseteq [A] \times [B]$ be the set of all pairs $(a,b)$ such that $x_{a,b}$ contains at least one position matched by $R$. 
  Consider first the case $Z = \emptyset$, so that $w_{x,y} = 0$. Denote by $m_5,m_6,m_7$ the number of matched symbols $5,6,7$ in $x$. Note that each matched $\bar 5$ and each matched $7$ covers one $6$. Hence, at most $B-m_5 + B-m_7 + AB$ 6's can be matched. Since the number of $\bar 6$'s is $(A+1)B$, we have $m_6 \le \min\{B-m_5 + B-m_7 + AB, (A+1)B\}$. We thus obtain an upper bound on $w_x$ of 
  \begin{align*} 
    (m_5 + 2m_6 + m_7) \cdot 4AW &\le (m_5 + m_7) \cdot 4AW + \min\{B-m_5 + B-m_7 + AB, (A+1)B\} \cdot 8AW \\
    &= \min\{ 2(A+2)B-m_5-m_7, 2(A+1)B + m_5 + m_7 \} \cdot 4AW. 
  \end{align*}
  Optimizing over $m_5,m_7$ yields 
  \[ w_x \le (2A+3)B = \rho. \]
  Hence, in the current case $Z = \emptyset$ we have $w_{x,y} = 0$ and $w_x = \rho$, which yields
  \[ \WRNA(y_1\, x\, y_2) \le \rho + \WRNA(y_1 y_2) \le \rho + \max_{b \in [B]} \WRNA\Big(y_1 \concat \big(\bigconcat_{a \in [A]} x_{a,b} \big) \concat y_2\Big). \]
  
  Now consider the remaining case $|Z| \ge 1$.
  Write $Z = \{ x_{a_1,b_1},\ldots, x_{a_\ell,b_\ell} \}$, lexicographically sorted by $(a,b)$. Then we can bound $w_{x,y} \le \ell \cdot W$, since the total weight of each $x_{a,b}$ is bounded from above by $W$.
  
  In the following we bound $w_x$. Note that between $x_{a_i,b_i}$ and $x_{a_{i+1},b_{i+1}}$ the only symbols contributing to $w_x$ are $6$ and $\bar 6$. We count $(a_{i+1} - a_i)B$ 6's and $(a_{i+1}-a_i)B + b_{i+1} - b_i$ $\bar 6$'s in this substring. Hence, this contribution is bounded from above by 
  \[ \min\{(a_{i+1} - a_i)B, (a_{i+1}-a_i)B + b_{i+1} - b_i\} \cdot 8 AW  = \big(2(a_{i+1} - a_i)B + \min\{0, 2b_{i+1}-2b_i\}\big) \cdot 4AW. \] 
  Using the identity $\min\{0,2z\} = z - |z|$, we can rewrite this bound as
  \[ \big(2(a_{i+1} - a_i)B + b_{i+1}-b_i - |b_{i+1}-b_i|\big) \cdot 4AW. \] 
  
  We next analyze the contribution to $w_x$ before $x_{a_1,b_1}$. We count $(a_1-1)B+b_1$ $\bar 6$'s and $a_1 B$ 6's as well as $B$ 5's and $\bar 5$'s in this substring of $x$. Denote by $m_5$ the number of matched 5's, and note that this covers $m_5$ 6's from matching with $\bar 6$'s. Hence, we can match at most $\min\{(a_1-1)B+b_1, a_1 B - m_5 \}$ 6's. Summing up the weights, we obtain an upper bound on the contribution to $w_x$ before $x_{a_1,b_1}$ of 
  \[ m_5 \cdot 4 A W + \min\{(a_1-1)B+b_1, a_1 B - m_5 \} \cdot 8 A W = \min\{2(a_1-1)B+2b_1+m_5, 2a_1 B - m_5\} \cdot 4AW. \]
  Optimizing over $m_5$, we obtain an upper bound of $((2 a_1-1)B+b_1)\cdot 4AW$.
  
  Lastly, we analyze the contribution to $w_x$ after $x_{a_\ell,b_\ell}$. We count $(A-a_\ell+2)B-b_\ell$ $\bar 6$'s and $(A-a_\ell+2)B$ 6's as well as $B$ 7's and $\bar 7$'s. Similarly to the last paragraph, when matching $m_7$ 7's we obtain an upper bound on the contribution of
  \begin{align*} 
    &m_7 \cdot 4AW + \min\{(A-a_\ell+2)B - b_\ell, (A-a_\ell+2)B- m_7\}\cdot 8AW \\
    &= \min\{2(A-a_\ell+2)B - 2b_\ell + m_7, 2(A-a_\ell+2)B- m_7\} \cdot 4AW.
  \end{align*}
  Optimizing over $m_7$ yields an upper bound of $(2(A-a_\ell+2)B - b_\ell)\cdot 4AW$. 
  
  Summing over all three cases, we obtain an upper bound on $w_x$ of
  \[\Big( (2 a_1-1)B+b_1 + 2(A-a_\ell+2)B - b_\ell + \sum_{i=1}^{\ell-1} \big( 2 (a_{i+1}-a_i)B + b_{i+1} - b_i - |b_{i+1}-b_i| \big) \Big) \cdot 4AW. \]
  Note that all $a_i$'s and almost all $b_i$'s cancel as they form telescoping sums. What remains is 
  \[ w_x \le \Big(-B+2(A+2)B-\sum_{i=1}^\ell |b_{i+1} - b_i|\Big) \cdot 4AW = \rho - 4AW \sum_{i=1}^{\ell-1} |b_{i+1} - b_i|. \]
  In combination with the inequalities $w_{x,y} \le \ell W$ and $w_x + w_{x,y} \ge \rho$ shown above, we obtain
  \[ \sum_{i=1}^{\ell-1} |b_{i+1} - b_i| \le \frac{\ell}{4A}. \]
  Note that we have $|b_{i+1} - b_i| = 0$ for at most $A-1$ $i$'s, since $b_{i+1} = b_i$ implies $a_{i+1} > a_i$. This yields
  \[ \sum_{i=1}^{\ell-1} |b_{i+1} - b_i| \ge \ell - 1 - (A-1) = \ell - A. \]
  Together with the upper bound, we obtain $\ell - A \le \ell/(4A) \le \ell/2$, which yields $\ell \le 2A$. Hence, we have 
  \[ \sum_{i=1}^{\ell-1} |b_{i+1} - b_i| \le \frac{\ell}{4A} \le 1/2 < 1, \]
  which implies that $b_{i+1} = b_i$ for all $i$. Let $b := b_1 = \ldots = b_\ell$. Then $R$ matches only the strings $x_{a,b}$ for $a \in A$, among all strings in $X$. Since we showed $w_x \le \rho$, we indeed obtain 
  \[ \WRNA(T) \le \rho + \max_{b \in [B]} \WRNA\Big(y_1 \concat \big(\bigconcat_{a \in [A]} x_{a,b} \big) \concat y_2\Big). \qedhere \]
\end{proof}

\end{proof}

\section{Disjointness, Hamming Distance, and Subsequence} \label{sec:partial}

% !TEX root = main.tex

In this section we consider the following three problems on compressed sequences. In all problems we are given SLPs $\cT$ and $\cP$ of size $n$ and $m$, representing a text $T = \eval(\cT)$ of length $N$ and a pattern $P = \eval(\cP)$ of length $M$.

\begin{problem}[Disjointness]
	Given two compressed sequences $\cT$ and $\cP$ of equal decompressed lengths $N=M$ over alphabet $\{0,1\}$, decide whether there is a position such that both sequences have symbol $1$ at that position, i.e., whether $T[i]=P[i]=1$ holds for some $i$.
\end{problem}

\begin{problem}[Hamming Distance]
	Given two compressed sequences $\cT$ and $\cP$ of equal decompressed lengths $N=M$, output $\hamming(P,T)=|\{i|P[i] \neq T[i]\}|$. That is, output the number of positions where the decompressed sequences differ.
\end{problem}

\begin{problem}[Subsequence]
	Given two compressed sequences $\cT$ and $\cP$ of decompressed length $N \ge M$, decide whether the pattern sequence $P$ is a subsequence of the text sequence~$T$.
\end{problem}

We note that in the uncompressed setting all three problems have linear time trivial algorithms.
This immediately implies that all three problems can be solved in time $O(N)$ by decompressing the sequences and running the trivial algorithms. 
Below we show that this running time is not optimal and can be improved for all three problems for sufficiently compressible strings.
Furthermore, we show conditional lower bounds for the three problems assuming the Combinatorial $k$-Clique conjecture, $k$-SUM conjecture, and Strong $k$-SUM conjecture (see \secref{hardnessassumptions} for definitions). We were, however, not able to establish matching upper and lower bounds and we leave it as an open problem to close the gap.

\paragraph{Known Lower Bounds from Classic Complexity Theory} 
In~\cite{lifshits2007processing} it was shown that the Hamming Distance problem is \#{\sf P}-complete and thus a polynomial time $(nm)^{O(1)}$ algorithm for it is unlikely to exist. Lohrey~\cite{lohrey2011leaf} showed that the Subsequence problem is at least as hard as {\sf PP} and is contained in {\sf PSPACE}. It is conjectured that the subsequence problem is {\sf PSPACE}-complete~\cite{Lohrey12}. Note that the class {\sf PP} contains computationally very difficult problems. In particular, Toda's theorem states that the entire polynomial hierarchy PH is contained in ${\sf P}^{\sf PP}$.

We can easily check that the Disjointness problem is in {\sf NP}. A variant of our Theorem~\ref{lb1} below implies that the Subset Sum problem can be reduced to the Disjointness problem and thus Disjointness is in fact {\sf NP}-complete.

\subsection{Algorithms} \label{sec:subsequpper}

We start this section by showing a simple algorithm for the Subsequence problem that runs in time $O((n|\Sigma| + M) \log N)$ (see Theorem~\ref{alg_subsequence}). 
An algorithm with very similar guarantees was obtained in~\cite{bille2014compressed}.
Note that in a natural setting, namely when $|\Sigma|\leq O(1)$, $n\leq M$ and $N\leq M^{O(1)}$, the algorithm runs in time $\tOh(M)$. That is, we do not need to decompress the text sequence to be able to solve the  Subsequence problem.

In Theorems \ref{faster_hamming} and \ref{faster_subsequence} below we show $O\left(\max(m,n)^{1.5}\cdot N^{0.6}\right)$ time algorithms for the Hamming Distance and Subsequence problems, respectively. We observe that both running times that we obtain for the Subsequence problem are incomparable. 
Finally, by Theorem \ref{disj_subs} from Section~\ref{sec:subseqlower}, the Disjointness problem can be reduced to the Subsequence problem. This implies an $O(\max(m,n)^{1.5}\cdot N^{0.6})$ time algorithm for the Disjointness problem.
To the best of our knowledge these upper bounds are new. 

\begin{thm} \label{alg_subsequence}
	The Subsequence problem can be solved in time $O((n|\Sigma| + M) \log N)$.
\end{thm}
\begin{proof}
	We start by decompressing the pattern sequence $\cP$ in $O(M)$ time.
	To decide whether $P$ is a subsequence of the text sequence $T$, for $i=1, \ldots, M$ (in this order) we will find the smallest $j\in \{1, \ldots, N\}$ such that $P[1..i]$ (the prefix of the decompressed pattern of length $i$) is a subsequence of $T[1..j]$. In the rest of the proof we will describe how to do this efficiently.
	
	We start by transforming the compressed text $\cT$ into an AVL-grammar of size $O(n\log N)$ and depth $O(\log N)$ according to Theorem~\ref{AVL}. This takes $O(n\log N)$ time. Additionally, for every alphabet symbol $\sigma\in \Sigma$ and every non-terminal $T_i$ of the AVL-grammar, we decide whether the sequence produced by the non-terminal $T_i$ contains the symbol $\sigma$. For every symbol, this can be done in $O(n\log N)$ time. Since the size of the alphabet is $|\Sigma|$, this takes $O(n|\Sigma|\log N)$ total time.
	
	Given an index $i=1, \ldots, M$, suppose that we know the smallest index $j\in \{1, \ldots, N\}$ such that $P[1..i]$ is a subsequence of $T[1..j]$. 
	We will show how to find the smallest $j'>j$ such that $P[1..i+1]$ is a subsequence of $T[1..j']$.
	The required running time will follow since we will be able to do this in $O(\log N)$ time for every index $i$.
	We find the smallest $j'>j$ in two steps. In the first step we traverse the parse tree bottom-up from the symbol $T[j]$ until the current node has $T[j]$ in the left subtree and the right subtree contains symbol $P[i+1]$. In the second step we go to the right subtree and then keep going to the left-most child that contains the symbol $P[i+1]$. Since the height of the parse tree is $O(\log N)$, this takes $O(\log N)$ time. This finishes the description of the algorithm. Note that we did not decompress the text sequence $T$ in this process.
\end{proof}

\begin{thm} \label{faster_hamming}
	The Hamming Distance problem can be solved in time
	$$
		\tOh\left(\max(m,n)^{2 - 1/\log_2(2\varphi)}\cdot N^{1/\log_2(2\varphi)}\right)=\tOh\left(\max(m,n)^{1.409\ldots}\cdot N^{0.592\ldots}\right),
	$$
	where $\varphi=\frac{1+\sqrt{5}}{2}$ is the golden ratio.
\end{thm}
\begin{proof}
	Let $P_1,P_2,\ldots,P_m$ be the SLP $\cP$ corresponding to the decompressed pattern sequence $P$
	and let $T_1,T_2,\ldots,T_n$ be the SLP $\cT$ corresponding to the decompressed text sequence $T$.
	We assume that the decompressed length of the sequences $P$ and $T$ is $|P|=|T|=N$.

	By Theorem \ref{AVL} we can assume that $P_1,P_2,\ldots,P_m$ and $T_1,T_2,\ldots,T_n$ are AVL-grammars. This increases the running time by a factor of at most $\poly \log N$, which is hidden in the $\tOh(\cdot)$ notation. Fix an $i=1, \ldots, m$ and consider the sequence $\eval(P_i)$ with the corresponding parse tree of height $\depth(P_i)$. Then one can verify that the length of the sequence is bounded from above by $|\eval(P_i)|\leq 2^{\depth(P_i)}$ and from below by 
	\begin{equation} \label{balanced}
		|\eval(P_i)|\geq F_{\depth(P_i)}\geq \Omega\left(\varphi^{\depth(P_i)}\right),
	\end{equation}
	where $F_{\depth(P_i)}$ is the $\depth(P_i)$-th Fibonacci number and $\varphi$ is the golden ratio~\cite{Cormen:2001:IA:580470}. Analogous properties hold for $T_j$ for $j=1, \ldots, n$.

	For every $P_i$ and $T_j$ we precompute the length of $\eval(P_i)$ and $\eval(T_j)$, respectively.
	We define the function 
	$$
		\ham(P_i, T_j, d):=\sum_{r: \ \substack{r \in \{1,\ldots,|\eval(P_i)|\},\\r+d \in \{1,\ldots,|\eval(T_j)|\}}} \big[\eval(P_i)_r \neq \eval(T_j)_{r+d}\big],
		%:=\hamming(\eval(\hat p_i),\text{suffix}(\eval(\hat t_j),d))
	$$
	where $d$ is a negative or a non-negative integer.
	In other words, $\ham(P_i, T_j, d)$ is equal to the Hamming distance between $T_j$ and a shifted $P_i$ (by $d$ symbols to the right if $d>0$ and by $|d|$ symbols to the left otherwise), where we consider only the symbols that have aligned counterparts.
	Clearly, we can solve the Hamming Distance problem by outputting $\hamming(P,T) = \ham(P_m, T_n, 0)$.

	A simple algorithm for computing the Hamming distance is the following recursive method.
	Assume that the sequence $\eval(P_i)$ is longer than the sequence $\eval(T_j)$, and $P_i \to P_{\ell(i)}, P_{r(i)}$. Then
	$$
		\ham(P_i, T_j, d) = \ham(P_{\ell(i)}, T_j, d) + \ham(P_{r(i)}, T_j, d + |\eval(P_{\ell(i)})|).
	$$
	Otherwise, if the sequence $\eval(T_j)$ is longer and $T_j \to T_{\ell'(j)}, T_{r'(j)}$, then
	$$
		\ham(P_i, T_j, d) = \ham(P_i, T_{\ell'(j)}, d) + \ham(P_i, T_{r'(j)}, d - |\eval(T_{\ell'(j)})|).
	$$

	Clearly, for any recursive subproblem where the argument $d$ is such that no
	symbols get aligned, we can immediately return 0. When $P_i$ or $T_j$
	encode a single symbol, we compute their Hamming distance in a constant time.

	We use this recursive algorithm with memoization, i.e., if we call the same inputs twice, then we return the stored answer.

	\paragraph{Running Time} We crucially use the fact that we split
	the longer text in each step, and property~\ref{balanced}. Both
	together imply that 
	$$
		|\eval(T_j)| \geq \Omega\left(|\eval(P_i)|^{\log_2 \varphi}\right) = \Omega\left(|\eval(P_i)|^{0.694\ldots}\right)
	$$
	for each call $\ham(P_i,T_j,d)$.
	We bound the running time by counting for each $P_i$ how many different
	calls there are of the form $\ham(P_i,T_j,d)$ with $|\eval(T_j)| \leq |\eval(P_i)|$.
	The running time corresponding to the calls with $|\eval(T_j)| > |\eval(P_i)|$ can be analyzed analogously.
	Note that $|\eval(T_j)| \leq |\eval(P_i)|$ implies $|d| \leq O(|\eval(P_i)|)$, as larger shifts
	immediately give answer 0. Let $0 < \alpha < 1$ to be fixed later. If $|\eval(P_i)| < N^\alpha$ we can thus bound the
	contribution of $P_i$ to the running time by $n N^\alpha$ (there are $n$ $T_j$'s and
	$N^\alpha$ possible offsets $d$). Otherwise, if $|\eval(P_i)| \ge N^\alpha$, then
	$|\eval(T_j)| \ge N^{\alpha\log_2 \varphi}$, and thus there are at most $N^{1 - \alpha\log_2 \varphi}$
	calls to such $T_j$ in the parse tree for $T=\eval(T_n)$. Thus, there are at
	most this many calls $\ham(P_i,T_j,d)$, so the contribution of $P_i$ to the
	running time is at most $N^{1 - \alpha\log_2 \varphi}$. Summed over all $m$ different $P_i$'s the total
	running time is bounded by
	$O\left(m (n N^\alpha + N^{1 - \alpha\log_2 \varphi})\right)$.
	Minimizing over $\alpha$ gives the running time $O\left(m\cdot n^{1 - 1/\log_2(2\varphi)}\cdot N^{1/\log_2(2\varphi)}\right)$. 
	The running time corresponding to the calls with $|\eval(T_j)| > |\eval(P_i)|$ can be similarly  bounded by $O\left(n\cdot m^{1 - 1/\log_2(2\varphi)}\cdot N^{1/\log_2(2\varphi)}\right)$. 
	It remains to observe that the total running time is bounded by $O\left(\max(m,n)^{2 - 1/\log_2(2\varphi)}\cdot N^{1/\log_2(2\varphi)}\right)$ as required.
\end{proof}

\begin{thm} \label{faster_subsequence}
	The Subsequence problem can be solved in time
	$$
		\tOh\left(\max(m,n)^{2 - 1/\log_2(2\varphi)}\cdot N^{1/\log_2(2\varphi)}\right)=\tOh\left(\max(m,n)^{1.409\ldots}\cdot N^{0.592\ldots}\right),
	$$
	where $\varphi=\frac{1+\sqrt{5}}{2}$ is the golden ratio.
\end{thm}
\begin{proof}
	The algorithm follows a similar recursive method as in Theorem~\ref{faster_hamming}. As above, we assume that the SLPs $\cP$ and $\cT$ are AVL-grammars.
	
	For non-terminals $P_i$ and $T_j$ and an integer $d$ we define the function $\subseq(P_i,T_j,d)$. If $d\geq 0$, then we assume that we already matched a prefix of $\eval(P_i)$ of length $d$ (the prefix is a subsequence of an earlier part of the text) and our goal is to match the rest of $\eval(P_i)$ with $\eval(T_j)$. On the other hand, if $d<0$, then we assume that we already matched a prefix of $\eval(T_j)$ of length $|d|$ (a previous part of the pattern is a subsequence of the prefix) and our goal is to match $\eval(P_i)$ to the rest of $\eval(T_j)$. The function returns an integer as follows.
	Let $d'$ be the length of the longest prefix of $\eval(P_i)$ that can be matched to $\eval(T_j)$. (If $d\geq 0$, then we match only the remainder of $\eval(P_i)$ to $\eval(T_j)$. If $d<0$, then we match $\eval(P_i)$ to the remainder of $\eval(T_j)$.) If $d'<|\eval(P_i)|$, that is, we cannot match entire $\eval(P_i)$ to $\eval(T_j)$, then the function $\subseq(P_i,T_j,d)$ returns $d'$. Otherwise, if $d'=|\eval(P_i)|$, the function returns the length of the shortest prefix of (the remainder of) $\eval(T_j)$ that can be matched to (the remainder of) $\eval(P_i)$.
	
	Given the description of the function, the recursive implementation of it is straightforward and is described below.
	To evaluate $\subseq(P_i,T_j,d)$, we consider three cases.
	
	\paragraph{Case 1} $P_i$ or $T_j$ represents a single symbol. The problem is trivial to solve in this case.
	
	\paragraph{Case 2} $|\eval(P_i)|\geq |\eval(T_j)|$. Let $P_i \to P_{\ell(i)},P_{r(i)}$ be the SLP rule corresponding to $P_i$. If $d\geq |\eval(P_{\ell(i)})|$, then the function returns $\subseq(P_{r(i)},T_j,d-|\eval(P_{\ell(i)})|)$, which we compute recursively. If, on the other hand, $d<|\eval(P_{\ell(i)})|$, we recursively compute $d':=\subseq(P_{\ell(i)},T_j,d)$ and return $d'$ if $d'\geq 0$ or return $\subseq(P_{r(i)},T_j,d')$ if $d'<0$.
	
	\paragraph{Case 3} $|\eval(P_i)|<|\eval(T_j)|$. This case is similar to the previous one. Let $T_j \to T_{\ell'(j)},T_{r'(j)}$ be the SLP rule corresponding to $T_j$. If $-d>|\eval(T_{\ell'(j)})|$, we return $\subseq(P_{i},T_{r'(j)},-d-|\eval(T_{\ell'(j)})|)$, which we compute recursively. Otherwise, we define $d':=\subseq(P_i,T_{\ell'(j)},d)$ and return $d'$ if $d'< 0$ or return $\subseq(P_{i},T_{r'(j)},d')$ if $d'\geq 0$.
	
	The correctness of the algorithm follows from the description and the definition of the function $\subseq(P_i,T_j,d)$.
	The running time analysis is similar to Theorem~\ref{faster_hamming} and we omit it.
\end{proof}

% !TEX root = main.tex

\subsection{Lower Bounds} \label{sec:subseqlower}

In this section we show conditional lower bounds for the Disjointness, Hamming Distance and Subsequence problems.
First, we show that the Disjointness problem can be reduced to the Subsequence problem (Theorem~\ref{disj_subs}) and to the Hamming Distance problem (Theorem~\ref{disj_hamm}). Thus, any algorithmic improvement for the latter two problems implies a faster algorithm for the Disjointness problem. Alternatively, we can think about the Disjointness problem as the core hard problem explaining hardness for the two other problems. Second, we show a matching $N^{1-o(1)}$ lower bound for combinatorial algorithms for the Subsequence problem in the setting where $N \approx M \approx n^2 \approx m^2$. We use the combinatorial $k$-Clique conjecture to establish this hardness. Finally, we use the $k$-SUM conjecture (Conjecture~\ref{conj1}) for all three aforementioned problems. The lower bounds that we show are not tight. We show that assuming a stronger version of the $k$-SUM conjecture (Conjecture~\ref{conj2}) allows us to get higher lower bounds, but still not matching.

\begin{thm} \label{disj_subs}
	The Disjointness problem can be reduced to the Subsequence problem. The reduction loses at most constant factors in the length of compressed and decompressed sequences.
\end{thm}
\begin{proof}
	Let $P$ and $T$ be two binary sequences, forming an instance of the Disjointness problem.
	We construct a sequence $P'$ from $P$ by replacing every symbol $0$ with symbol ``0'' and every symbol $1$ with two symbols ``10''. 
	Similarly, we construct a sequence $T'$ from $T$ by replacing every symbol $0$ with two symbols ``10'' and every symbol $1$ with ``0''.
	
	The resulting sequences $P'$ and $T'$ are compressible similarly as $P$ and $T$. We can check that $P'$ is a subsequence of $T'$ if and only if we have $P[i]=0$ or $T[i]=0$ for all $i$. This completes the reduction.
\end{proof}

\begin{thm} \label{disj_hamm}
	The Disjointness problem can be reduced to the Hamming Distance problem. The reduction loses at most constant factors in the length of compressed and decompressed sequences.
\end{thm}
\begin{proof}
	Let $P$ and $T$ be two binary sequences, forming an instance of the Disjointness problem.
	We construct a sequence $P'$ from $P$ by replacing every symbol $0$ with three symbols ``011'' and every symbol $1$ with three symbols ``000''. 
	Similarly, we construct a sequence $T'$ from $T$ by replacing every symbol $0$ with ``001'' and every symbol $1$ with ``111''.
	
	These four gadget sequences have Hamming distance $1$ for all pairs except when both original symbols are $1$. In this case the Hamming distance between the two gadgets is $3$.
	We conclude that $\hamming(P',T')>N=|P|=|T|$ if and only if there exists $i$ with $P[i]=T[i]=1$. This concludes the reduction.
\end{proof}

\iffalse
\begin{lem}
  Reduction from the compressed disjointness problem to the subsequence problem. Disjointness can also be reduced to Hamming Distance. 
\end{lem}
\karl{do we want to state some conjecture on disjointness?}
\begin{proof}[Proof Sketch]
$a$ and $b$ are two binary vectors of equal length.
We construct sequence $a'$ from $a$ by replacing every 0 with ``\$'' and every 1 with ``x\$'' and concatenating the short sequences.
Similarly, we construct $b'$ from $b$ by replacing every 0 with ``x\$'' and every 1 with ``\$''.
$a'$ is a subsequence of $b'$ if and only if we have $a[i]=0$ or $b[i]=0$ for all $i$.

Second claim: Just replace every
\begin{itemize}
\item 0 in $a$ by 011,
\item 1 in $a$ by 000,
\item 0 in $b$ by 011,
\item 1 in $b$ by 111.
\end{itemize}
These gadgets have Hamming distance 1 for all pairs except (1,1), where
it is 3. So the final Hamming distance is $>N$ iff $a,b$ are not disjoint.
\end{proof}
\fi

\begin{thm} \label{thm:subseqlower}
	The Subsequence problem has no combinatorial $O(N^{1-\eps})$ time algorithm for any $\eps > 0$ in the setting $N = \Theta(M) = \Theta(n^2) = \Theta(m^2)$ and $|\Sigma| = O(N^\eps)$, assuming the combinatorial $k$-Clique conjecture.
\end{thm}
\begin{proof}
	The reduction will rule out combinatorial algorithms with running time $N^{1-\eps}$ by using the Combinatorial $k$-Clique conjecture~\ref{conj:combclique} with $k=O(1/\eps)$. 
	Let $k \ge 4$ be even, and let $G=(V,E)$ be an instance of $k$-Clique. 
	%We can assume that $G$ is $k$-partite with $V$ vertices in each part.
	In the following we will construct an equivalent instance of the Subsequence problem, i.e., a text $T = \eval(\cT)$ and a pattern $P = \eval(\cP)$, satisfying $N = |T| = O(V^{k+1})$, $n = |\cT| = O(V^{(k/2)+1})$, $M = |P| = O(V^{k})$, and $m = |\cP| = O(V^{k/2})$. The alphabet size will be $|\Sigma| = O(V)$. By a simple padding\footnote{Specifically, let $\natural$ be a fresh symbol and add $\natural^{V^{k+2}}$ as a prefix to $T$ and $P$. Compress this string $\natural^{V^{k+2}}$ to length $V^{k/2+1}$ by writing it as $(\natural^{V^{k/2+1}})^{V^{k/2+1}}$ and using Observation~\ref{obs:repetition}.}, we can then ensure that $N,M = \Theta(V^{k+2})$ and $n,m = \Theta(V^{k/2+1})$, so that indeed $N = \Theta(M) = \Theta(n^2) = \Theta(m^2)$, and we have $|\Sigma| = O(N^\eps)$ for any $k \ge 1/\eps$. Finally, a combinatorial $O(N^{1-\eps})$ algorithm for the Subsequence problem in this setting would yield a combinatorial algorithm for $k$-Clique in time $O(V^{(k+2)(1-\eps)}) = O(V^{k(1-\eps/2)})$ for any $k \ge 4/\eps$, contradicting the combinatorial $k$-Clique conjecture.
	
	We first construct clique gadgets and then the pattern and the text.
	The alphabet will be $\Sigma = V \cup \{\#,\$\}$.
	
	\paragraph{Construction of the clique gadgets \boldmath$CG$}
	Given a $(k/2)$-clique $C=\{v_1, ..., v_{k/2}\}$, we construct the clique gadget $CG(C)$ as:
	$$
		CG(C) = (v_1 v_2 ... v_{k/2} \#)^{k/2}.
	$$
	That is, we write down the labels of the vertices (in increasing order), put ``\#'' at the end and repeat the resulting sequence $k/2$ times.

	\paragraph{Construction of the clique gadgets \boldmath$CG'$}
	Given a $(k/2)$-clique $C'=\{u_1, ..., u_{k/2}\}$, we construct $CG'(C')$ as:
	$$
		CG'(C') =
			\text{Neighbors}(u_1) \#
			\text{Neighbors}(u_2) \#
			...  \#
			\text{Neighbors}(u_{k/2}) \#
	$$
	where $\text{Neighbors}(u)$ lists all neighbors of vertex $u$ in increasing order.
	
	We can check that for any $(k/2)$-cliques $C,C'$, $CG(C)$ is a subsequence of $CG'(C')$ if and only if $C \cup C'$ forms a $k$-clique.
	
	\paragraph{Construction of the sequence \boldmath$Z$} We construct $Z$ as:
	$$
		Z:=(L\#)^{k/2},
	$$
	where $L$ is the sequence containing all $V$ vertices in the graph in increasing order. We can verify the any clique gadget $CG(C)$ is a subsequence of $Z$.
	
	\paragraph{Construction of the Pattern}
	The pattern consists of clique gadgets as follows. 
	Enumerate all $(k/2)$-cliques $C_1,\ldots,C_Q$ with $Q \le V^{k/2}$ in $G$. 
	The pattern sequence $P$ is constructed as: 
	$$
		P := \left(CG(C_1) \$ CG(C_2) \$ ... \$ CG(C_Q) \$\right)^Q.
	$$
	That is, we concatenate the $Q$ clique gadgets $CG(C_1), \ldots, CG(C_Q)$ in one sequence and put ``\$'' after every gadget, and repeat the resulting sequence $Q$ times. Note that the symbol ``\$'' does not appear in any clique gadget.
	
	\paragraph{Construction of the Text}
	The text is somewhat similar to the pattern, defined by:
	$$
		T := 
			(CG'(C_1) \$ Z \$)^Q \,
			(CG'(C_2) \$ Z \$)^Q \,
			...
			(CG'(C_{Q-1}) \$ Z \$)^Q \,
			(CG'(C_Q) \$ Z \$)^{Q-1} \,
			CG'(C_Q) \$.
	$$

	\paragraph{Correctness} The pattern consists of $Q^2$ clique gadgets with the symbol $\$$ in between any two of them. The text consist of $Q^2$ cliques gadgets with the sequence $\$Z\$$ in between any two of them. Since there are only $Q^2-1$ $Z$'s in the text, we cannot match all clique gadgets of the pattern to $Z$'s in the text. Hence, if $P$ is a subsequence of $T$, then at least one clique gadget $CG(C_i)$ is a subsequence of $CG'(C_j)$ for some $i,j$. This happens only if $C_i \cup C_j$ form as $k$-clique in $G$. 
	
	For the other direction, we show that if $G$ contains a $k$-clique, so that there are $i,j$ with $C_i \cup C_j$ forming a $k$-clique, implying that $CG(C_i)$ is a subsequence of $CG'(C_j)$, then the pattern is a subsequence of the text. Indeed, let $q=j\cdot Q+i$. The $q$-th clique gadget in the pattern is $CG(C_i)$ and the $q$-th clique gadget in the text is $CG'(C_j')$. We match all clique gadgets before the $q$-th one as well as after the $q$-th one to $Z$'s, and we match  $CG(C_i)$ to $CG'(C_j')$. This shows that $P$ is a subsequence of $T$.
    
    Since $|L| = V$, $Q \le V^{k/2}$, and $k$ is a constant, the length bounds $N = O(V^{k+1})$ and $M = O(V^{k})$ are immediate. Using Observation~\ref{obs:repetition} to compress strings of the form $X^Q$ to size $O(|X| + \log Q)$, we also immediately obtain $n = O(V^{k/2+1})$ and $m=O(V^{k/2})$.
    This finishes the proof.
\end{proof}

\begin{thm} \label{lb1}
	Let $k\geq 1$ be an integer. Consider the Disjointness problem with $N=M=\Theta(n^{4k+1})=\Theta(m^{4k+1})$. Solving the Disjointness problem in this setting requires $N^{\frac{1}{4}+\frac{3}{16k+4}-o(1)}$ time assuming the $(2k+1)$-SUM conjecture.
\end{thm}

\begin{thm} \label{lb2}
	Let $k\geq 1$ be an integer. Consider the Disjointness problem with $N=M=\Theta(n^{3k+1})=\Theta(m^{3k+1})$. Solving the Disjointness problem in this setting requires $N^{\frac{1}{3}+\frac{2}{9k+3}-o(1)}$ time assuming the Strong $(2k+1)$-SUM conjecture.
\end{thm}

%\karl{Strong variants of the conjecture give better lower bounds. Extend to Hamming distance if possible.}

By Theorems~\ref{disj_subs} and \ref{disj_hamm}, the same kind of hardness holds for the Subsequence and Hamming Distance problems.

\begin{proof}[Proof of Theorems \ref{lb1} and \ref{lb2}]
	Let $k \geq 1$ be an integer and let $A\subseteq \{0,1, \ldots, R-1, R\}$ be an instance of the $(2k+1)$-SUM problem with $|A|=r$ and target sum $t$.
	Without loss of generality, $R$ is divisible by $k+1$ and $t$ is divisible by $k$. We define the set $B:=\left\{\frac{t}{k}+R-a \mid a \in A\right\}$ and the set $C:=\left\{\frac{Rk}{k+1}+a \mid a \in A\right\}$. We can verify that there exist $b_1, \ldots, b_k \in B$ and $c_1, \ldots, c_{k+1} \in C$ with $b_1 + \ldots + b_{k}=c_1 + \ldots + c_{k+1}$ if and only if there exist $a_1, \ldots, a_{2k+1} \in A$ with $a_1 + \ldots + a_{2k+1}=t$. We note that $B,C \subseteq \{1, 2, \ldots, R'\}$ for $R':=2R$.
	
	In $O(r \log r)$ time we will construct an instance to the Disjointness problem with the following properties.
	\begin{itemize}
		\item Pattern $P = \eval(\cP)$ is constructed from the set $B$ and has length $M = R'\cdot r^{2k}$ and compressed size $m= O(r \log r)$,
		\item Text $T = \eval(\cT)$ is constructed from the set $C$ and has length $N = R'\cdot r^{2k}$ and compressed size $n= O(r \log r)$,
		\item There exists $i$ such that $P[i]=T[i]=1$ if and only if there exist $b_1, \ldots, b_k \in B$ and $c_1, \ldots, c_{k+1} \in C$ with $b_1 + \ldots + b_{k}=c_1 + \ldots + c_{k+1}$.
	\end{itemize}
	Simply padding allows us to increase the text length and pattern length to $R' r^{2k} \log^{k'} r$ for any $k' \ge 0$, and to achieve $n,m = \Theta(r \log r)$. Setting $R=r^{2k+1}$, we thus have $N = M = 2 r^{4k+1} \log^{4k+1} r = \Theta(n^{4k+1}) = \Theta(m^{4k+1})$. Any $O(N^{1/4 + 3/(16k+4) - \eps}) = O(N^{(k+1-\eps)/(4k+1)})$ time algorithm for Disjointness would now imply an algorithm for $(2k+1)$-SUM in time $O((r \log r)^{k+1-\eps}) = O(r^{k+1-\eps/2})$, contradicting the $(2k+1)$-SUM conjecture (Conjecture~\ref{conj1}). This proves Theorem~\ref{lb1}.
	Similarly, setting $R=r^{k+1}$ and using the Strong $(2k+1)$-SUM conjecture (Conjecture~\ref{conj2}) we obtain Theorem~\ref{lb2}.
	
	In the remainder of the proof we present the promised construction.
	
	Without loss of generality, we have $R' > 10k \cdot \max (B\cup C)$.
	
	\paragraph{Construction of the Pattern} 
	We define the pattern as
	\[ P := \Big( \bigconcat_{b_1,\ldots,b_k \in B} 0^{b_1+\ldots+b_k}\, 1\, 0^{R'-(b_1+\ldots+b_k)-1} \Big)^{r^k}, \]
	where the $\bigconcat$ goes over all tuples $(b_1,\ldots,b_k) \in B^k$ in lexicographic order.
	That is, P consists of $r^k$ repetitions of a sequence $Z$ of length $R'\cdot r^k$. The sequence $Z$ consists of sequences $Z_1,\ldots,Z_{r^k}$, corresponding to $k$-tuples $(b_1,\ldots,b_k) \in B^k$. Each sequence $Z_i$ has length $R'$, and the sequence $Z_i$ corresponding to tuples $(b_1,\ldots,b_k)$ has 0's everywhere except at position $b_1+\ldots+b_k+1$.
	
	\paragraph{Construction of the Text}
	We define the text as 
	\begin{equation} \label{eq:ksumtext}
		T := \bigconcat_{c_1,\ldots,c_k \in C} \Big( Y(c_1,\ldots,c_k) \Big)^{r^k},
	\end{equation}
	where $Y(c_1,\ldots,c_k)$ is a string of length $R'$ with $Y(c_1,\ldots,c_k)[j+1] = 1$ if $j \in \{c+c_1+\ldots+c_k \mid c\in C\}$, and $Y(c_1,\ldots,c_k)[j+1] = 0$ otherwise.
	
    \paragraph{Analysis}
    Note that there is an index $i$ with $P[i]=T[i]=1$ if and only if there exist $b_1, \ldots, b_k \in B$ and $c_1, \ldots, c_{k+1} \in C$ with $b_1 + \ldots + b_{k}=c_1 + \ldots + c_{k+1}$. Hence, correctness of the reduction can be easily verified.
    The length $N = M = R' r^{2k}$ is immediate. It remains to show that the pattern and the text are compressible.
	
	\paragraph{Compressing the Pattern} 
	Since $P = Z^{r^k}$, by Observation~\ref{obs:repetition} it suffices to compress $Z$. 
	We construct the sequence $Z$ inductively. We write $B = \{B_1,\ldots,B_r\}$.
	We define $S_0\to 1$ to be a non-terminal generating a sequence of length $1$ containing a single symbol $1$.
	For $i \in [k]$ we define the non-terminal $S_i$ as follows:
	\begin{equation} \label{induction1}
		S_i\to\left( \bigcirc_{w=1}^{r-1} 0^{B_w} S_{i-1} 0^{R'r^{i-1}-B_w-|\eval(S_{i-1})|} \right)\circ 0^{B_r}S_{i-1}.
	\end{equation}
	
	Finally, we set $S\to S_{k} \circ 0^{R'r^k-|S_k|}$.
	Here the right hand side contains more than two SLP non-terminals, but using Observation~\ref{obs:repetition} it is easy to convert this into a proper SLP of size $O(r \log r)$ as required.
	It remains to check that $Z = \eval(S)$, i.e., $\eval(S)=\bigconcat_{b_1,\ldots,b_k \in B} 0^{b_1+\ldots+b_k}\, 1\, 0^{R'-(b_1+\ldots+b_k)-1}$. Indeed, a straightforward induction shows that we constructed $S_i$, $i \in [k]$ such that 
	$$
		\eval(S_i)\circ 0^{Rr^{i}-|\eval(S_i)|}=\bigconcat_{b_1,\ldots,b_i \in B} 0^{b_1+\ldots+b_i}\, 1\, 0^{R'-(b_1+\ldots+b_i)-1}.
	$$
	The induction step is performed by using the derivation rule~\eqref{induction1}.
	%where $(b_1',\ldots,b_i')\in B^i$ is the last $i$-tuple in the lexicographic order.
	
	\paragraph{Compressing the Text}
	Let $W$ be a string of length $R'$ consisting only of $0$'s except $W[j+1]=1$ for any $j \in C$.
	We define an SLP non-terminal $Y'$ that generates the shortest prefix of $W$ containing all $1$'s of $W$. We set 
	$$
		Y_0\to \left(Y'0^{R'-|\eval(Y')|}\right)^{r^k-1}Y'.
	$$
	Note that $\eval(Y') 0^{R'-|\eval(Y')|} = W$. Hence, $Y_0$ generates the string $W^{r^k}$ where we removed the longest suffix of $0$'s.
	We write $C = \{C_1,\ldots,C_r\}$.
	
	For $i=1,\ldots,k$ we define sequence $Y_i$ as follows:
	\begin{equation} \label{induction2}
		Y_i\to\left( \bigcirc_{w=1}^{r-1} 0^{C_w} Y_{i-1} 0^{R'r^{k+i-1}-C_w-|\eval(Y_{i-1})|} \right)\circ 0^{C_r}Y_{i-1}.
	\end{equation}
	
	Finally, we set $\cT\to Y_{k} \circ 0^{R'r^{2k}-|\eval(Y_k)|}$. It is easy to verify that the size of the above SLP $\cT$ is $O(r \log r)$. It remains to show that $\eval(\cT)=T$ as in~\eqref{eq:ksumtext}. That is, we want to show that $\eval(\cT)=\bigconcat_{c_1,\ldots,c_k \in C} \Big( Y(c_1,\ldots,c_k) \Big)^{r^k}$. This follows by a straightforward induction. We can check that for $i=0,1,\ldots,k$ we have
	$$
		\eval(Y_i)\circ 0^{R'r^{k+i}-|\eval(Y_i)|}=\bigconcat_{c_1,\ldots,c_i \in C} \Big( Y(c_1,\ldots,c_i) \Big)^{r^k}.
	$$
	The induction step is performed by using the derivation rule~\eqref{induction2}.

\end{proof}

% !TEX root = main.tex

\section{Conclusion} \label{sec:conclusion}

With this paper we started the fine-grained complexity of analyzing compressed data, thus providing lower bound tools for a practically highly relevant area. We focused on the most basic problems on strings, leaving many other stringology problems for future work. 
Besides strings, there is a large literature on grammar-compressed other forms of data, e.g.\ graphs. It would be interesting to apply our framework and classify the important problems in these contexts as well.

Specifically, we leave the following open problems.
\begin{itemize}
\item Determine the optimal running time for the Disjointness, Hamming Distance, and Subsequence problems.

\item Generalize our lower bound for LCS to Edit Distance.

\item For NFA Acceptance we obtained tight bounds in case of a potentially dense automaton with $q$ states and up to $O(q^2)$ transitions. Prove tight bounds for the case of sparse automata with $O(q)$ transitions.

\item For large (i.e.\ superconstant) alphabet size, some bounds given in this paper are not tight, most prominently for Generalized Pattern Matching, Substring Hamming Distance, and Pattern Matching with Wildcards. Determine the optimal running time in this case.

\item For all lower bounds presented in this paper, check whether they can be improved to work for binary strings.
\end{itemize}

\medskip
\section*{Acknowledgements} 
This paper would not have been possible without Oren Weimann and \emph{Schloss Dagstuhl}.
Inspired by a Dagstuhl seminar on Compressed Pattern Matching in October, and while attending a Dagstuhl seminar on Fine-Grained Complexity in November, Oren asked in the open problems session whether SETH can explain the lack of $O((nN)^{1-\eps})$ algorithms for problems like LCS on compressed strings.
Later, in January, three of the authors of this paper attended a Dagstuhl seminar on Parameterized Complexity and made key progress towards the results of this work. Part of the work was also performed while visiting the Simons Institute for the Theory of Computing, Berkeley, CA. 
We thank Pawe{\l} Gawrychowski for helpful comments.

A.A. was supported by Virginia Vassilevska Williams' NSF Grants CCF-1417238 and CCF-1514339, and BSF Grant BSF:2012338. Arturs Backurs was supported by an IBM PhD Fellowship, the NSF and the Simons Foundation. While performing part of this work, M. K\"unnemann was affiliated with University of California, San Diego.

\bibliographystyle{abbrv}
\bibliography{ref}

\end{document}